\documentclass[format=acmsmall, review=false]{acmart}
\renewcommand\footnotetextcopyrightpermission[1]{} %
\usepackage[utf8]{inputenc}
\usepackage{hyphenat}
\usepackage{xspace}
\usepackage{amsmath}
\usepackage{mathtools,amssymb,mathrsfs}
\usepackage{mathrsfs}
\usepackage{stmaryrd}

\usepackage{paralist} 
\usepackage[boxed,ruled,vlined]{algorithm2e} 
\usepackage{enumerate}
\usepackage{enumitem} 
\usepackage{mdframed}
\usepackage{booktabs} %
\usepackage{url}
\usepackage{subcaption}

\usepackage{multirow}
\usepackage{color}

\usepackage{tikz} 

\newif\ifdraft
\draftfalse
\usepackage{savesym}
\savesymbol{comment}

\usepackage{array}
\usepackage{macros}

\usetikzlibrary{positioning} 
\usetikzlibrary{shapes,arrows,calc,matrix, decorations}

\newenvironment{myintro}%
  {\list{}{\leftmargin=0.1in\rightmargin=0.1in}\item[]}%
  {\endlist}

\setcopyright{none}
\title[Complexity Analysis of Generalized and Fractional Hypertree Decompositions]{Complexity Analysis of Generalized and Fractional Hypertree Decompositions}
\titlenote{This is a significantly extended and enhanced version of a paper presented at PODS 2018~\cite{DBLP:conf/pods/FischlGP18}.}

\author{Georg Gottlob}
\affiliation{%
  \institution{University of Oxford}
}
\email{georg.gottlob@cs.ox.ac.uk}

\author{Matthias Lanzinger}
\affiliation{%
  \institution{TU Wien}
}
\email{mlanzing@dbai.tuwien.ac.at}

\author{Reinhard Pichler}
\affiliation{%
  \institution{TU Wien}
}
\email{reinhard.pichler@tuwien.ac.at}

\author{Igor Razgon}
\affiliation{%
  \institution{Birkbeck University of London}
}
\email{igor@dcs.bbk.ac.uk}


\setcopyright{none}

\begin{document}

\begin{abstract}
\normalsize 
Hypertree decompositions (HDs), as well as the more powerful generalized hypertree 
decompositions (GHDs), and the yet more general fractional hypertree 
decompositions (FHDs) are hypergraph decomposition methods successfully used for 
answering conjunctive queries
and
for solving %
constraint satisfaction problems.
Every hypergraph $H$ has a width relative to each of these %
methods:   
its hypertree width $\hw(H)$, its generalized hypertree width $\ghw(H)$, and 
its  
fractional hypertree width $\fhw(H)$, respectively.
It is known that  $\hw(H)\leq k$ can be checked in polynomial time for fixed $k$, while 
checking $\ghw(H)\leq k$ is NP-complete for $k \geq 3$. 
The complexity of checking  $\fhw(H)\leq k$  for a fixed $k$ has been open for over a decade.

We settle this open problem by showing that  
checking  $\fhw(H)\leq k$ is NP-complete, even for $k=2$.
The same construction allows us to prove also the NP-completeness of checking 
 $\ghw(H)\leq k$ for $k=2$. 
After that, we identify meaningful restrictions
which make checking for bounded $\ghw$ or $\fhw$ tractable or
allow for an efficient approximation of the $\fhw$.
\end{abstract}

\maketitle

\section{Introduction and Background}
\label{sect:introduction}

\noindent{\bf Research Challenges Tackled.}\  
In this work we tackle computational problems on hypergraph decompositions, 
which
play a prominent role for %
answering Conjunctive Queries (CQs) and 
solving Constraint Satisfaction Problems (CSPs), which we discuss %
below.

\nop{The treewidth~\cite{} $tw(G)$ of a graph $G$ is a measure of  its degree 
of cyclicity. The treewidth $tw(G)$ is defined as the smallest width over 
all tree decompositions of a graph (definitions will be given below).} 
Many \np-hard graph-based problems become tractable for instances whose 
corresponding graphs have bounded treewidth.
There are, however, many problems for which the structure of an instance is 
better described by a hypergraph than by a graph, for example,  the 
above mentioned CQs and CSPs. 
Treewidth does not generalize hypergraph acyclicity (sometimes also referred to as 
$\alpha$-acyclicity, see e.g.,~\cite{fagin1983degrees,DBLP:conf/vldb/Yannakakis81}).
Hence, proper  hypergraph 
decomposition methods have been developed, in particular, {\it hypertree 
decompositions (HDs)\/}  
\cite{2002gottlob}, the more general  {\it generalized 
hypertree decompositions (GHDs)}~\cite{2002gottlob}, and the yet 
more general   
{\it fractional hypertree decompositions (FHDs)}~\cite{2014grohemarx}, and 
corresponding notions of width of a hypergraph $H$ have been defined: the  {\it 
hypertree width} $\hw(H)$,  the {\em generalized hypertree width} 
$\ghw(H)$, and the {\em fractional hypertree width} $\fhw(H)$, where for every 
hypergraph $H$, $\fhw(H)\leq \ghw(H)\leq \hw(H)$ holds. 
Definitions are given in 
Section~\ref{sect:prelim}. A number of highly relevant hypergraph-based 
problems such as  
CQ\hyp{}evaluation and CSP\hyp{}solving become tractable for classes of instances of 
bounded 
$\hw$, $\ghw$, or, $\fhw$. For each of the mentioned types of 
decompositions it would thus be %
useful to be able to recognize for each 
constant $k$ whether a given hypergraph $H$ has corresponding  
width at most  
$k$, and if so, to compute such a %
decomposition. More formally, for 
{\it decomposition\/} $\in \{$HD, GHD, FHD$\}$ and $k > 0$, we consider the 
following family of problems:

\medskip
\noindent
\rec{{\it decomposition\/},\,$k$}\\
\begin{tabular}{ll}
 \bf input: & hypergraph $H = (V,E)$;\\
 \bf output: & {\it decomposition\/}  of $H$ of width $\leq k$ if it 
exists and  \\ & answer `no' otherwise.
\end{tabular}

\smallskip

As shown in~\cite{2002gottlob}, \rec{HD,\,$k$}\  is in \ptime for every fixed $k$. However, 
little
is known
about  \rec{FHD,\,$k$}. 
In fact, this has been 
an open problem since the 2006 paper \cite{DBLP:conf/soda/GroheM06},
where Grohe and Marx state: \lq\lq{}It remains an important open question whether 
there
is a polynomial-time algorithm that determines (or approximates)
the fractional hypertree width and constructs a corresponding
decomposition.\rq\rq{}
Since then, the approximation problem has been resolved~\cite{DBLP:journals/talg/Marx10}.
Regarding the problem of determining the exact width, the 2014 journal version \cite{2014grohemarx}
still 
mentions this as 
open 
and it is conjectured that the problem might be \np-hard. 
The open problem is restated in \cite{bevern2015}, where further evidence for 
the hardness of the problem is 
given by  showing that ``it is not expressible in monadic second-order logic 
whether a 
hypergraph has bounded (fractional, generalized) hypertree width''.
We will tackle this open problem here:

\begin{myintro}
\noindent{\bf Research Challenge 1:} \ Is  \rec{FHD,\,$k$} tractable?
\end{myintro}

Let us now  turn to generalized hypertree decompositions. In~\cite{2002gottlob} 
the complexity of  \rec{GHD,\,$k$} was stated as an open problem. 
In \cite{2009gottlob}, it was shown that 
\rec{GHD,\,$k$} is \np-complete for $k\geq 3$. 
For 
$k=1$ the problem is trivially tractable because
$\ghw(H)=1$ just means 
$H$ is acyclic. However the case $k=2$ has been left open.  This case is quite 
interesting, because it was observed that the majority of practical queries 
from 
various  benchmarks that are not acyclic have 
$\ghw = 2$ \cite{DBLP:journals/pvldb/BonifatiMT17,pods/FischlGLP19}, and that a decomposition in such cases can be very 
helpful.
Our second 
research goal is to finally  settle the complexity of \rec{GHD,\,$k$} 
completely.

\begin{myintro}
\noindent{\bf Research Challenge 2:} \ Is \rec{GHD,\,$2$} tractable?
\end{myintro}

For those problems which are  known to be intractable, for example, 
\rec{GHD,\,$k$} for $k\geq 3$, and for those others that will turn out to be 
intractable, we would like to find large islands of tractability 
that correspond to meaningful restrictions of the input 
hypergraph instances. Ideally, such restrictions should  fulfill two main 
criteria: (i) they need to be {\it realistic} in the sense that they apply to a 
large number of CQs and/or CSPs in real-life applications, and 
(ii) 
they need to be {\it non-trivial}  in the sense that the restriction itself 
does 
not already imply bounded $\hw$, $\ghw$, or $\fhw$. 
Trivial restrictions would be, for example, 
acyclicity or bounded 
treewidth.
Hence, our third research problem is as follows:

\begin{myintro}
\noindent{\bf Research Challenge 3:} \ Find realistic, non-trivial restrictions 
on hypergraphs which entail the tractability of  the 
\rec{{\it decomp\/},\,$k$} problem for {\it decomp\/} $\in 
\{$GHD, 
FHD$\}$.
\end{myintro}

Where we do not achieve {\sc Ptime} algorithms 
for the precise computation of a decomposition of optimal width, we would like 
to find tractable methods for achieving good approximations. 
Note that 
for GHDs, the problem of approximations is solved, since $\ghw(H) \leq 3 \cdot \hw(H) +1$ holds for every 
hypergraph $H$ \cite{DBLP:journals/ejc/AdlerGG07}.
In contrast, for FHDs, the best known polynomial-time approximation is cubic. More precisely, 
in \cite{DBLP:journals/talg/Marx10}, a polynomial-time algorithm is presented which, 
given a hypergraph $H$ with $\fhw(H) = k$,
computes an FHD of width $\calO(k^3)$. 
We would like to find 
meaningful restrictions that guarantee significantly tighter approximations
in polynomial time. This leads to the fourth research problem:

\begin{myintro}
\noindent{\bf Research Challenge 4:} \ Find realistic, non-trivial restrictions 
on hypergraphs which allow us to compute in {\sc Ptime} good 
approximations of $\fhw(k)$.
\end{myintro}

\smallskip

\noindent{\bf Background and Applications.}\ Hypergraph decompositions have 
meanwhile found their way into commercial database systems such as LogicBlox 
\cite{DBLP:conf/sigmod/ArefCGKOPVW15,
olteanu2015size,BKOZ13,KhamisNRR15,KhamisNR16} and advanced research prototypes 
such as 
EmptyHeaded~\cite{DBLP:conf/sigmod/AbergerTOR16,tu2015duncecap,aberger2016old}. 
Moreover, 
since CQs and CSPs of bounded hypertree width fall into the highly 
parallelizable complexity class LogCFL~\cite{2002gottlob}, hypergraph decompositions have also 
been 
discovered as a useful tool for parallel query processing with MapReduce 
\cite{DBLP:journals/corr/AfratiJRSU14}. Hypergraph decompositions, in 
particular, HDs and GHDs have been used in many other contexts, e.g., 
in combinatorial auctions~\cite{gottlob2013decomposing} and automated selection 
of Web services based on recommendations from social 
networks~\cite{hashmi2016snrneg}.
Exact algorithms for computing the generalized and fractional hypertree width were published, for example, in~\cite{moll2012}; these algorithms require exponential time, which, in the light of~\cite{2009gottlob} and our present results cannot be improved in the general case.

CQs are the most basic and arguably the most important class
of queries
in the database world.
Likewise, CSPs constitute one of the most fundamental classes of problems
in Artificial Intelligence. 
Formally, CQs and CSPs are the same problem and correspond to first-order 
formulae using $\{\exists,\wedge\}$ but disallowing $\{\forall, \vee, \neg\}$ 
as 
connectives, that need to be evaluated over a set of finite relations:  the 
{\em 
database relations} for CQs, and the {\em constraint relations} for CSPs.
In practice, CQs have often fewer conjuncts (query atoms) and larger relations, 
while CSPs have more conjuncts but smaller relations.
These problems are well-known to be \np-complete
\cite{DBLP:conf/stoc/ChandraM77}. 
Consequently, there has been an intensive search for tractable fragments of CQs 
and/or CSPs
over the past decades. 
For our work, the approaches based on decomposing the 
structure of 
a given CQ or CSP are most relevant,~see~e.g.\
\cite{DBLP:conf/adbt/GyssensP82,%
DBLP:journals/ai/DechterP89,%
DBLP:conf/aaai/Freuder90,%
DBLP:journals/ai/GyssensJC94,%
DBLP:journals/jcss/KolaitisV00,%
DBLP:conf/stoc/GroheSS01,%
DBLP:conf/cp/DalmauKV02,%
DBLP:journals/tcs/ChekuriR00,%
2002gottlob,%
DBLP:conf/cp/ChenD05,%
DBLP:journals/jacm/Grohe07,%
DBLP:journals/jcss/CohenJG08,%
DBLP:journals/mst/Marx11,%
DBLP:journals/jacm/Marx13,%
DBLP:journals/siamcomp/AtseriasGM13,%
2014grohemarx}.
The underlying structure of both %
is nicely captured by 
hypergraphs. 
The hypergraph $H = (V(H), E(H))$ underlying a CQ (or a CSP)  $Q$ has as vertex 
set $V(H)$ the set of variables occurring in $Q$; moreover, for every atom in 
$Q$, $E(H)$ contains a hyper\-edge consisting 
of all variables occurring in this atom.  
From now on, we shall mainly talk about hypergraphs with the understanding that 
all our results are equally applicable to CQs and~CSPs.

\medskip

\noindent
{\bf Main Results.}    
First of all,  we have investigated  the above mentioned open problem 
concerning 
the recognizability of $\fhw \leq k$ for fixed $k$. Our initial hope was to 
find 
a simple adaptation of the \np-hardness proof in \cite{2009gottlob}
for recognizing $\ghw(H) \leq k$, for $k\geq 3$. Unfortunately, this proof 
dramatically fails for the fractional case. In fact,  the hypergraph-gadgets in 
that proof are such that both  \lq\lq{}yes\rq\rq{} and \lq\lq{}no\rq\rq{} 
instances may yield the same $\fhw$. However, via crucial modifications, 
including 
the introduction of novel gadgets, we succeed to construct a reduction from 
3SAT that allows us to control the $\fhw$ of the resulting  hypergraphs
such that those hypergraphs arising from  \lq\lq{}yes\rq\rq{} 3SAT instances 
have  $\fhw(H)=2$ and those 
arising from  \lq\lq{}no\rq\rq{} instances have $\fhw(H)>2$. Surprisingly, 
thanks to our new gadgets, the resulting proof is 
actually significantly simpler than the \np-hardness proof for recognizing 
$\ghw(H) \leq k$ in \cite{2009gottlob}. We thus obtain the following result: 

\begin{myintro}
\noindent{\bf Main Result 1:} Deciding $\fhw(H) \leq 2$ for hypergraphs $H$ is 
\np-complete and, therefore, 
\rec{FHD,\,$k$} is  intractable
even for $k = 2$.
\end{myintro} 

\noindent
This  result can be extended to the \np-hardness of recognizing
$\fhw(H) \leq k$ for arbitrarily large $k$. Moreover, the same 
construction can be used to prove that recognizing ghw $\leq 2$ is 
also \np-hard, thus killing two birds with one stone.

\begin{myintro}
\noindent{\bf Main Result 2:} Deciding $\ghw(H) \leq 2$ for hypergraphs $H$ is 
\np-complete and, therefore, 
\rec{GHD,\,$2$} is  intractable
even for $k = 2$.
\end{myintro}

The Main Results 1 and 2 are presented in Section \ref{sect:hardness}. 
These results close some smouldering open problems with bad news. We thus 
further 
concentrate on Research Challenges 3 and 4 in order to obtain some positive results 
for 
restricted hypergraph classes. 

We first study GHDs, where we succeed to identify very 
general, realistic, 
and 
non-trivial restrictions that make the \rec{GHD,\,$k$} problem tractable. 
These results are based on new insights about
the differences of GHDs and HDs,
  and conceptually splitting the problem into two tasks:
  The first problem, for a given list of possible bags of vertices and a hypergraph $H$, consists in 
  checking whether there exists a tree decomposition
  of $H$ using only bags from the input. In Section~\ref{sect:framework}, we show that this problem is \np-complete in general but becomes tractable when we introduce a mild restriction on the tree decompositions.
  The second problem consists in finding restrictions under which we only need to consider a polynomial number of possible bags.

In particular, we concentrate on the {\em bounded intersection property 
(BIP)}, which, for a class $\classC$ of hypergraphs  requires that for some 
constant $i$, for each pair of distinct edges $e_1$ and $e_2$ of each 
hypergraph 
$H\in {\classC}$,
$|e_1\cap e_2|\leq i$, and its generalization, the {\em bounded 
multi-intersection property (BMIP)}, which
requires that for some 
constant $c$ any intersection of $c$ distinct hyperedges of $H$ has at most $i$ 
elements for some constant $i$. 
A recent empirical study~\cite{pods/FischlGLP19} 
of a large 
number of known CQ and CSP benchmarks showed that a high portion 
of instances coming from real-life applications indeed enjoys the BIP for low constant $i$ 
and a yet higher portion enjoys the BMIP for very low constants $c$ and $i$. We 
obtain the following favorable results, which are presented in Section \ref{sect:ghd}. 

\begin{myintro}
\noindent{\bf Main Result 3:} For classes of hypergraphs fulfilling the BIP or 
BMIP, for every constant $k$,  the problem \rec{GHD,\,$k$} is tractable. 
Tractability holds even 
for classes~$\classC$ 
of 
hypergraphs where for some constant $c$ all intersections of $c$ distinct edges 
of every $H\in{\classC}$ of size $n$ have $\calO(\log n)$ elements.
Our complexity analysis reveals that for fixed $k$ and $c$
the problem 
 \rec{GHD,\,$k$} 
is %
fixed-parameter tractable
parameterized by
$i$ 
of the BMIP.
\end{myintro} 

  The tractability proofs for GHDs do not directly carry over to FHDs.
  Still, under slightly less general conditions and with some 
  additional combinatorial insights it is possible to reduce the \rec{FHD,\,k}
  problem to the \rec{GHD,\,k} scenario of the previous result. In particular, we obtain results for the BIP and a further special case of the BMIP.
  We then consider 
 the {\em degree\/} $d$ of  a hypergraph $H = (V(H), E(H))$, 
which is defined as the 
maximum number of hyperedges in which a vertex occurs, i.e., 
$d = \max_{v \in V(H)} |\{ e \in E(H) \mid v \in E(H)\}|$.
We say that a class $\classC$ of hypergraphs  has the {\em bounded degree property (BDP)\/}, 
if there exists $d \geq 1$, 
such that every hypergraph $H\in {\classC}$ has degree $\leq d$.
We 
obtain the following results, which are presented in Section~\ref{sect:fhd-exact}. 

\begin{myintro}
  \noindent{\bf Main Result 4:}
For classes of hypergraphs fulfilling either the BDP or the BIP and
for every constant $k$,  the problem \rec{FHD,\,$k$} is tractable. 
\end{myintro} 

To get yet bigger tractable classes, we also consider approximations of an optimal FHD. 
Towards this goal, we study the $\fhw$ in case of the BMIP and we establish an 
interesting connection between the BMIP and the 
Vapnik--Chervonenkis dimension (VC-dimension) of hypergraphs. 
Our research, presented in Section~\ref{sect:fhd} is summarized as 
follows.

\begin{myintro}
\noindent{\bf Main Result 5:} For rather general, realistic, and non-trivial 
hypergraph restrictions, there exist {\sc Ptime} algorithms that, for 
hypergraphs $H$ with $\fhw(H)=k$, where $k$ is a constant,  produce FHDs whose 
widths are significantly smaller than the best previously known 
approximation. 
In particular, the BMIP allows us to compute in polynomial time 
an FHD whose width  is $\leq k + \epsilon$
for arbitrarily chosen constant $\epsilon > 0$. 
Bounded VC-dimension 
allows us to compute in polynomial time an FHD whose width  is $\calO(k \log k)$. 
\end{myintro} 

We finally turn our attention also to the optimization problem of fractional hypertree width, i.e., given a hypergraph $H$, 
determine $\fhw(H)$ and find an FHD of width $\fhw(H)$. All our algorithms for the \rec{FHD,\,$k$} problem 
have a runtime exponential in the desired width $k$. Hence, even with the restrictions 
to the BIP or BMIP we cannot expect 
an efficient approximation of $\fhw$ if $\fhw$ can become arbitrarily large. We will therefore study the following 
\boundedopt\ problem
for constant $K \geq 1$: 

\medskip
\noindent
\boundedopt\\
\begin{tabular}{lrl}
 \bf input: &  \multicolumn{2}{l}{hypergraph $H = (V,E)$;} \\
 \bf output: & if $\fhw(H) \leq K$: & find an FHD $\mcF$ of $H$ with minimum width; \\
 & otherwise: & answer ``$\fhw(H) > K$''.
\end{tabular}

\medskip
\noindent
For this bounded version of the optimization problem, we  will prove the following result: 

\begin{myintro}
  \noindent{\bf Main Result 6:} There exists a polynomial time
  approximation scheme (PTAS; for details see Section \ref{sect:fhd})
  for the \boundedopt\ problem in case of the BMIP for any fixed
  $K\geq 1$.
\end{myintro}

\section{Preliminaries}
\label{sect:prelim}

For integers $n \geq 1$ we write $[n]$ to denote the set $\{1, \dots, n\}$.

\subsection{Hypergraphs}

A {\em hypergraph} is a pair $H = (V(H), E(H))$, consisting of a set $V(H)$ of 
{\em vertices} and a set $E(H)$ of {\em hyperedges} (or, simply {\em edges\/}), 
which are non-empty 
subsets of $V(H)$. We assume that hypergraphs do not have isolated 
vertices, i.e.\ for each $v \in V(H)$, there is at least one edge $e \in E(H)$, 
s.t.\ $v \in e$. For a set $C \subseteq V(H)$, we define $\edges(C) = \{ e \in 
E(H) \mid e \cap C \neq \emptyset \}$ and for a set $S \subseteq E(H)$, we 
define $\V(S) = \{ v \in e \mid e \in S \}$.
The {\em rank} of a hypergraph $\HH$ (denoted $\rarity{\HH}$)
is the 
maximum cardinality of any edge  $e$ of $\HH$.
We refer to the number of edges and vertices as $|E(H)|$ and $|V(H)|$, respectively. 
The \emph{size} of (some reasonable representation of) $H$ will be denoted as $||H||$, 
i.e., $n \leq |V(\HH)| + |E(\HH)| \cdot |V(\HH)|$.

We sometimes identify sets of edges with hypergraphs. If a set of edges $E$ is used, where instead a hypergraph is expected, then we mean the hypergraph $(V,E)$, where $V$ is simply the union of all edges in $E$. 
For a set $S$ of edges, it is convenient to write $\bigcup S$ (and $\bigcap S$, respectively)  to denote the set of vertices obtained by taking the union
(or the intersection, respectively) of the edges in $S$. Hence, we can write $\V(S)$ simply as $\bigcup S$.

For a hypergraph $H$ and a set $C \subseteq V(H)$, we say that
a \path{$C$} $\pi$ from $v$ to $v'$ consists of a sequence
$v = v_0,\dots,v_h = v'$ of vertices and a sequence of edges
$e_0, \dots, e_{h-1}$ ($h \geq 0$) such that
$\{ v_i, v_{i+1} \} \subseteq ( e_i \setminus C)$, for each
$i \in \{0,\ldots, h-1\}$.  We denote by $\vertices(\pi)$ the set of
vertices occurring in the sequence $v_0,\ldots, v_h$.  Likewise, we
denote by $\edges(\pi)$ the set of edges occurring in the sequence
$e_0,\ldots,e_{h-1}$.
A set $W \subseteq
V(H)$ of vertices is \connected{$C$} if $\forall v,v' \in W$ there is a 
\path{$C$} from $v$ to $v'$. A \comp{$C$} is a maximal \connected{$C$}, 
non-empty
set of vertices $W \subseteq V(H)\setminus C$.

The \emph{primal graph} $G$ of a hypergraph $H$ is the graph with the same vertices as $H$ and an edge between vertices $v$ and $u$
iff there exists an edge $e \in E(H)$ such that $\{v,u\}\subseteq e$.

Given a hypergraph $H = \{V,E)$, the \emph{dual hypergraph}
$H^d  = (W,F)$ 
is defined as $W = E$ and $F = \{ \{e \in E \mid v \in e\} \mid v \in V\}$.

\subsection{(Fractional) Edge Covers}
Let $H = (V(H),E(H))$ be a hypergraph 
and consider (edge-weight) functions $\lambda \colon E(H) \ra \{0,1\}$ and
$\gamma \colon E(H) \ra [0,1]$. 
For $\theta \in \{\lambda, \gamma\}$, 
we denote by $B(\theta)$ the set of all 
vertices {\em covered\/} by $\theta$:
\[ B(\theta) = \left\{ v\in V(H) \mid \sum_{e\in E(H), v\in e} \theta(e) \geq 1 
\right\}\]
The weight of 
such a 
function 
$\theta$ is defined as
\[ \weight(\theta) = \sum_{e \in E(H)} \theta(e). \]
Following \cite{2002gottlob}, we will sometimes consider
$\lambda$ as a set with $\lambda \subseteq E(H)$ 
(i.e., the set of edges $e$ with $\lambda(e) = 1$)
and the weight of $\lambda$ as the cardinality of this set. %
However, for the  sake of a 
uniform treatment with function $\gamma$,  we shall %
prefer 
to treat $\lambda$ as a 
function.

\begin{definition}
An {\em edge cover} of a hypergraph $H$ %
is a function 
$\lambda : E(H) \ra \{0,1\}$ 
such that $V(H) = B(\lambda)$. The {\em edge cover number} $\rho(H)$
is the  minimum weight of all edge covers of $H$.
\end{definition}

Note that the edge cover number
can be calculated by the following integer linear 
program (ILP).
 \[
 \begin{aligned}
  \text{minimize: } & \sum_{e\in E(H)} \lambda(e) & \\
  \text{subject to: } & \quad\sum_{\mathclap{e \in E(H), v \in e}} \;\lambda(e) 
                        \geq 1, & & \text{for all } v \in V(H)\\
                      & \lambda(e) \in \{0,1\} & & \text{for all } e \in E(H)
 \end{aligned}
 \]
By substituting all $\lambda(e)$ by $\gamma(e)$ and by relaxing the last condition of the ILP  above
to $\gamma(e) \geq 0$, 
we arrive
at the linear program (LP) for computing the fractional edge cover number
to be defined next.
Note that even though our weight function is defined
to take values between 0 and 1, we do not need to add $\gamma(e) \leq 1$ as a 
constraint,
because implicitly by the minimization itself the weight on an edge for
an edge cover is never greater than 1.
Also note that now the program above
is an LP, which (in contrast to an ILP) can be solved in \ptime even if  $k$ is not fixed.

\begin{definition}
A {\em fractional edge cover} of a hypergraph $H=(V(H),E(H))$ is a 
function 
$\gamma : E(H) \ra [0,1]$  
such that $V(H) = B(\gamma)$. The {\em fractional edge cover number} $\rho^*(H)$ of 
$H$ %
is the 
 minimum weight of all fractional edge covers of $H$.
We write $\cov(\gamma)$ to denote the {\em support\/} of 
$\gamma$, i.e., 
$\cov(\gamma) := \{ e\in E(H) \mid \gamma(e) > 0\}$.
\end{definition}

  We also extend the above definitions to subsets $S \subseteq V(H)$, i.e., an edge cover of $S$ in $H$
  is a function $\lambda: E(H) \ra \{0,1\}$ such that $S \subseteq B(\lambda)$. 
  If $H$ is clear from the context, we shall simply speak of ``an edge cover of $S$'' 
  without explicitly mentioning $H$.
  The edge cover number $\rho(S)$ is then the minimum weight of all edge covers of $S$. The definitions for the fractional case are extended analogously.

Clearly, we have $\rho^*(H) \leq \rho(H)$ for every hypergraph $H$, and 
$\rho^*(H)$ can  
be much smaller than $\rho(H)$. However, below we give 
an example,  which is important for 
our proof of Theorem \ref{thm:npcomp} and where $\rho^*(H)$ and $\rho(H)$ 
coincide.

\begin{lemma}
\label{lem:cliquewidth}
    Let $K_{2n}$ be a clique of size $2n$. Then the equalities $\rho(K_{2n}) = \rho^*(K_{2n}) 
    = n$ hold.
\end{lemma}

\begin{proof}
Since we have to cover each vertex with weight $\geq 1$, the total 
weight on the 
vertices of the graph is $\geq 2n$. As the weight of each edge adds to the 
weight of 
at most 2 vertices, we need at least weight $n$ on the edges to achieve $\geq 
2n$ 
weight on the vertices.
On the other hand,  we can use $n$ edges each with weight 1 to cover $2n$ 
vertices. Hence, in total, we get
$n \leq \rho^*(K_{2n}) \leq \rho(K_{2n}) \leq  n$.
  \end{proof}

\subsection{HDs, GHDs, and FHDs}

We now define tree decompositions and three types of hypergraph decompositions:

\begin{definition}
  \label{def:TD}
  A {\em tree decomposition\/} (TD) of a hypergraph 
$H=(V(H),E(H))$ 
is a tuple 
$\left< T, (B_u)_{u\in N(T)} \right>$, such that 
$T = \left< N(T),E(T)\right>$ is a rooted tree and 
the 
following conditions hold:

\begin{enumerate}[label=\emph{(\arabic*})]
 \item[(1)] for each $e \in E(H)$, there is a node $u \in N(T)$ with $e \subseteq 
B_u$;
 \item[(2)] for each $v \in V(H)$, the set $\{u \in N(T) \mid v \in B_u\}$ is 
connected 
\end{enumerate}

\end{definition}

\begin{definition}
  \label{def:GHD}
 A 
{\em generalized hypertree decomposition\/} (GHD) 
of a hypergraph 
$H=(V(H),E(H))$ is a tuple 
$\left< T, (B_u)_{u\in N(T)}, (\lambda)_{u\in N(T)} \right>$, where
$\left< T, (B_u)_{u\in N(T)} \right>$ is a TD of $H$ and, additionally, the following 
condition (3) holds:
\begin{enumerate}[label=\emph{(\arabic*})]
\item[(3)]
for each $u\in N(T)$, $\lambda_u$ is a function $\lambda_u \colon E(H)  \ra  \{0,1\}$
 with 
 $B_u \subseteq  B(\lambda_u)$.
\end{enumerate}

\end{definition}

Let us clarify some notational conventions used throughout this paper.
To avoid confusion, we will consequently refer to the 
elements in 
$V(H)$ as {\em vertices\/} (of the hypergraph) and to the elements in $N(T)$ as 
the {\em nodes\/}
of $T$ (of the decomposition). 
Now consider a decomposition $\mcG$ with tree structure~$T$. 
For a node $u$ in $T$, 
we write $T_u$ to denote the subtree of $T$ rooted at $u$.
By slight abuse of notation, we will often write $u' \in T_u$ to denote
that $u'$ is a node in the subtree $T_u$ of $T$.
Moreover, we define $\VTu 
:= \bigcup_{u' \in T_u} B_{u'}$
and, for a set $V' 
\subseteq V(H)$, we define $\nodes(V') = 
\{ u \in T \mid B_u \cap V' \neq \emptyset \}$.
If we want to make explicit the decomposition $\mcG$, 
we also write $\nodes(V', \mcG)$ synonymously with $\nodes(V')$.
By further overloading the $\nodes$ operator, we also write 
$\nodes(T_u)$ or $\nodes(T_u, \mcG)$ to denote the 
nodes in a subtree $T_u$ of $T$, i.e.,  $\nodes(T_u) = \nodes(T_u,\mcG) =\{ v \mid v \in T_u \}$.

\begin{definition}
 \label{def:HD}
 A {\em hypertree decomposition\/} (HD) of a hypergraph 
$H=(V(H),E(H))$  is a GHD, which in addition also 
satisfies the following condition (4):
\begin{enumerate}[label=\emph{(\arabic*})]
 \item[(4)] for each $u\in N(T)$, $ V(T_u) \cap B(\lambda_u) \subseteq B_u$ 
\end{enumerate}
\end{definition}

\begin{definition}
 \label{def:FHD}
 A 
{\em fractional hypertree decomposition\/} (FHD) 
\cite{2014grohemarx}
of a hypergraph 
$H=(V(H),E(H))$ is a tuple 
$\left< T, (B_u)_{u\in N(T)}, (\gamma)_{u\in N(T)} \right>$, where
$\left< T, (B_u)_{u\in N(T)} \right>$ is a TD of $H$ and, additionally, the following 
condition 
(3') holds:
\begin{enumerate}[label=\emph{(\arabic*})]
 \item[(3')] for each $u\in N(T)$, $\gamma_u$ is a function $\gamma_u : E(H) \ra 
[0,1]$
with $B_u \subseteq  B(\gamma_u)$.
\end{enumerate}
\end{definition}

The width of a GHD, HD, or FHD is the maximum weight of the functions 
$\lambda_u$ or $\gamma_u$, 
respectively, over all nodes $u$ in $T$. Moreover, the generalized hypertree 
width,
hypertree width, and fractional hypertree width of $H$ (denoted $\ghw(H)$, 
$\hw(H)$, $\fhw(H)$) is the minimum width over all GHDs, HDs, and FHDs of $H$, 
respectively.
  Alternatively, we could define the $\ghw$ of a TD 
  $\left< T, (B_u)_{u\in N(T)} \right>$
  as  $\max_{u \in N(T)}\rho(B_u)$. It is clear that the definitions lead to equivalent notions of $\ghw$ for a hypergraph.
The same is true for the $\fhw$ of a TD which corresponds to $\max_{u\in N(T)} \rho^*(B_u)$.
Condition~(2) is usually called the ``connectedness condition'', and condition~(4) is 
referred to as ``special condition'' \cite{2002gottlob}. The set $B_u$ is often 
referred to as the 
``bag'' at node $u$. 
Note 
that, 
strictly speaking, only HDs require that the underlying tree $T$ be rooted. 
For the sake of a uniform treatment we assume that also the tree underlying a 
GHD or an FHD is rooted (with the understanding that the root is arbitrarily 
chosen).

We now recall two fundamental properties
of the various notions of decompositions and width.

\begin{lemma}
  \label{lem:subhypergraph}
  Let $H$ be a hypergraph and let $H'$ be a vertex induced subhypergraph of $H$, then 
$\hw(H') \leq \hw(H)$,   
$\ghw(H') \leq \ghw(H)$,
and $\fhw(H') \leq \fhw(H)$ hold. 
\end{lemma}

\begin{lemma}
  \label{lem:clique}
  Let $H$ be a hypergraph. If there exists a vertex set
  $S \subseteq V(H)$ such that $S$ is a clique in the primal graph,
  then every HD, GHD, or FHD of $H$ has a node $u$ such that
  $S \subseteq B_u$.
\end{lemma}
\noindent
Strictly speaking, Lemma~\ref{lem:clique} is a well-known property of tree 
decompositions -- independently of the $\lambda$- or $\gamma$-label.

\section{NP-Hardness}
\label{sect:hardness}

The main result in this section is the \np-hardness
of \rec{{\it decomp\/},\,$k$} with  
{\it decomp\/} $\in \{$GHD, FHD$\}$ and $k = 2$.
The proof is rather technical, we therefore begin with an informal overview of the overall strategy.

The reduction is from 3SAT. We introduce the sets $Y, Y'$ for the literals, sets $A = \{a_1, \dots, a_m\}$, $A'=\{a_1',\dots,a'_m\}$, set $S$, and two special vertices $\{z_1, z_2\}$. The hypergraph $H$ to be constructed consists of 3 main parts: two versions of the gadget in Figure~\ref{fig:gadgetH0} (which we will refer to as $H_0$ and $H'_0$, respectively) and a  subhypergraph encoding the clauses of the 3SAT instance. The subhypergraph $H_0$  contains the vertices 
$A,Y,S,z_1,z_2$, while the subhypergraph $H'_0$ contains the vertices 
$A',Y',S,z_1,z_2$. The intended GHD/FHD consists of a long path which connects a decomposition of $H_0$ at its left end and a decomposition of $H'_0$ at its right end. A sketch of this  decomposition is shown in 
Figure \ref{fig:overviewnp}.

We assume that each of the bags $B_1,\dots,B_m$ on the long path contains a subset of $A\cup A'$, a subset $Y_{sol}$ of $Y \cup Y'$ corresponding to the solution, $S$, and $z_1, z_2$. Each such bag should be covered by a pair of edges, 
such that one covers $z_1$ and there is a unique edge covering $z_2$ that it can be paired with and vice versa: 
every edge covering $z_1$ covers all but a
small subset of vertices of $S$, and there is a unique edge covering $z_2$ that fills this hole. Furthermore, every pair of edges covering $S \cup \{z_1, z_2\}$ covers exactly a subset $A^*_j = \{a'_1,\dots, a'_j, a_j,\dots, a_m\}$ of $A \cup A'$ for some $j$. By making $a_j$ and $a'_j$ adjacent for every $j$, we can make sure that the only way the decomposition can move from $A$ on the left to $A'$ on the right is if bag $B_j$ contains exactly such a subset $A^*_j$ of $A \cup A'$. Then we make the following connection to the original 3SAT instance: there are exactly 3 pairs of edges covering $\{a'_1, \dots, a'_j, a_j,\dots, a_m\}$ corresponding to the three literals of the $j$-th clause $(\ell_1 \lor \ell_2 \lor \ell_3)$ of the formula. Each such pair fully covers $Y \cup Y'$ except for the \emph{negation} of one of the literals (i.e., $\overline{\ell_1}$, $\overline{\ell_2}$, or $\overline{\ell_3}$). If $Y_{sol}$ satisfies the clause, then it \emph{does not } contain one of $\overline{\ell_1}$, $\overline{\ell_2}$, or $\overline{\ell_3}$, so one of the three pairs can fully cover $Y_{sol}$ and hence all of bag $B_j$. In the end, the bags $B_1, \dots, B_m$ verify that each of the $m$ clauses of the formula is satisfied by $Y_{sol}$.

Our gadget construction makes sure that the decompositions of $H_0$ and $H'_0$ indeed contain the specified sets of vertices
(see Lemma~\ref{lem:gadgetH0} for the details). Moreover, we have to ensure that the same subset $Y_{sol}$ of $Y \cup Y'$ appears in every $B_1,\dots,B_j$ if we want to read out a satisfying assignment from the decomposition. In principle it is possible that a $y'_i$ appears or a $y_i$ disappears as we move from $B_1$ to $B_m$. But there is an easy standard solution for this problem: let us repeat the path $2n+1$ times, doing the full check of the $m$ clauses $2n+1$ times. As there are only $2n$ possible changes of the $y_i$ and $y'_i$ vertices, one of the $2n+1$ copies is free of changes, and hence can be used to deduce a satisfying assignment.
\pgfdeclaredecoration{discontinuity}{start}{
  \state{start}[width=0.5\pgfdecoratedinputsegmentremainingdistance-0.5\pgfdecorationsegmentlength,next state=first wave]
  {}
  \state{first wave}[width=\pgfdecorationsegmentlength, next state=second wave]
  {
    \pgfpathlineto{\pgfpointorigin}
    \pgfpathmoveto{\pgfqpoint{0pt}{\pgfdecorationsegmentamplitude}}
    \pgfpathcurveto
        {\pgfpoint{-0.25*\pgfmetadecorationsegmentlength}{0.75\pgfdecorationsegmentamplitude}}
        {\pgfpoint{-0.25*\pgfmetadecorationsegmentlength}{0.25\pgfdecorationsegmentamplitude}}
        {\pgfpoint{0pt}{0pt}}
    \pgfpathcurveto
        {\pgfpoint{0.25*\pgfmetadecorationsegmentlength}{-0.25\pgfdecorationsegmentamplitude}}
        {\pgfpoint{0.25*\pgfmetadecorationsegmentlength}{-0.75\pgfdecorationsegmentamplitude}}
        {\pgfpoint{0pt}{-\pgfdecorationsegmentamplitude}}
}
\state{second wave}[width=0pt, next state=do nothing]
  {
    \pgfpathmoveto{\pgfqpoint{0pt}{\pgfdecorationsegmentamplitude}}
    \pgfpathcurveto
        {\pgfpoint{-0.25*\pgfmetadecorationsegmentlength}{0.75\pgfdecorationsegmentamplitude}}
        {\pgfpoint{-0.25*\pgfmetadecorationsegmentlength}{0.25\pgfdecorationsegmentamplitude}}
        {\pgfpoint{0pt}{0pt}}
    \pgfpathcurveto
        {\pgfpoint{0.25*\pgfmetadecorationsegmentlength}{-0.25\pgfdecorationsegmentamplitude}}
        {\pgfpoint{0.25*\pgfmetadecorationsegmentlength}{-0.75\pgfdecorationsegmentamplitude}}
        {\pgfpoint{0pt}{-\pgfdecorationsegmentamplitude}}
    \pgfpathmoveto{\pgfpointorigin}
}
  \state{do nothing}[width=\pgfdecorationsegmentlength,next state=do nothing]{
    \pgfpathlineto{\pgfpointdecoratedinputsegmentlast}
  }
  \state{final}
  {
    \pgfpathlineto{\pgfpointdecoratedpathlast}
  }
}

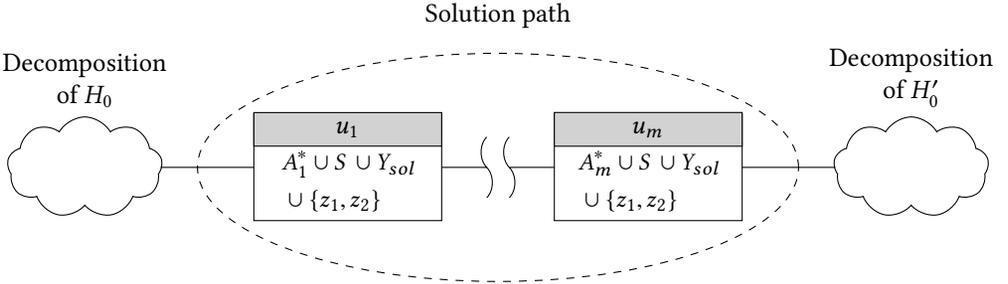
\begin{figure}[h]
  \centering
  \begin{tikzpicture}

    \node [cloud, draw,cloud puffs=9,cloud puff arc=120, aspect=2, inner ysep=1em,
    label={[align=center]Decomposition\\of $H_0$}] at (-0.5,0)
    {};

    \node [cloud, draw,cloud puffs=9,cloud puff arc=120, aspect=2, inner ysep=1em,
     label={[align=center]Decomposition\\of $H'_0$}] at (10.5,0)
     {};

     \draw[dashed] (5,0) ellipse (4cm and 1.5cm);
     \node at (5, 2) {Solution path};

     \node (rect) at (3,0) [draw,minimum width=2.5cm,minimum height=1.4cm] {};
     \node at (3, -0.22) {\small
       $\begin{aligned}
         &A^*_1 \cup S\, \cup  Y_{sol} \\ &\cup \{z_1, z_2\}
     \end{aligned}$};
     \node (rect) at (3,0.5) [draw,minimum width=2.5cm,minimum height=0.4cm, fill=lightgray!70] {$u_1$};

     \node (rect) at (7,0) [draw,minimum width=2.5cm,minimum height=1.4cm] {};
     \node at (7, -0.22) {\small
       $\begin{aligned}
         & A^*_m \cup S\, \cup  Y_{sol}\\ & \cup \{z_1, z_2\}
     \end{aligned}$};
     \node (rect) at (7,0.5) [draw,minimum width=2.5cm,minimum height=0.4cm, fill=lightgray!70] {$u_m$};

     \draw (0.51,0) -- (1.75,0);
     \draw (9.49,0) -- (8.25,0);
     \draw [decoration={%
       discontinuity,
       meta-segment length=0.4cm,
       segment length=0.3cm,
       amplitude=0.4cm}, decorate]
     (4.25,0) -- (5.75,0);

  \end{tikzpicture}
  \caption{Sketch of the intended GHD/FHD.}
  \label{fig:overviewnp}
\end{figure}

We begin by defining the aforementioned gadget (Figure~\ref{fig:gadgetH0}), that will play 
an integral part of this construction. 
Its crucial properties are stated below.

\newcommand{\lemGadgetH}{%
Let $M_1$, $M_2$ be disjoint sets and $M=M_1\cup M_2$. Let $H = (V(H),E(H))$ 
be a hypergraph and $H_0 =$ $(V_0, E_A \cup E_B \cup E_C)$ a subhypergraph of 
$H$ with $V_0=\{a_1,a_2,b_1,b_2,c_1,c_2,d_1,d_2\} \cup M$ and
  \begin{align*}
   E_A = \{ &\{a_1,b_1\} \cup M_1, \{ a_2, b_2 \} \cup M_2, 
         \{a_1,b_2\},\{a_2,b_1\}, \{a_1, a_2\} \} \\ 
   E_B = \{ &\{b_1,c_1\} \cup M_1, \{ b_2, c_2 \} \cup M_2, \{b_1,c_2\},
                   \{b_2,c_1\}, \{ b_1,b_2\}, \{c_1,c_2\} \} \\ 
   E_C = \{ & \{c_1,d_1\} \cup M_1, \{ c_2, d_2 \} \cup M_2, 
           \{c_1,d_2\},
                   \{c_2,d_1\}, \{ d_1,d_2\} \}  
  \end{align*}
\noindent  
where no element from the set $R = \{ a_2, b_1, b_2, c_1, c_2, d_1, d_2 \}$ 
occurs in any edge of $E(H) \setminus (E_A\cup E_B\cup E_C)$.
\noindent  
Then, every FHD $\mcF =\left<T,(B_u)_{u\in T},(\gamma_u)_{u\in T}\right>$ of 
width $\leq 2$ of H has nodes $u_A, u_B, u_C$ s.t.:
\begin{itemize} 
 \item $\{a_1,a_2,b_1,b_2\} \subseteq B_{u_A} \subseteq M \cup \{a_1,a_2,b_1,b_2\}$
 \item $B_{u_B} = \{b_1,b_2,c_1,c_2\} \cup M$,
 \item $\{c_1,c_2,d_1,d_2\} \subseteq B_{u_C} \subseteq M \cup \{c_1,c_2,d_1,d_2\}$, and
 \item $u_B$ is on the path from $u_A$ to $u_C$.%
\end{itemize}
}

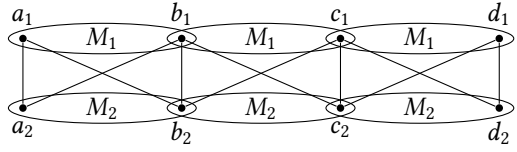
\begin{figure}[t]
    \centering

\usetikzlibrary{positioning,shapes,calc}

\tikzset{
    between/.style args={#1 and #2}{
         at = ($(#1)!0.5!(#2)$)
    }
}

\begin{tikzpicture}[
   vert/.style={fill, circle, inner sep = 1pt},
   ell/.style={ellipse,draw,minimum width=2.5cm, inner sep=0cm}]
   \node[vert,label=above:$a_1$] (a1) {};
   \node[vert,label=below:$a_2$, below=0.8 of a1] (a2) {};
   \node[vert,label=above:$b_1$, right=2 of a1] (b1) {};
   \node[vert,label=below:$b_2$, right=2 of a2] (b2) {};
   \node[vert,label=above:$c_1$, right=2 of b1] (c1) {};
   \node[vert,label=below:$c_2$, right=2 of b2] (c2) {};
   \node[vert,label=above:$d_1$, right=2 of c1] (d1) {};
   \node[vert,label=below:$d_2$, right=2 of c2] (d2) {};

   \draw (a1) -- (a2);
   
   \node[ell, between=a1 and b1] {$M_1$};
   \draw (a1) -- (b2);
   \draw (a2) -- (b1);

   \node[ell, between=a2 and b2] {$M_2$};
   \draw (b1) -- (b2);
   
    \draw (c1) -- (c2);
    \node[ell, between=b1 and c1] {$M_1$};
   \draw (c1) -- (b2);
   \draw (c2) -- (b1);
   \node[ell, between=b2 and c2] {$M_2$};

   \draw (d1) -- (d2);
   \node[ell, between=c1 and d1] {$M_1$};
   \draw (c1) -- (d2);
   \draw (c2) -- (d1);
   \node[ell, between=c2 and d2] {$M_2$};
\end{tikzpicture}
    \caption{Basic structure of $H_0$ in Lemma~\ref{lem:gadgetH0}} 
    \label{fig:gadgetH0}
\end{figure}

\begin{lemma}
\label{lem:gadgetH0} 
\lemGadgetH
\end{lemma}
\begin{proof}
Consider an arbitrary FHD 
$\mcF = \left<T,(B_u)_{u\in T},(\gamma_u)_{u\in T}\right>$ of 
width $\leq 2$ of H. 
      Observe that $a_1, a_2, b_1$, and $b_2$ form a clique of size 4 in the primal graph. Hence, by
      Lemma~\ref{lem:clique}, there is a node $u_A$ in $\mcF$, such that
      $\{a_1,a_2,b_1,b_2\} \subseteq B_{u_A}$. 
      It remains to show that also $B_{u_A} \subseteq M \cup \{a_1,a_2,b_1,b_2\}$ holds.
      To this end, we use a similar reasoning as in the proof of Lemma~\ref{lem:cliquewidth}: 
      to cover each vertex in $\{a_1,a_2,b_1,b_2\}$, we have to put weight $\geq 1$ on each of these 4 vertices. 
      By assumption, the only edges containing 2 out of these 4 vertices are the edges in $E_A \cup \{ \{b_1,b_2\} \}$. 
      All other edges in $E(H)$ contain at most 1 out of these 4 vertices. Hence, in order to cover  
      $\{a_1,a_2,b_1,b_2\}$ with weight $\leq 2$, we are only allowed to put non-zero weight on the edges in 
      $E_A \cup \{ \{b_1,b_2\} \}$. It follows, that $B_{u_A} \subseteq M \cup \{a_1,a_2,b_1,b_2\}$ indeed holds.

      Analogously, for the cliques $b_1, b_2, c_1, c_2$ and $c_1, c_2, d_1, d_2$, 
      there must exist nodes $u_B$ and $u_C$ in $\mcF$ with 
    $\{b_1,b_2,c_1,c_2\}  \subseteq  B_{u_B} \subseteq M \cup \{b_1,b_2,c_1,c_2\}$ and
    $\{c_1,c_2,d_1,d_2\} \subseteq B_{u_C} \subseteq M \cup \{c_1,c_2,d_1,d_2\}$. 
    
    It remains to show that $u_B$ is on the path from $u_A$ to $u_C$ and 
    $B_{u_B} = \{b_1,b_2,c_1,c_2\} \cup M$ holds.
We first show that $u_B$ is on the path between $u_A$ and $u_C$. 
Suppose to the contrary that it is not. We distinguish three cases. 
First, assume that $u_A$ is on the path between $u_B$ and $u_C$. 
Then, by 
connectedness, $\{ c_1, c_2 \} \subseteq B_{u_A}$, which contradicts the property 
$B_{u_A} \subseteq M \cup \{a_1,a_2,b_1,b_2\}$ shown above.
Second, assume $u_C$ is on the path between $u_A$ and $u_B$. In this case, 
we have $\{ b_1, b_2 \} \subseteq B_{u_C}$, which contradicts the property 
$B_{u_C} \subseteq M \cup \{c_1,c_2,d_1,d_2\}$ shown above.
Third, suppose there is no path containing $u_A$, $u_B$, and $u_C$. Then there exists a node $u$
such that removing $u$ from $T$ would put $u_A$, $u_B$, $u_C$ in three different components.
Node $u$ is therefore on the path between any pair of the three nodes. Because $u$ is on the path from $u_A$ to $u_B$,
$\{b_1, b_s\} \subseteq B_u$ by connectedness. Analogously, also $\{c_1, c_2\}\subseteq B_u$ because of the path from $u_C$ to $u_B$. Then, by the same argument as in the beginning we have $\{b_1, b_2, c_1, c_2\} \subseteq B_u \subseteq M \cup \{b_1, b_2, c_1, c_2\}$. Hence, $u$ satisfies all the properties we have established for $u_B$ and we can just consider the node $u$, which is on the path from $u_A$ to $u_C$, to be our $u_B$.

We now show that also $B_{u_B} = \{b_1,b_2,c_1,c_2\} \cup M$ holds.
Since we have already established    
    $\{b_1,b_2,c_1,c_2\}  \subseteq  B_{u_B} \subseteq M \cup \{b_1,b_2,c_1,c_2\}$,
it suffices to show $M \subseteq B_{u_B}$.    
First, let $T'_a$ be the subgraph of $T$ induced by 
$\nodes(\{a_1,a_2\},\mcF)$
     and let $T'_d$ be the subgraph of $T$ induced by 
$\nodes(\{d_1,d_2\},\mcF)$.
     We show that each of the subgraphs $T'_a$ and $T'_d$ is connected
     (i.e., a subtree of $T$) and that the two subtrees are disjoint.
     The connectedness is immediate: by the connectedness condition, each of 
     $\nodes(\{a_1\},\mcF)$, $\nodes(\{a_2\},\mcF)$, 
     $\nodes(\{d_1\},\mcF)$, and $\nodes(\{d_2\},\mcF)$ is connected. 
     Moreover, since $H$ contains an edge $\{a_1,a_2\}$ (resp.\ $\{d_1,d_2\}$), 
     the two subtrees induced by
     $\nodes(\{a_1\},\mcF)$, $\nodes(\{a_2\},\mcF)$ (resp.\ 
     $\nodes(\{d_1\},\mcF)$, $\nodes(\{d_2\},\mcF)$) must be connected, hence
     $T'_a$ and $T'_d$ are subtrees of $T$.
     It remains to show that $T'_a$ and 
     $T'_d$ are disjoint. 

     Clearly, $u_A \in T'_a$ and $u_C \in T'_d$. We have established
     above that $u_B$ is on a path from $u_A$ to $u_C$. Also,
     $B_{u_B} \subseteq M \cup \{b_1, b_2, c_1, c_2\}$ and therefore
     $u_B$ is in neither $T'_a$ nor $T'_d$. Because both are subtrees
     of the tree $T$, there is only one path connecting them. But $u_B$
     is on that path and in neither subtree, i.e., $T'_a$ and $T'_d$ are disjoint.

As every edge
     must be covered, there are nodes in 
     $T'_a$ that cover $\{a_1,b_1\} \cup M_1$ and $\{a_2,b_2\} \cup 
     M_2$, respectively. 
Hence, the subtree $T'_a$ covers $M = M_1 \cup M_2$, 
i.e., 
     $M \subseteq \bigcup_{u\in T'_a} B_u$. Likewise, 
     $T'_d$ covers $M$. Since both subtrees are
     disjoint and $u_B$ is on the path between them, by the connectedness 
     condition, we have $M \subseteq B_{u_B}$.
\end{proof}

\newcommand{\thmNpcomp}{%
The \rec{{\it decomp\/},\,$k$} problem is \np-complete for 
{\it decomp\/} $\in \{$GHD, FHD$\}$ and $k = 2$.%
}

\begin{theorem}
 \label{thm:npcomp}
\thmNpcomp 
\end{theorem}

\newcommand{\lemComplEdge}{%
   Let $\mcF = \left< T, (B_u)_{u\in T}, (\gamma_u)_{u\in T} \right>$ be 
   an  FHD of width $\leq 2$ of the hypergraph
   $H$ constructed above. For every 
   node $u$ with $S \cup \{z_1, z_2\} \subseteq B_u$ and every pair $e, e'$ of
   complementary edges, it holds that $\gamma_u(e) = \gamma_u(e')$.
}

\newcommand{\lemCovering}{%
   Let $\mcF = \left< T, (B_u)_{u\in T}, (\gamma_u)_{u\in T} \right>$ be 
   an  FHD of width $\leq 2$ of the hypergraph
   $H$ constructed above and let $p \in [2n+3;m]^-$. For every 
   node $u$ with $S \cup A'_p \cup
   \overbar{A_p} \cup \{ z_1, z_2 \} \subseteq B_u$, 
the only way to cover $S 
   \cup A'_p \cup \overbar{A_p} \cup \{ z_1, z_2 \}$ 
by a fractional edge cover $\gamma$ of weight $\leq 2$ is 
by putting non-zero weight 
exclusively on edges 
   $e^{k,0}_p$ and $e^{k,1}_p$ with $k \in \{1,2,3\}$. 
Moreover, 
   $\sum_{k=1}^3 \gamma(e^{k,0}_p) = 1$ and $\sum_{k=1}^3 \gamma(e^{k,1}_p) = 1$  must hold.
}

\newcommand{\clmA}{
The nodes $u'_A,u'_B,u'_C$ (resp. $u_A,u_B,u_C$) are not 
on the path from $u_A$ to $u_C$ (resp. $u'_A$ to $u'_C$).
}

\newcommand{\clmB}{
The following equality holds:
$\nodes(A\cup A',\mcF) \cap \{u_A,u_B,u_C,u'_A,u'_B,u'_C\} = 
\emptyset$.

}

\newcommand{\clmC}{
The FHD $\mcF$ has a path containing nodes %
$\hat{u}_1, \dots, $ $\hat{u}_N$ for some $N$, such that the edges 
$e_{\min \ominus 1}, e_{\min}$, $e_{\min \oplus 1}$, \dots, $e_{\max \ominus 
1}$, $e_{\max}$
are covered in this order. More formally, there is a mapping $f: \{ 
\min\ominus 1,$ $\ldots, \max \} \ra \{1, \ldots, N\}$, s.t.
\begin{itemize}
 \item $\hat{u}_{f(p)}$ covers $e_p$ and
 \item if $p < p'$ then $f(p) \leq f(p')$.
\end{itemize}
By a {\em path containing nodes\/} $\hat{u}_1, \dots, \hat{u}_N$ we mean that 
$\hat{u}_1$ and $\hat{u}_N$ are nodes in $\mcF$, such that the nodes 
$\hat{u}_2, 
\dots, \hat{u}_{N-1}$ lie (in this order) on the path from $\hat{u}_1$ to  
$\hat{u}_N$. 
Of course, the path from $\hat{u}_1$ to  $\hat{u}_N$ may also contain further 
nodes, but we are not interested in whether they 
cover  any of the edges $e_p$.
}

\newcommand{\clmD}{
In the FHD $\mcF$ of $H$ of 
width $\leq 2$, the path from $u_A$ to $u'_A$ has non-empty intersection with
$\ppi$.
}

\newcommand{\clmE}{In the FHD $\mcF$ of $H$ of 
width $\leq 2$ there are two distinct nodes $\hat{u}$ and $\hat{u}'$ in the intersection 
of the path from $u_A$ to $u'_A$ with $\ppi$, s.t.\ $\hat{u}$ is the node in $\ppi$ closest to $u_A$
and $\hat{u}'$ is the node in $\ppi$ closest to $u'_A$. 
Then, on the path $\ppi$, $\hat{u}$ comes before $\hat{u}'$.
See Figure~\ref{fig:u-and-uprime} (a) for a graphical illustration of the arrangement of the 
nodes $\hat{u}_1$, $\hat{u}$, $\hat{u}'$, and 
$\hat{u}_N$
on the path $\ppi$.}

\newcommand{\clmF}{
In the FHD $\mcF$ of $H$ of 
width $\leq 2$ the path %
$\ppi$ has at least 3 nodes $\hat{u}_i$, i.e., $N \geq 
3$. 
}

\newcommand{\clmG}{In the FHD $\mcF$ of $H$ of 
width $\leq 2$ all the nodes $\hat{u}_2, \ldots, \hat{u}_{N-1}$ are on 
the path from $u_A$ to $u'_A$.
For the nodes $\hat{u}$ and $\hat{u}'$ from Claim E, this means that the nodes 
$\hat{u}_1, \hat{u},$ $\hat{u}_2,\hat{u}_{N-1}$, 
$ \hat{u}'$, $\hat{u}_N$ are arranged in precisely this order on the path $\ppi$
from $\hat{u}_1$ to $\hat{u}_N$, cf.\ Figure~\ref{fig:u-and-uprime} (b).
The node $\hat{u}$ may possibly coincide with $\hat{u}_1$    
and $\hat{u}'$ may possibly coincide with $\hat{u}_{N}$. 
}

\newcommand{\clmH}{
 Each of the nodes 
$\hat{u}_1, \dots, \hat{u}_N$ covers 
exactly one of the edges
$e_{\min \ominus 1}$, $e_{\min}$, $e_{\min \oplus 1}$, \dots, $e_{\max \ominus 
1}$, $e_{\max}$.
}

\newcommand{\clmI}{
The constructed truth assignment $\sigma$ %
is %
a model of $\varphi$.
}

\begin{proof}
The problem is clearly in \np: guess a tree decomposition and check in 
polynomial
 time for each node $u$ whether $\rho(B_u) \leq 2$ 
or $\rho^*(B_u) \leq 2$, respectively, holds.
The \np-hardness is proved by a reduction from 3SAT. 
Before presenting this reduction, we first 
introduce some 
useful notation.

\medskip
\noindent
{\bf Notation.}
  For $i,j \geq 1$, we denote $\{1,\ldots,i\} \times \{1,\ldots,j\}$ by $[i;j]$.
  For each $p \in [i;j]$, we denote by $p \oplus 1$ ($p \ominus 1$) 
  the 
  successor (predecessor) of $p$ in the 
  usual lexicographic order on pairs, that is, the order $(1,1),\ldots,(1,j),$ 
   $(2,1),\ldots,(i,1)$, $\ldots, (i,j)$. We 
  refer to the first element $(1,1)$ 
as $\min$ and 
  to the last element $(i,j)$ as $\max$. 
  We denote by $[i;j]^-$ the set $[i;j]\setminus\{\max\}$, i.e.\ $[i;j]$ without
  the last element.
  
Now let $\varphi = \bigwedge_{j=1}^m 
  (L_j^1 
  \vee L_j^2 \vee L_j^3)$ be an arbitrary instance of  3SAT with $m$ clauses 
and 
variables 
  $x_1,\ldots,x_n$. 
From this we  will
construct 
a hypergraph $H = (V(H),E(H))$,
which consists of two copies $H_0, H'_0$ of the 
  (sub-)hypergraph $H_0$ of Lemma~\ref{lem:gadgetH0} plus additional edges
  connecting $H_0$ and $H'_0$. We use the sets $Y = \{ y_1, \ldots, y_n\}$ 
  and $Y' = \{ y'_1, \ldots, y'_n\}$ to encode the truth values of the 
variables 
of $\varphi$.
  We  denote by $Y_{\ell}$ ($Y'_{\ell}$) the set $Y 
  \setminus \{y_{\ell}\}$ ($Y' \setminus \{ y'_{\ell} \}$). Furthermore, we use the sets 
  $A = \{ a_p \mid p \in [2n+3; m] \}$ and $A' = \{ a'_p \mid p \in 
  [2n+3;m]\}$, and
we define the following subsets of  $A$ and $A'$, respectively:
  \begin{align*}
     A_p &= \{ a_{\min},\ldots,a_p \} & 
  \overbar{A_p} &=  \{ a_{p},\ldots,a_{\max} \} \\
     A'_p &= \{ a'_{\min},\ldots,a'_p \} & 
  \overbar{A'_p} &= \{ a'_{p},\ldots,a'_{\max} \} 
  \end{align*}
  
  In addition,
  we will use another set $S$ of elements, that controls and restricts the ways
  in which edges are combined in a possible FHD or GHD. 
  Such a decomposition will have, implied by Lemma~\ref{lem:gadgetH0}, 
  two nodes $u_B$ and $u'_B$ 
  such that $S \subseteq 
  B_{u_B}$ and $S \subseteq B_{u'_B}$. From this, we will reason on the path 
  connecting $u_B$ and $u'_B$.
  
The concrete set $S$ used in our construction of  $H$ is obtained as follows.
Let $Q = [2n+3;m] \cup 
  \{(0,1),(0,0),(1,0)\}$, hence $Q$ is an extension of the set $[2n+3;m]$ with 
  special elements $(0,1),(0,0),(1,0)$. Then we define the set $S$ as 
  $ S = Q \times \{1,2,3\}.$ 

  The elements in $S$
  are pairs, which we denote as $(q \mid k)$. The values $q \in Q$ are themselves pairs of 
integers
  $(i,j)$. Intuitively, $q$ indicates the position of a node on the 
  ``long'' path $\pi$ in the desired FHD or GHD. The integer $k$ 
  refers to a literal in the $j$-th clause.
  We will write the wildcard $*$ to indicate that a component in some element 
of 
$S$
  can take an arbitrary value.  For example,
  $(\min \mid *)$ denotes the set of tuples 
  $(q \mid k)$
  where $q = \min = (1,1)$ and
  $k$ can take an arbitrary value in $\{1,2,3\}$.
  We will denote by $S_p$ the set $(p \mid *)$. For instance, 
  $(\min \mid *)$  will be denoted as $S_{\min}$. Further, for 
  $p \in [2n+3; m]$ and
  $k \in \{1,2,3\}$, we define singletons
  $S^{k}_p = \{ (p \mid k)\}$.    
  
\medskip
\noindent
{\bf Problem reduction.}
Let $\varphi = \bigwedge_{j=1}^m 
  (L_j^1 
  \vee L_j^2 \vee L_j^3)$ be an arbitrary instance of  3SAT with $m$ clauses 
and 
variables 
  $x_1,\ldots,x_n$. 
From this we 
construct a hypergraph $H = (V(H),E(H))$, 
that is,
an instance of \rec{{\it decomp\/},\,$k$} with 
{\it decomp\/} $\in  \{$GHD, FHD$\}$ and $k = 2$.

We start by defining the vertex set $V(H)$: 
  \begin{align*}
    V(H) = &\; S \;\cup\; A \;\cup\; A' \;\cup\;  Y \;\cup\; Y' 
\;\cup\; \{ 
z_1, 
z_2 \} \;\cup \\
    &\; \{ a_1, a_2, b_1, b_2, c_1, c_2, d_1, d_2, %
    a'_1, a'_2, b'_1, b'_2, c'_1, c'_2, d'_1, d'_2 \}.
  \end{align*}

The edges of $H$ are defined in 3 steps. First, we take two 
    copies of the subhypergraph $H_0$ used in Lemma~\ref{lem:gadgetH0}:
    
  \begin{itemize}
   \item Let $H_0 = (V_0, E_0)$ be the hypergraph of Lemma~\ref{lem:gadgetH0} 
   with $V_0=\{a_1,a_2,b_1,b_2$, $c_1,c_2,d_1,d_2\} \cup M_1 \cup M_2$ and 
   $E_0 = E_A \cup E_B \cup E_C$, where we set $M_1 = S \setminus 
   S_{(0,1)} \cup \{z_1\}$ and $M_2 = Y \cup 
   S_{(0,1)} \cup \{z_2\}$.
   \item Let $H'_0 = (V'_0, E'_0)$ be the corresponding hypergraph, 
         with $V'_0=\{a'_1, a'_2, b'_1,$ $b'_2, 
         c'_1, c'_2, d'_1, d'_2\}\cup M'_1 \cup M'_2$ and $E'_A, E'_B, E'_C$ 
are the primed versions of the 
         egde sets $M'_1 = S 
         \setminus S_{(1,0)} \cup \{z_1\}$ and $M'_2 = Y' \cup S_{(1,0)} \cup  
         \{z_2\}$.
   \end{itemize}
  
In the second step, we define the edges which (as we will see) enforce the 
existence of a ``long'' path~$\pi$
between the nodes covering $H_0$ and the nodes covering  $H'_0$ in any 
FHD of width $\leq 2$. 
  \begin{itemize}
   \item $e_{p} = A'_p \cup \overbar{A_p}$,
         for $p \in [2n+3;m]^-$,
   \item $e_{y_i} = \{ y_i, y'_i \}$, for $1 \leq i \leq n$,
   \item For $p = (i,j)  \in [2n+3;m]^-$ and $k \in \{1,2,3\}$:
     \begin{align*}
      e^{k,0}_p = & \begin{cases}
                    \overbar{A_p} \cup (S\setminus S^{k}_p) 
                       \cup Y  \cup \{z_1\} & \mbox{if } L^k_j = x_{\ell} \\
                    \overbar{A_p} \cup (S\setminus S^{k}_p)
                       \cup Y_{\ell} \cup \{z_1\} & \mbox{if } 
                       L^k_j = \neg x_{\ell}, 
                  \end{cases} \\
          e^{k,1}_p = & \begin{cases}
                    A'_p \cup S^{k}_p \cup 
                    Y'_{\ell} \cup \{ z_2\} & \mbox{if } L^k_j = x_{\ell} \\
                    A'_p \cup S^{k}_p \cup 
                    Y' \cup \{z_2\} & \mbox{if } L^k_j = \neg x_{\ell}. 
                  \end{cases}        
     \end{align*}
  \end{itemize}

Finally, we need edges that connect $H_0$ and $H'_0$ with the above edges covered 
by the nodes of the 
``long'' path $\pi$ in a GHD or FHD:
  
  \begin{itemize}
   \item $e^0_{(0,0)}= %
              \{ a_1 \} \cup A \cup S \setminus S_{(0,0)} \cup Y \cup \{ z_1\}
         $
   \item $e^1_{(0,0)} = S_{(0,0)} \cup Y' \cup \{ z_2\}$
   \item $e^0_{\max} = %
                S \setminus S_{\max} \cup Y
          \cup \{ z_1\}$
   \item $e^1_{\max} = \{a'_1 \} \cup A' \cup  S_{\max} \cup 
                Y' \cup \{z_2\}$
  \end{itemize}
  
This concludes the construction of the hypergraph $H$. Before we prove the correctness of the problem reduction, we give an example that will help to illustrate the intuition underlying this construction.

\begin{example}
\label{bsp:NPhardnessProof}
Suppose that an instance of 3SAT is given by the propositional 
formula $\varphi = (x_1 \vee \neg x_2 \vee x_3) \wedge 
(\neg x_1 \vee x_2 \vee \neg x_3)$, i.e.: we have $n= 3$ variables 
and $m = 2$ clauses. From this we construct a hypergraph 
$H = (V(H),E(H))$. 
First, we instantiate the sets $Q,A,A',S,Y$, and $Y'$ from our problem 
reduction.
  \begin{eqnarray*}
A & =  & \{ a_{(1,1)}, a_{(1,2)}, a_{(2,1)}, a_{(2,2)}, \dots, a_{(9,1)}, 
a_{(9,2)} \},
\\
A' & =  & \{ a'_{(1,1)}, a'_{(1,2)}, a'_{(2,1)}, a'_{(2,2)}, \dots, a'_{(9,1)}, 
a'_{(9,2)} \},
\\
Q &  = & \{(1,1), (1,2), (2,1), (2,2), \dots, (9,1), (9,2)\} \cup  \{(0,1),(0,0),(1,0)\}, 
\\
S & = &  Q \times \{1,2,3\}, %
\\
Y & =  & \{y_1, y_2, y_3\}, \\
Y' & = & \{y'_1, y'_2, y'_3\}.
  \end{eqnarray*}
According to our problem reduction, the set $V(H)$ of vertices of $H$ is 
  \begin{align*}
    V(H) = &\; S \;\cup\; A \;\cup\; A' \;\cup\; Y \;\cup\; Y' 
           \;\cup\; \{ z_1, z_2 \} \;\cup \\
    &\; \{ a_1, a_2, b_1, b_2, c_1, c_2, d_1, d_2 \} \cup \{ a'_1, a'_2, 
    b'_1, b'_2, c'_1, c'_2, d'_1, d'_2 \}.
  \end{align*}
The  set $E(H)$ of edges of $H$ is defined in several steps. First, the edges 
in 
$H_0$ and $H'_0$ are defined: We thus have the subsets 
$E_A,E_B,E_C,E'_A,$ $E'_B, E'_C \subseteq E(H)$, whose definition is based on the 
sets 
$M_1 = S \setminus S_{(0,1)} \cup \{z_1\}$,
$M_2 = Y \cup  S_{(0,1)} \cup \{z_2\}$,
$M'_1  =  S \setminus S_{(1,0)} \cup \{z_1\}$, 
and
$M'_2   =  Y' \cup S_{(1,0)} \cup \{z_2\}$.
The definition of the edges 
\begin{eqnarray*}
e_{p} & = & A'_p \cup \overbar{A_p} \hskip52pt \mbox{for } p \in \{(1,1), (1,2), \dots (8,1),(8,2),(9,1)\}, 
\\ 
e_{y_i} & = & \{ y_i, y'_i \} \hskip55pt \mbox{ for } 1 \leq i \leq 3,
\\
e^0_{(0,0)} & = & \{ a_1 \} \cup A \cup S \setminus S_{(0,0)} \cup Y \cup \{ 
z_1\} ,
\\
e^1_{(0,0)} & = & S_{(0,0)} \cup Y' \cup \{ z_2\}, 
\\
e^0_{(9,2)} & = & S \setminus S_{(9,2)} \cup Y \cup \{ z_1\}, \hskip2pt \mbox{ and}
\\
e^1_{(9,2)} & = & \{a'_1 \} \cup A' \cup  S_{(9,2)} \cup Y' \cup \{z_2\}
\end{eqnarray*}
is straightforward. We concentrate on the edges
$e^{k,0}_p$ and $e^{k,1}_p$ 
for $p \in \{(1,1), (1,2), \dots (8,1),(8,2)$, $(9,1)\}$ and $k \in \{1,2,3\}$.
These edges play the key role for covering the bags of the nodes 
along the ``long'' path $\pi$
in any FHD or GHD of $H$. This path can be thought of as being 
structured in 
9 blocks. Consider an arbitrary $i \in \{1, \dots, 9\}$.
Then $e^{k,0}_{(i,1)}$ and $e^{k,1}_{(i,1)}$ encode the $k$-th literal of the 
first clause
and $e^{k,0}_{(i,2)}$ and $e^{k,1}_{(i,2)}$ encode the $k$-th literal of the 
second clause 
(the latter is only defined for $i \leq 8$).
These edges are defined as follows:
the edges $e^{1,0}_{(i,1)}$ and $e^{1,1}_{(i,1)}$ encode the first literal of 
the first clause, i.e., 
the positive literal $x_1$. We thus have
  \begin{eqnarray*}
e^{1,0}_{(i,1)} &  = &  
\overbar{A_{(i,1)}} \cup (S\setminus S^{1}_{(i,1)}) \cup \{y_1,y_2,y_3\}  
\cup 
\{z_1\} \mbox{ and}
\\
e^{1,1}_{(i,1)}  &  = &  A'_{(i,1)} \cup S^{1}_{(i,1)} \cup 
\{y'_2,y'_3\} 
\cup 
\{ z_2\} 
  \end{eqnarray*}
The edges $e^{2,0}_{(i,1)}$ and $e^{2,1}_{(i,1)}$ encode the second literal of 
the first clause, i.e., 
the negative literal $\neg x_2$. Likewise, $e^{3,0}_{(i,1)}$ and 
$e^{3,1}_{(i,1)}$ encode the third literal of the first clause, i.e., 
the positive literal $x_3$. Hence,
  \begin{eqnarray*}
e^{2,0}_{(i,1)} &  = &  
\overbar{A_{(i,1)}} \cup (S\setminus S^{2}_{(i,1)}) \cup \{y_1,y_3\}  \cup 
\{z_1\}, 
\\
e^{2,1}_{(i,1)}  &  = &  A'_{(i,1)} \cup S^{2}_{(i,1)} \cup 
 \{y'_1,y'_2,y'_3\} 
\cup \{ z_2\} 
\\
e^{3,0}_{(i,1)} &  = &  
\overbar{A_{(i,1)}} \cup (S\setminus S^{3}_{(i,1)}) \cup \{y_1,y_2,y_3\}  
\cup 
\{z_1\}, 
\mbox{ and}
\\
e^{3,1}_{(i,1)}  &  = &  A'_{(i,1)} \cup S^{3}_{(i,1)} 
\cup \{y'_1,y'_2\} 
\cup   
\{ z_2\} 
  \end{eqnarray*}
Analogously, the 
edges $e^{1,0}_{(i,2)}$ and $e^{1,1}_{(i,2)}$ (encoding the 
first literal of the second clause, i.e., $\neg x_1$), 
the edges $e^{2,0}_{(i,2)}$ and $e^{2,1}_{(i,2)}$ (encoding the 
second literal of the second clause, i.e., $x_2$), and 
the edges $e^{3,0}_{(i,2)}$ and $e^{3,1}_{(i,2)}$ (encoding the 
third literal of the second clause, i.e., $\neg x_3$) are defined as follows:
  \begin{eqnarray*}
e^{1,0}_{(i,2)} &  = &  
\overbar{A_{(i,2)}} \cup (S\setminus S^{1}_{(i,2)}) \cup \{y_2,y_3\}  \cup 
\{z_1\}, \\
e^{1,1}_{(i,2)}  &  = &  A'_{(i,2)} \cup S^{1}_{(i,2)} \cup 
 \{y'_1,y'_2,y'_3\} 
\cup \{ z_2\}, 
\\ 
e^{2,0}_{(i,2)} &  = &  
\overbar{A_{(i,2)}} \cup (S\setminus S^{2}_{(i,2)}) \cup \{y_1,y_2,y_3\}  
\cup 
\{z_1\}, 
\\
e^{2,1}_{(i,2)}  &  = &  A'_{(i,2)} \cup S^{2}_{(i,2)} 
\cup \{y'_1,y'_3\} 
\cup 
\{ z_2\} 
\\
e^{3,0}_{(i,2)} &  = &  
\overbar{A_{(i,2)}} \cup (S\setminus S^{3}_{(i,2)}) \cup \{y_1,y_2\}  \cup 
\{z_1\}, 
\mbox{ and}
\\
e^{3,1}_{(i,2)}  &  = &  A'_{(i,2)} \cup S^{3}_{(i,2)} \cup 
\{y'_1,y'_2,y'_3\} 
\cup \{ z_2\}.
  \end{eqnarray*}
The crucial property of these pairs of edges 
$e^{k,0}_{(i,j)}$ and $e^{k,1}_{(i,j)}$ is that they together encode
the $k$-th literal of the $j$-th clause in the following way: 
if the literal is of the form $x_{\ell}$
(resp.\ of the form $\neg x_{\ell}$), 
then 
$e^{k,0}_{(i,j)} \cup e^{k,1}_{(i,j)}$ covers all of $Y \cup Y'$ except for 
$y'_{\ell}$ (resp.\ except for $y_{\ell}$).

Formula $\varphi$ in this example is clearly satisfiable, e.g., by the truth assignment $\sigma$ with 
$\sigma(x_1)=$ true and $\sigma(x_2) =\sigma(x_3)=$ false. Hence, for the 
problem
reduction to be correct, there must exist a GHD (and thus also an FHD) of width 
2 
of $H$. In Figure~\ref{fig:DecompPath}, the tree structure $T$ plus the 
bags $(B_t)_{t\in T}$ of such a GHD is displayed. Moreover, in   
Table~\ref{tab:np_decomp},
the precise definition of $B_u$ and $\lambda_u$ of every node $u\in T$ is given:
in the column labelled $B_u$, the set of vertices contained in $B_u$ for each node $u \in T$ is shown. 
In the column labelled $\lambda_u$, the two edges with weight 1 are shown. 
For the row with label $u_{p \in [2n+3;m]^-}$, the entry in the last column is 
$e^{k_p,0}_p, e^{k_p,1}_p$. By this we mean that, for every $p$, an appropriate value
$k_p \in \{1,2,3\}$ has to be determined. It will be explained below how to find an appropriate
value $k_p$ for each $p$.
The set $Z$ in the bags of this GHD is defined as 
$Z = \{y_i \mid \sigma(x_i) = $ true\,$\} \cup \{y'_i \mid \sigma(x_i) = $ 
false\,$\}$.
In this example, for the chosen truth assignment $\sigma$, 
we thus have $Z= \{y_1,y'_2,y'_3\}$.
The bags $B_t$ and the edge covers $\lambda_t$ for each $t\in T$ are explained 
below.

The nodes $u_C,u_B,u_A$ to cover the edges of the subhypergraph $H_0$ and the 
nodes
$u'_A,u'_B,u'_C$ to cover the edges of the subhypergraph $H'_0$ are clear by 
Lemma \ref{lem:gadgetH0}. The purpose of the nodes $u_{\min \ominus 1}$ and $u_{\max}$ is mainly to 
make sure that 
each edge $\{y_i,y'_i\}$ is covered by some bag. 
Recall that the set $Z$ contains exactly one of $y_i$ and $y'_i$ for every $i$. 
Hence, the node $u_{\min \ominus 1}$ (resp.\ $u_{\max}$) covers each edge $\{y_i,y'_i\}$, such that
$y'_i \in Z$ (resp.\ $y_i \in Z$). 

We now have a closer look at the nodes $u_{(1,1)}$ to $u_{(9,1)}$
on the ``long'' path $\pi$. More precisely, let us look at the nodes 
$u_{(i,1)}$ and $u_{(i,2)}$ for some $i \in \{1, \dots, 8\}$, i.e., the 
``$i$-th 
block''.
It will turn out that 
the bags at these nodes can be covered by edges from $H$ because $\varphi$ is 
satisfiable.
Indeed, our choice of $\lambda_{u_{(i,1)}}$ and $\lambda_{u_{(i,2)}}$ is guided 
by the literals 
satisfied by the truth assignment $\sigma$, namely: for $\lambda_{u_{(i,j)}}$, 
we have to choose
some $k_j$, such that the $k_j$-th literal in the $j$-th clause is true in 
$\sigma$. 
For instance, we may define 
$\lambda_{u_{(i,1)}}$ and $\lambda_{u_{(i,2)}}$ as follows: 
  \begin{align*}
\lambda_{u_{(i,1)}} &  =  \{ e^{1,0}_{(i,1)}, e^{1,1}_{(i,1)} \} &\lambda_{u_{(i,2)}} &  =  \{ e^{3,0}_{(i,2)}, e^{3,1}_{(i,2)} \} 
  \end{align*}
The covers $\lambda_{u_{(i,1)}}$ and $\lambda_{u_{(i,2)}}$ were chosen because 
the first literal of the first clause and the third literal of the second 
clause 
are true in $\sigma$. 
Now let us verify that
$\lambda_{u_{(i,1)}}$ and $\lambda_{u_{(i,2)}}$ are indeed covers
of $B_{u_{(i,1)}}$ and $B_{u_{(i,2)}}$, respectively.
By the definition of the 
edges $e^{k,0}_{(i,j)}, e^{k,1}_{(i,j)}$ for $j \in \{1,2\}$ and $k \in 
\{1,2,3\}$, it is 
immediate that 
$e^{k,0}_{(i,j)} \cup e^{k,1}_{(i,j)}$
covers 
$\overbar{A_{(i,j)}} \cup A'_{(i,j)} \cup S \cup \{z_1,z_2\}$.
The only non-trivial question is if $\lambda_{u_{(i,j)}}$ also covers $Z$.
Recall that by definition, 
$(e^{1,0}_{(i,1)} \cup e^{1,1}_{(i,1)}) \supseteq (Y \cup Y') \setminus 
\{y'_1\}$.
Our truth assignment $\sigma$ sets $\sigma(x_1) = $ true. Hence, by our 
definition of $Z$,
we have $y_1 \in Z$ and $y'_1 \not\in Z$. This means that 
$e^{1,0}_{(i,1)} \cup e^{1,1}_{(i,1)}$ indeed covers $Z$ and, hence, all of 
$B_{u_{(i,1)}}$. 
Note that we could have also chosen 
$\lambda_{u_{(i,1)}}   =   \{ e^{2,0}_{(i,1)}, e^{2,1}_{(i,1)} \}$, since 
also the second literal of the first clause (i.e., $\neg x_2$) is true in 
$\sigma$. In this case, 
we would have 
$(e^{2,0}_{(i,1)} \cup e^{2,1}_{(i,1)})
\supseteq (Y \cup Y') \setminus \{y_2\}$ and $Z$ indeed does not contain $y_2$.
Conversely, setting  
$\lambda_{u_{(i,1)}}   =   \{ e^{3,0}_{(i,1)}, e^{3,1}_{(i,1)} \}$ would fail, 
because
in this case, $y'_3 \not\in 
(e^{3,0}_{(i,1)} \cup e^{3,1}_{(i,1)})$ since $x_3$ occurs positively in the 
first clause. 
On the other hand, we have $y'_3 \in Z$  by definition of $Z$, 
because $\sigma(x_3) =$ false holds.

Checking that $\lambda_{u_{(i,2)}}$ as defined above covers $Z$ is done 
analogously. Note that in the second
clause, only the third literal is satisfied by $\sigma$. Hence, 
setting $\lambda_{u_{(i,2)}}   =   \{ e^{3,0}_{(i,2)}, $ $e^{3,1}_{(i,2)} \}$ 
is 
the 
only option to cover $B_{u_{(i,2)}}$ (in particular, to cover $Z$). 
Finally, note that $\sigma$ as defined above is not the only satisfying truth 
assignment of $\varphi$. For instance, we could have chosen 
$\sigma(x_1) = \sigma(x_2) = \sigma(x_3) =$ true. In this case, we would 
define $Z = \{ y_1,y_2,y_3\}$ and the covers $\lambda_{u_{(i,j)}}$ would have 
to 
be 
chosen according to an arbitrary choice of one literal per clause that is 
satisfied by 
this assignment $\sigma$.
\hfill$\Diamond$
\end{example}

  \begin{table*}
   \centering
   \begin{tabular}{|c|c|c|}
     \hline
     $u \in T$ & $B_u$ & $\lambda_u$ \\ 
     \hline 
     $u_C$ & $\{d_1, d_2, c_1, c_2\} \cup Y \cup S \cup \{z_1,z_2\}$ & 
$\{c_1,d_1\}\cup M_1$, 
                                                      $\{c_2,d_2\}\cup M_2$   \\
     $u_B$ & $\{c_1, c_2, b_1, b_2\} \cup Y \cup S \cup \{z_1,z_2\}$ & 
$\{b_1,c_1\}\cup M_1$, 
                                                      $\{b_2,c_2\}\cup M_2$  \\
     $u_A$ & $\{b_1, b_2, a_1, a_2\} \cup Y \cup S \cup \{z_1,z_2\}$ & 
$\{a_1,b_1\}\cup M_1$, 
                                                      $\{a_2,b_2\}\cup M_2$  \\
     $u_{\min\ominus 1}$ & $\{a_1 \} \cup A \cup Y \cup S \cup Z 
\cup \{z_1,z_2\}$ &
           $e^0_{(0,0)}, e^1_{(0,0)}$  \\
     $u_{p \in [2n+3;m]^-}$ & $A'_p \cup \overbar{A_p} \cup 
           S \cup Z \cup \{z_1,z_2\}$ & $e^{k_p,0}_p, e^{k_p,1}_p$  \\
     $u_{\max}$ & $\{ a'_1\} \cup A' \cup Y' \cup S \cup Z \cup \{z_1,z_2\}$ & 
           $e^0_{\max}, e^1_{\max}$  \\
     $u'_A$ & $\{a'_1, a'_2, b'_1, b'_2\} \cup Y' \cup S \cup \{z_1,z_2\}$ & 
$\{a'_1,b'_1\} \cup 
                                              M_1'$, $\{a'_2,b'_2\}\cup M'_2$  
\\
     $u'_B$ & $\{b'_1, b'_2, c'_1, c'_2\} \cup Y' \cup S \cup \{z_1,z_2\}$ & 
$\{b'_1,c'_1\} \cup 
                                              M_1'$, $\{b'_2,c'_2\}\cup M'_2$  
\\
     $u'_C$ & $\{c'_1, c'_2, d'_1, d'_2\} \cup Y' \cup S \cup \{z_1,z_2\}$ & 
$\{c'_1,d'_1\} \cup 
                                              M_1'$, $\{c'_2,d'_2\}\cup M'_2$  
\\
     \hline
   \end{tabular}
   \vspace{3mm}
   \caption{Definition of $B_u$ and $\lambda_u$ for GHD of $H$.}
   \label{tab:np_decomp} 
  \end{table*}
  \begin{figure*}
    \centering
    \includegraphics[width=\textwidth]{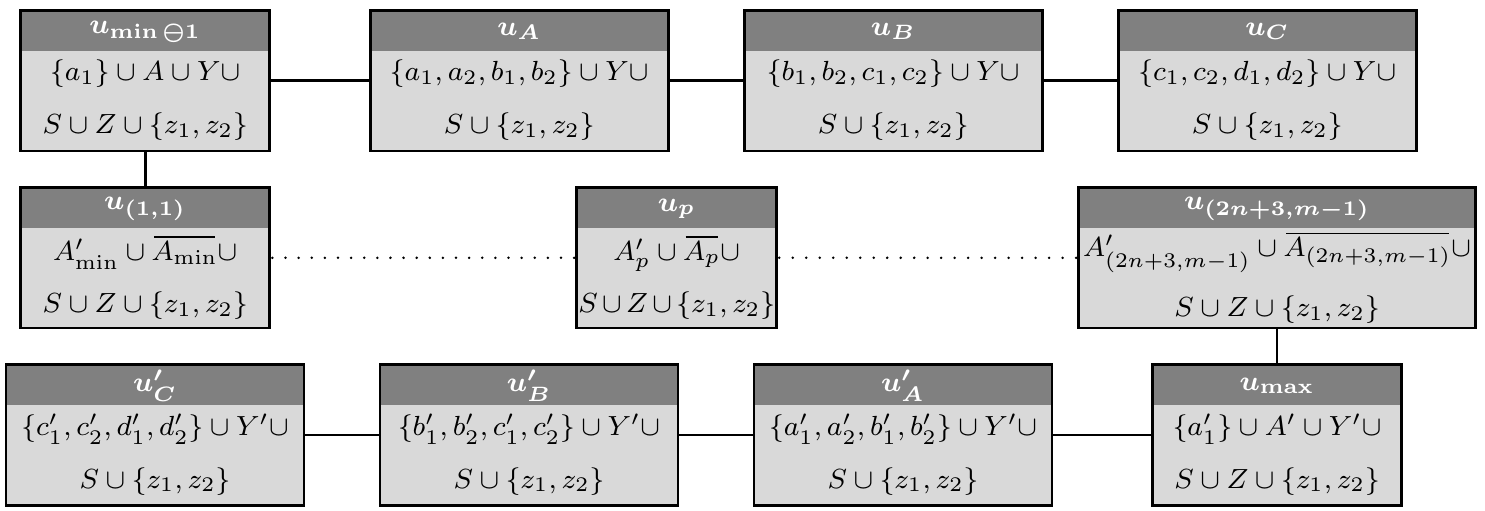}
    \caption{Intended path of the GHD of hypergraph $H$ in the proof of 
             Theorem~\ref{thm:npcomp}} 
    \label{fig:DecompPath}
  \end{figure*}

To prove the correctness of our problem reduction, we have to show the two equivalences: first, that 
$\ghw(H) \leq 2$ if and only if 
$\varphi$ is satisfiable and second, that
$\fhw(H) \leq 2$ 
if 
and only if $\varphi$ is satisfiable. We prove the two directions of these equivalences 
separately.

  \medskip
 \noindent
 {\bf Proof of the ``if''-direction.}
  First assume that $\varphi$ is satisfiable. It suffices to show that 
then 
  $H$ has a GHD of width $\leq 2$, because $\fhw(H) \leq \ghw(H)$ holds. 
  Let $\sigma$ be a 
  satisfying truth assignment. Let us fix for each $j \leq m$, some $k_j \in
  \{ 1,2,3 \}$ such that $\sigma(L^{k_j}_j) = 1$. By ${\ell}_j$, we denote the index 
of the variable in the literal $L^{k_j}_j$, that is, $L^{k_j}_j = x_{{\ell}_j}$ or 
$L^{k_j}_j = 
\neg x_{{\ell}_j}$. For $p=(i,j)$, let $k_p$ refer to $k_j$ and let $L^{k_p}_p$ 
refer to $L^{k_j}_j$. Finally, we define $Z$ as
  $Z = \{ y_i \mid \sigma(x_i) = 1 \} \cup \{ y'_i \mid \sigma(x_i) = 0 \}$.
  
  A GHD  $\mcG = \left< T, (B_u)_{u\in T}, (\lambda_u)_{u\in T} \right>$ of 
width 2 
for $H$ is 
  constructed as follows. $T$ is a path $u_C$, $u_B$, $u_A$, 
$u_{\min\ominus1}$, 
  $u_{\min}$,\dots, $u_{\max}$, $u'_A$, $u'_B$, $u'_C$. The 
construction is 
  illustrated in 
  Figure~\ref{fig:DecompPath}.
  The precise definition of $B_u$ and $\lambda_u$ is given in
  Table~\ref{tab:np_decomp}. Clearly, the GHD  has width $\leq 2$. We 
  now show that $\mcG$ is 
indeed a GHD of $H$:
  \begin{enumerate}
   \item[(1)] For each edge $e \in E$, there is a node $u \in T$, such that 
         $e \subseteq B_u$:
         \begin{itemize}
          \item $\forall e \in E_X : e \subseteq B_{u_X}$ for all $X \in 
\{A,B,C\}$, 
          \item $\forall e' \in E'_X : e' \subseteq B_{u'_X}$ for all $X \in 
\{A,B,C\}$, 
          \item $e_p \subseteq B_{u_p}$ for $p \in [2n+3;m]^-$,
          \item $e_{y_i} \subseteq B_{u_{\min\ominus 1}}$ 
                (if $y'_i \in Z$) or 
                $e_{y_i} \subseteq B_{u_{\max}}$ 
                (if $y_i \in Z$), respectively,
          \item $e^{k,0}_p \subseteq B_{u_{\min\ominus 1}}$ 
                for $p \in [2n+3;m]^-$,
          \item $e^{k,1}_p \subseteq B_{u_{\max}}$ 
                for $p \in [2n+3;m]^-$,
          \item $e^0_{(0,0)} \subseteq B_{u_{\min\ominus 1}}$, 
                $e^1_{(0,0)} \subseteq B_{u_{\max}}$,
          \item $e^0_{\max} \subseteq B_{u_{\min\ominus 1}}$ and 
                $e^1_{\max} \subseteq B_{u_{\max}}$.
         \end{itemize}
         All of the above inclusions can be 
         verified in 
         Table~\ref{tab:np_decomp}. 
        
   \item[(2)] For each vertex $v \in V$, the set $\{ u\in T \mid v \in 
         B_u \}$ induces a connected subtree of $T$, which %
         is easy to 
verify in Table~\ref{tab:np_decomp}.
         
   \item[(3)] For each $u \in T$, $B_u \subseteq B(\lambda_u)$:
         the only inclusion which cannot be easily verified in 
         Table~\ref{tab:np_decomp} is $B_{u_p} \subseteq 
         B(\lambda_{u_p})$. In fact, this is the 
         only place in the proof where we
         make use of the assumption that $\varphi$ is satisfiable.
         First, notice that the set $A'_p \cup \overbar{A_p} \cup
         S \cup \{z_1,z_2\}$ is clearly a subset of $B(\lambda_{u_p})$. It 
remains to
         show that $Z \subseteq B(\lambda_{u_p})$ holds for 
         arbitrary $p \in [2n+3;m]^-$. We show this property by a case distinction
         on the form of $L^{k_p}_p$.
         
         \smallskip
         
         Case (1): First, assume that $L^{k_p}_p = x_{{\ell}_j}$ holds. Then
         $\sigma(x_{{\ell}_j}) = 1$ and, therefore, $y'_{{\ell}_j} \not\in Z$. But, by
         definition of $e^{k_p,0}_p$ and $e^{k_p,1}_p$, vertex $y'_{{\ell}_j}$ is the 
only 
         element of $Y \cup Y'$ not contained in $B(\lambda_{u_p})$. 
         Since $Z \subseteq (Y \cup Y')$ and $y'_{{\ell}_j} \not\in Z$, we have that
         $Z \subseteq B(\lambda_{u_p})$. 
         
         \smallskip

         Case (2): Now assume that $L^{k_p}_p = \neg x_{{\ell}_j}$ holds. Then
         $\sigma(x_{{\ell}_j}) = 0$ and, therefore, $y_{{\ell}_j} \not\in Z$. But, by
         definition of $e^{k_p,0}_p$ and $e^{k_p,1}_p$, vertex $y_{{\ell}_j}$ is the 
only 
         element of $Y \cup Y'$ not contained in $B(\lambda_{u_p})$. 
         Since $Z \subseteq (Y \cup Y')$ and $y_{{\ell}_j} \not\in Z$, we have that
         $Z \subseteq B(\lambda_{u_p})$.
  \end{enumerate}

 \smallskip
 \noindent
 {\bf  Two crucial lemmas.}
  Before we prove the ``only if'-direction, 
  we define the notion of complementary 
edges and state two important lemmas related to this notion.

  \begin{definition}
  \label{def:complementary}
   Let $e$ and $e'$ be two edges from the hypergraph $H$ as defined 
before. We say $e'$ is the {\em complementary} edge of $e$ (or, simply, 
$e,e'$ are complementary edges) whenever
   \begin{itemize}
    \item $e \cap S  =  S \setminus S'$ for some $S' \subseteq S$ and
    \item $e' \cap S =  S'$.
   \end{itemize}
  \end{definition}
  
  Observe that for every edge in our construction that covers 
  $S\setminus S'$ for some $S' \subseteq S$ there is a complementary 
edge that covers $S'$, for example 
$e^{k,0}_p$ and $e^{k,1}_p$, $e^0_{(0,0)}$ and $e^1_{(0,0)}$, and so on. In 
particular %
there is no edge that covers $S$ completely. 
Moreover, consider arbitrary subsets $S_1,S_2$ of $S$, 
s.t.\ (syntactically)
$S \setminus S_i$ is part of the definition of $e_i$  for some $e_i \in E(H)$ 
with $i \in \{1,2\}$.
Then $S_1$ and $S_2$ are~disjoint. 

We now present two lemmas needed for the ``only if''-direction.

\begin{lemma}
\label{lem:compl_edge} 
\lemComplEdge
\end{lemma}
\begin{proof}
     First, we try to cover $z_1$ and $z_2$. For $z_1$ we have to put total 
     weight~$1$ on the edges in $E^0$, and to cover 
     $z_2$ we have to put total  weight~$1$ on the edges in $E^1$, where
     \begin{align*}
      E^0 =&   \{ e^{k,0}_p \mid p\in [2n+3;m]^- \mbox{ and } 1 \leq k 
                \leq 3\} \; \cup \\ 
             & \{ e^0_{(0,0)},e^0_{\max} \} \;\cup \\
             & \{\{ a_1,b_1 \} \cup M_1, \{ b_1,c_1\}\cup M_1, \{c_1,d_1\}  
                \cup M_1 \} \;\cup \\
            & \{\{ a'_1,b'_1 \} \cup M'_1, \{ b'_1,c'_1\}\cup M'_1,
                  \{c'_1,d'_1\} \cup M'_1 \} \\
     E^1 = &  \{ e^{k,1}_p \mid p\in [2n+3;m]^- \mbox{ and } 1 \leq k 
               \leq 3\} \; \cup  \\
            & \{ e^1_{(0,0)}, e^1_{\max} \} \; \cup \\
            & \{\{ a_2,b_2 \} \cup M_2, \{ b_2,c_2\}\cup M_2, \{c_2,d_2\}
            \cup M_2 \} \;\cup \\
            & \{ \{ a'_2,b'_2 \} \cup M'_2, \{ b'_2,c'_2\}\cup M'_2,
                  \{c'_2,d'_2\}\cup M'_2 \}
     \end{align*}
     In order to also cover $S$ with weight 1, we are only allowed to 
     assign weights to the above edges.
     Let $e_i^0$ be an arbitrary edge in $E^0$ with $\gamma_u(e^0_i) = w_i > 0$.
     Then there exists a subset $S_i$ of $S$ with $S_i \neq \emptyset$, s.t.\ $e^0_i \cap S = S \setminus S_i$.
     Still, we need to put weight $1$ on the vertices
     in $S_i$. In order to do so, we can put at most weight $1-w_i$ on the 
     edges in $E^0 \setminus \{ e^0_i\}$, which covers $S_i$ with weight at most 
     $1-w_i$. Hence, the edges in $E^1$ have to put weight $\geq w_i$ on $S_i$.
     The only edge in $E^1$ that intersects $S_i$ is the complementary edge 
     $e^1_i$ of $e^0_i$. Hence, we have to set $\gamma_u(e^1_i) \geq w_i$. In other words, 
     we have to set $\gamma_u(e^1_i) = w_i + \epsilon_i$ for some $\epsilon_i \geq 0$. 
     Note that $e_i^0$ was arbitrarily chosen. Hence, we can conclude for every 
     $e_i^0 \in E^0$ with $\gamma_u(e^0_i) = w_i > 0$ and with complementary edge $e_i^1 \in E^1$ 
     that 
     $\gamma_u(e^1_i) = w_i + \epsilon_i$ for some $\epsilon_i \geq 0$. 
     Hence, $\sum_{e^1\in E^1} \gamma_u(e^1) = \sum_{e^0\in E^0} \gamma_u(e^0) + \sum_i \epsilon_i$.
     Now recall that both $\sum_{e^0\in E^0} \gamma_u(e^0) = 1$  and    
     $\sum_{e^1\in E^1} \gamma_u(e^1) = 1$ hold. This is only possible if $\epsilon_i = 0$ for every $i$. 
     In other words, $\gamma_u(e_i^0) = 
     \gamma_u(e_i^1) = w_i$ for every $e_i^0 \in E^0$ and its complementary 
     edge $e_i^1 \in E^1$. 
  \end{proof}
\begin{lemma}
\label{lem:covering} 
\lemCovering
\end{lemma}
\begin{proof}
     As in the proof of Lemma~\ref{lem:compl_edge}, to cover $z_1$ we 
have to put 
     weight $1$ on the edges in $E^0$ and to cover 
     $z_2$ we have to put weight $1$ on the edges in $E^1$, where $E^0$ and $E^1$ 
     are defined as in the proof of Lemma~\ref{lem:compl_edge}. Since we have 
     $\width(\mcF) \leq 2$, we have to cover $A'_p \cup \overbar{A_p} \cup S$
     with the weight already put on the edges in $E^0 \cup E^1$. 
     In order to cover 
     $A'_p$, we have to put weight 1 on the edges in $E^1_p$, where
     \[ E^1_p = \{ e^{k,1}_r \mid r \geq p \} \cup \{ e^1_{\max} \}. \]
     Notice that, $E^1_p \subseteq E^1$ and 
     therefore $\sum_{e \in E^1 \setminus E^1_p} \gamma_u(e) = 0$. 
     Similar,
     in order to cover $\overbar{A_p}$, we have to put weight 1 on the edges in
     $E^0_p$, where 
     \[ E^0_p = \{ e^{k,0}_s \mid s \leq p \} \cup  \{ e^0_{(0,0)} \}.\] 
     Again, since 
     $E^0_p \subseteq E^0$,
     $\sum_{e \in E^0 \setminus E^0_p} \gamma_u(e) = 0$. It remains to cover 
     $S \cup \{z_1,z_2\}$. By Lemma~\ref{lem:compl_edge}, in order to cover 
$S$, 
$z_1$ and 
     $z_2$, we have to put the same weight $w$ on complementary 
     edges $e$ and $e'$. The only complementary edges in the sets $E^0_p$ and
     $E^1_p$ are edges of the form $e^{k,0}_p$ and $e^{k,1}_p$ with $k \in 
\{1,2,3\}$. 
 In total, we thus have
     $\sum_{k=1}^3 e^{k,0}_p = 1$ and $\sum_{k=1}^3 e^{k,1}_p = 1$.
  \end{proof}

 \noindent
 {\bf Proof of the ``only if''-direction.}
  It remains to show that $\varphi$ is satisfiable if $H$ has a GHD or FHD of 
  width $\leq 2$. Due to the inequality $\fhw(H) \leq \ghw(H)$,
  it suffices to show that $\varphi$ is satisfiable if $H$ 
  has an FHD of  width $\leq 2$.
  For this, let $\mcF = \left< T, (B_u)_{u\in T}, (\gamma_u)_{u\in T} 
  \right>$ be such an 
  FHD. Let $u_A, u_B, u_C$ and $u'_A, u'_B, u'_C$ be the nodes that are
  guaranteed by Lemma~\ref{lem:gadgetH0}.
  We state several %
  properties of the path connecting $u_A$ and $u'_A$, which heavily 
rely on Lemmas~\ref{lem:compl_edge} and~\ref{lem:covering}. 

\medskip
{\sc  Claim A.}{\it \clmA}

\medskip
{\sc Proof of Claim A.}
We only show that none of the nodes $u'_i$ with $i \in \{A,B,C\}$ is on the 
path from $u_A$ to $u_C$. 
The other 
property is shown analogously. Suppose to the contrary that some $u'_i$ is on 
the path from $u_A$ to $u_C$. Since $u_B$ is also on the path between $u_A$ and 
$u_C$ we distinguish two cases:

\begin{itemize}
\item Case~(1): $u'_i$ is on the path between $u_A$ 
and $u_B$; then $\{ b_1, b_2 \} \subseteq B_{u'_i}$. This 
contradicts the property shown in Lemma \ref{lem:gadgetH0} that $u'_i$ cannot cover any
vertices outside $H'_0$. 

\item Case~(2): $u'_i$ is on the path between $u_B$ and $u_C$;
then $\{ c_1, c_2 \} \subseteq B_{u'_i}$, which again contradicts Lemma \ref{lem:gadgetH0}.
\end{itemize}
Hence, the paths from $u_A$ to $u_C$ and from $u'_A$ to $u'_C$ are indeed disjoint. $\hfill\diamond$

\medskip
{\sc  Claim B.}{\it \clmB}

\medskip
{\sc Proof of Claim B.}
Suppose to the contrary that there is a $u_X$ (the proof for $u'_X$ is analogous) 
for some $X \in 
\{A,B,C\}$, s.t.\ $u_X \in \nodes(A\cup A',\mcF)$; then
there is some $a \in (A\cup A')$, s.t.\ $a \in B_{u_X}$. This contradicts 
the property shown in Lemma \ref{lem:gadgetH0} that $u_X$ cannot cover any
vertices outside $H_0$. $\hfill \diamond$

\medskip
We are now interested in the sequence of nodes $\hat{u}_i$ that cover the 
edges 
$e^0_{(0,0)}, e_{\min}, e_{\min \oplus 1}$, \dots,
$e_{\max \ominus 1}$, $e_{\max}$. 
Before we formulate Claim~C, 
it is convenient to introduce the following notation. To be able to refer to 
the 
edges $e^0_{(0,0)}$, $e_{\min}$, $e_{\min \oplus 1}$, \dots, $e_{\max\ominus 
1}$, $e^1_{\max}$ in a uniform way,
we use $e_{\min \ominus 1}$ as synonym of $e^0_{(0,0)}$ and 
$e_{\max}$ as synonym of $e^1_{\max}$. We can thus define the natural order 
$e_{\min \ominus 1} < e_{\min} < e_{\min \oplus 1} < \dots < e_{\max \ominus 1} 
< e_{\max}$ on these edges.

\medskip
{\sc Claim C.}{\it \clmC}

\medskip
{\sc Proof of Claim C.}
Suppose to the contrary that no such path exists. Let $p \geq \min$ be the maximal value such 
that there is a path containing nodes $\hat{u}_1, \hat{u}_2, \ldots, 
\hat{u}_{\ell}$, which cover 
$e_{\min\ominus 1}, \ldots, e_p$ in this order. 
Clearly, there exists a node $\hat{u}$ that covers $e_{p \oplus 1} = A'_{p 
\oplus 1} \cup \overbar{A_{p\oplus 1}}$. We distinguish 
four cases:

\begin{itemize}
 \item Case~(1): $\hat{u}_1$ is on the path from $\hat{u}$ to all other nodes
                 $\hat{u}_i$, with $1 < i \leq \ell$. By the connectedness 
                 condition, $\hat{u}_1$ covers $A'_p$.
                 Hence, in total $\hat{u}_1$ covers $A'_p \cup 
                 A$ with 
                 $A'_{p} = \{a'_{\min}, \dots, a'_{p}\}$ and 
                 $A = \{a_{\min}, \dots, a_{\max}\}$.
                 Then 
                 $\hat{u}_1$ covers all edges $e_{\min\ominus 1}, \ldots, e_p$. 
Therefore, 
                 the path containing nodes $\hat{u}_1$ and $\hat{u}$ covers 
                 $e_{\min\ominus 1}, \ldots,$ $e_{p\oplus 1}$ in this order,
                 which contradicts
                 the maximality of $p$.
  \item Case~(2): $\hat{u} = \hat{u}_1$,
                 hence, $\hat{u}_1$ covers $A'_{p\oplus 1} \cup 
                 A$ with 
                 $A'_{p\oplus 1} = \{a'_{\min}, \dots, a'_{p\oplus 1}\}$ and 
                 $A = \{a_{\min}, \dots$, $a_{\max}\}$.
                 Then, 
                 $\hat{u}_1$ covers all %
                 $e_{\min\ominus 1}, \ldots, 
                 e_{p\oplus 1}$,
                 which contradicts
                 the maximality of $p$.

 \item Case~(3): $\hat{u}$ is on the path from $\hat{u}_1$ to $\hat{u}_{\ell}$
           and $\hat{u} \neq \hat{u}_1$. 
                 Hence, $\hat{u}$ is between two nodes $\hat{u}_{i}$ and 
                 $\hat{u}_{i+1}$ for some $1 \leq i < \ell$ or 
                 $\hat{u} = \hat{u}_{i+1}$ for some $1 \leq i < \ell-1$.
                 The following
                 arguments hold for both cases.
                 Now, there is some
                 $q \leq p$, such that $e_{q}$ is covered by $\hat{u}_{i+1}$
                 and $e_{q \ominus 1}$ is covered by $\hat{u}_{i}$.
                 Therefore,  $\hat{u}$ covers  $\overbar{A_q}$
                 either by the connectedness condition (if $\hat{u}$ is between 
$\hat{u}_{i}$ and $\hat{u}_{i+1}$)
                 or simply because $\hat{u} = \hat{u}_{i+1}$.
                 Hence, in total, $\hat{u}$
                 covers $A'_{p\oplus 1} \cup \overbar{A_q}$ with 
                 $A'_{p\oplus 1} = \{a'_{\min}, \dots, a'_{p \oplus 1}\}$ and 
                 $\overbar{A_q} = \{a_q, a_{q\oplus 1}, \dots, a_p, a_{p\oplus 
1}, \dots a_{\max}\}$.
                 Then, $\hat{u}$ covers all edges $e_q, e_{q\oplus 1}, 
                 \ldots, e_{p\oplus 1}$. Therefore, the path containing nodes
                 $\hat{u}_1, \dots, \hat{u}_i, \hat{u}$
                 covers $e_{\min\ominus 1}, \ldots, e_{p\oplus 1}$ in this 
order,
                 which contradicts
                 the maximality of $p$.
 \item Case~(4): There is a $u^*$ on the path from $\hat{u}_1$ to $\hat{u}_{\ell}$, 
                 such that the paths from $\hat{u}_1$ to $\hat{u}$ and from 
                 $\hat{u}$ to $\hat{u}_{\ell}$ go 
                 through $u^*$ and, moreover, $u^* \neq \hat{u}_1$.                 
                 Then, $u^*$ is either between %
                 $\hat{u}_i$ and $\hat{u}_{i+1}$ for some $1 \leq i < \ell$
                 or
                 $u^* = \hat{u}_{i+1}$ for some $1 \leq i < \ell-1$. The following
                 arguments hold for both cases. There is some $q \leq p$, 
                 such that
                 $e_q$ is covered by $\hat{u}_{i+1}$ and $e_{q\ominus 1}$ is
                 covered by $\hat{u}_i$. By the connectedness condition, 
                 $u^*$ covers 
                 \begin{itemize}
                  \item $A'_p = \{ a'_{\min},\ldots,a'_p \}$, since $u^*$ is on 
                        the path from $\hat{u}$ to $\hat{u}_{\ell}$, and
                  \item $\overbar{A}_q = \{a_q,  \dots, a_p, 
                        a_{p\oplus 1}, \dots a_{\max}\}$, since $u^*$ is on
                        the path from $\hat{u}_1$ to $\hat{u}_{i+1}$ or 
                        $u^* = \hat{u}_{i+1}$.
                 \end{itemize}
                 Then $u^*$ covers all edges $e_q,e_{q\oplus 
                 1},\ldots,e_p$. Therefore, the path containing the nodes 
                 $\hat{u}_1,\ldots,\hat{u}_{i}$, $u^*$, $\hat{u}$ 
                 covers $e_{\min\ominus 1}, 
                 \ldots, e_{p\oplus 1}$ in this order,
                 which contradicts
                 the maximality of~$p$.                
                 $\hfill \diamond$

\end{itemize}

So far we have shown, that there are three disjoint paths from $u_A$ to $u_C$, 
from $u'_A$ to $u'_C$ and from $\hat{u}_1$ to $\hat{u}_N$, respectively. It is 
easy to see, 
that $u_A$ 
is closer to the path $\hat{u}_1$, \dots, $\hat{u}_N$ than $u_B$ and $u_C$, 
since otherwise 
$u_B$ and $u_C$ would have to cover $a_1$ as well, which is impossible 
by Lemma~\ref{lem:gadgetH0}. 
The same also
holds for $u'_A$. In the next claims we will argue that the path from $u_A$
to $u'_A$ goes through some node $\hat{u}$ of the path from $\hat{u}_1$ to 
$\hat{u}_N$. We write $\ppi$ as a short-hand notation for 
the path 
from $\hat{u}_1$ to $\hat{u}_N$. Next, we state some important properties of 
$\ppi$ and the path from $u_A$ to $u_A'$.

\medskip
{\sc  Claim D.} 
{\it \clmD}

\medskip
{\sc Proof of Claim D.} Suppose to the contrary that the path from $u_A$ to 
$u'_A$ is disjoint from 
$\ppi$. We distinguish three cases:
\begin{itemize}
 \item Case~(1):
$u_A$ is on the path from $u'_A$ to (some node in) $\ppi$. Then, by the 
connectedness condition, $u_A$ must contain $a'_1$, which 
contradicts Lemma~\ref{lem:gadgetH0}.
 \item Case~(2): $u'_A$ is on the path from $u_A$ to $\ppi$. 
Analogously to Case (1), we get a contradiction by the fact that then 
$u'_A$ must contain $a_1$.
\item Case~(3): There is a node $u^*$ on the path from $u_A$ to $u'_A$, which 
is closest to $\ppi$, i.e., $u^*$ lies on the path from $u_A$ to $u'_A$ and 
both 
paths, the one
connecting $u_A$ with $\ppi$   and the one connecting $u'_A$ with $\ppi$, go 
through $u^*$. 
Hence, by the connectedness condition, the bag of $u^*$ contains $S \cup 
\{z_1,z_2,$ $a_1,a'_1\}$. By Lemma~\ref{lem:compl_edge}, 
in order to cover $S \cup \{ z_1, z_2 \}$ with weight $\leq 2$,
we are only allowed to put 
non-zero weight on pairs of complementary edges.
However, then it is impossible to achieve also weight $\geq 1$ on 
$a_1$ and $a'_1$ at the same time. $\hfill \diamond$
\end{itemize}

\begin{figure}[t]
    \centering 
  \footnotesize
    
    \begin{minipage}[c]{0.48\textwidth}
       \centering
       \begin{tikzpicture}[
   vert/.style={fill, circle, inner sep = 1pt},
   ell/.style={ellipse,draw,minimum width=2.5cm, inner sep=0cm}]
   \node[vert,label={[align=center]above:$\hat{u}_1$: \\ $\{ a_1 \} \subseteq B_{\hat{u}_1}$}] (u1) {};
   \node[vert,label=above:$\hat{u}$, right=1 of u1] (u) {};
   \node[vert,label={[align=center]below:$u_A$: \\ $S \cup \{ z_1, z_2, a_1 \} \subseteq B_{u_A}$},below right=.5 and .5 of u1] (ua) {};
   \node[vert,label=above:$\hat{u}'$, right=2 of u] (up) {};
   \node[vert,label={[align=center]above:$\hat{u}_N$: \\ $\{ a'_1 \} \subseteq B_{\hat{u}_N}$}, right=1 of up] (un) {};
   \node[vert,label={[align=center]below:$u'_A$: \\ $S \cup \{ z_1, z_2, a'_1 \} \subseteq B_{u'_A}$},below left=.5 and .5 of un] (uap) {};
   
   \draw[densely dotted, thick] (u1) -- (u);
   \draw[densely dotted, thick] (u) -- (up);
   \draw[densely dotted, thick] (up) -- (un);
   \draw[densely dotted] (ua) -- (u);
   \draw[densely dotted] (up) -- (uap);

\end{tikzpicture}

     \centering\medskip\noindent (a)
     \end{minipage}
     \begin{minipage}[c]{0.48\textwidth}
       \centering
       \begin{tikzpicture}[
   vert/.style={fill, circle, inner sep = 1pt},
   ell/.style={ellipse,draw,minimum width=2.5cm, inner sep=0cm}]
   \node[vert,label={[align=center]above:$\hat{u}_1$: \\ $\{ a_1 \} \subseteq B_{\hat{u}_1}$}] (u1) {};
   \node[vert,label=above:$\hat{u}$, right=1 of u1] (u) {};
   \node[vert,label={[align=center]below:$u_A$: \\ $S \cup \{ z_1, z_2, a_1 \} \subseteq B_{u_A}$},below right=.5 and 0.5 of u1] (ua) {};
   \node[vert,label=above:$\hat{u}'$, right=2 of u] (up) {};
   \node[vert,label={[align=center]above:$\hat{u}_N$: \\ $\{ a'_1 \} \subseteq B_{\hat{u}_N}$}, right=1 of up] (un) {};
   \node[vert,label={[align=center]below:$u'_A$: \\ $S \cup \{ z_1, z_2, a'_1 \} \subseteq B_{u'_A}$},below left=.5 and .5 of un] (uap) {};
   
   \node[vert,label=below:$\hat{u}_2$, right=.5 of u] (u2) {};
   \node[vert,label=below:$\hat{u}_{N-1}$, right=1 of u2] (un1) {};

   \draw[densely dotted, thick] (u1) -- (u);
   \draw[densely dotted, thick] (u) -- (up);
   \draw[densely dotted, thick] (up) -- (un);
   \draw[densely dotted] (ua) -- (u);
   \draw[densely dotted] (up) -- (uap);
\end{tikzpicture}

     \centering\medskip\noindent (b)
     \end{minipage}

    \caption{Arrangement of the nodes 
    $\hat{u}_1$, $\hat{u}$, $\hat{u}'$, and $\hat{u}_N$   from Claim E (a) and Claim G (b).} 
    \label{fig:u-and-uprime}
\end{figure}
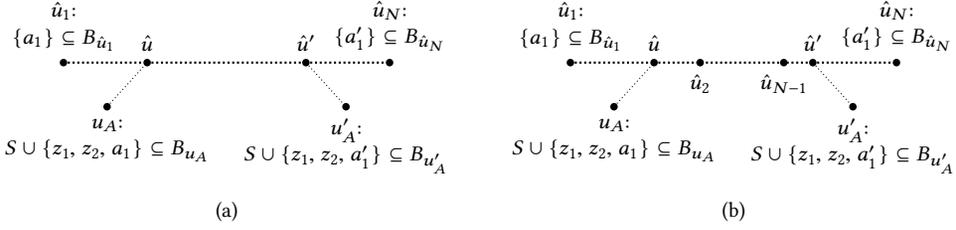

\medskip
{\sc Claim E.} 
{\it \clmE}

\medskip 
{\sc Proof of Claim E.}
  First, we show that $\hat{u}$ and $\hat{u}'$ are indeed distinct. 
Suppose towards a contradiction that they are not, i.e.\ $\hat{u} = \hat{u}'$. Then, 
by connectedness, $\hat{u}$ has to cover $S \;\cup $ $\{ z_1, 
z_2\}$, because $S \;\cup $ $\{ z_1, z_2\}$ is contained in 
$B_{u_A}$ and in  $B_{u'_A}$.
Moreover, again by connectedness, $\hat{u}$ also has to cover
$\{a_1, a'_1\}$, because $a_1$ is contained in $B_{\hat{u}_1}$ and in $B_{u_A}$ and 
$a'_1$ is contained in $B_{\hat{u}_N}$ and in $B_{u'_{A}}$.
As in Case (3) in the proof of Claim D, this is impossible
by Lemma~\ref{lem:compl_edge}. Hence, $\hat{u}$ and $\hat{u}'$ are distinct. 

Second, we show that, on the path from 
$\hat{u}_1$ to  $\hat{u}_N$, the node $\hat{u}$ comes
before $\hat{u}'$.
Suppose to the contrary that $\hat{u}'$ comes before $\hat{u}$. 
Then, by the connectedness condition,  $\hat{u}$ covers the following (sets of) 
vertices: 
\begin{itemize}
\item $a'_1$, since we are assuming that $\hat{u}'$ comes before $\hat{u}$, i.e., 
$\hat{u}$ is on the path from $\hat{u}_N$ to $u'_A$;
\item $a_1$, since $\hat{u}$ is on the path from  $\hat{u}_1$ to $u_A$;
\item $S \cup \{z_1,z_2\}$, since $\hat{u}$ is on the path from  $u_A$ to 
$u'_A$.
\end{itemize}
In total, $\hat{u}$  has to cover all vertices in 
$S \;\cup $  $\left\{ z_1, z_2, a_1, a'_1\right\}$. 
Again, by Lemma~\ref{lem:compl_edge}, 
this 
is impossible with weight $\leq 2$. $\hfill\diamond$

\medskip
{\sc Claim F.} 
{\it \clmF}

\medskip
{\sc Proof of Claim F.}
First, it is easy to verify that $N \geq 2$ must hold. Otherwise, a single node 
would have to cover
$\{e_{\min \ominus 1}, e_{\min}$, $e_{\min \oplus 1}$, \dots, $e_{\max \ominus 
1}$, $e_{\max}\}$
and, hence, in particular, $S \cup \{z_1,z_2,a_1,a'_1\}$, 
which is impossible as we  have already seen in Case (3) of the proof of Claim D.

It remains to prove $N \geq 3$. Suppose to the contrary that $N = 2$. By the 
problem reduction, hypergraph $H$ has distinct edges $e_{\min \ominus 1}$, $e_{\min}$ and $e_{\max}$.
Hence, 
$\hat{u}_1$ covers at least $e_{\min \ominus 1}$ and $\hat{u}_2$ covers at least $e_{\max}$.
Recall from Claim~E the nodes $\hat{u}$ and $\hat{u}'$, which constitute the endpoints of the intersection
of the path from $u_A$ to $u'_A$ with the path $\ppi$, cf.\ Figure~\ref{fig:u-and-uprime}(a).  
Here we are assuming $N = 2$. We now show that, by the connectedness condition of FHDs,
the nodes $\hat{u}$ and $\hat{u}'$ must cover certain vertices, which will lead to 
a contradiction by Lemma~\ref{lem:compl_edge}.

\begin{itemize}
 \item vertices covered by $\hat{u}$: node $\hat{u}$ is on the path between $u_A$ and $u'_A$. Hence, it covers $S \cup \{ z_1,z_2\}$. Moreover, $\hat{u}$ is on the path between
 $\hat{u}_1$ and $u_A$ (or even coincides with $\hat{u}_1$). Hence, 
 it also covers $a_1$. In total, $\hat{u}$  covers at least $S \cup \{ z_1,z_2, a_1\}$.

 \item vertices covered by $\hat{u}'$: node $\hat{u}'$ is on the path between $u_A$ and $u'_A$. Hence, it covers $S \cup \{ z_1,z_2\}$. Moreover, $\hat{u}'$ is on the path between
 $\hat{u}_2$ and $u'_A$ (or even coincides with $\hat{u}_2$). Hence, 
 it also covers $a'_1$. In total, $\hat{u}'$  covers at least $S \cup \{ z_1,z_2, a'_1\}$.
\end{itemize} 

\smallskip

\noindent
One of the nodes $\hat{u}_1$ or $\hat{u}_2$ must also cover the edge $e_{\min}$.
We inspect these 2 cases separately:

\begin{itemize}
\item Case~(1): suppose that the edge $e_{\min}$ is covered by $\hat{u}_1$. 
Then, $\hat{u}_1$ covers vertex $a'_{\min}$, which is also covered by $\hat{u}_2$.
Hence, also $\hat{u}$  covers $a'_{\min}$. In total, $\hat{u}$ covers 
$S \cup \{ z_1,z_2, a_1, a'_{\min}\}$. However, 
by Lemma~\ref{lem:compl_edge}, we know that, 
to cover $S \cup \{ z_1, z_2 \}$ with weight $\leq 2$,
we are only allowed to put 
non-zero weight on pairs of complementary edges. Hence, 
it is impossible to achieve also weight $\geq 1$ on $a_1$ and on $a'_{\min}$
at the same time. 

\item Case~(2): suppose that the edge $e_{\min}$ is covered by $\hat{u}_2$. 
Then, $\hat{u}_2$ covers vertex $a_{\min}$ (actually, it even covers all of $A$), 
which is also covered by $\hat{u}_1$.
Hence, also $\hat{u}'$  covers $a_{\min}$. In total, $\hat{u}'$ covers 
$S \cup \{ z_1,z_2, a'_1, a_{\min}\}$. Again, this is impossible 
by Lemma~\ref{lem:compl_edge}.
\end{itemize}

Hence, the path $\ppi$ indeed has at least 3 nodes $\hat{u}_i$. $\hfill\diamond$

\medskip
{\sc Claim G.} 
{\it \clmG}

\medskip
{\sc Proof of Claim G.}
We have to prove that $\hat{u}$ lies between 
$\hat{u}_1$ and $\hat{u}_2$ (not including $\hat{u}_2$)
and $\hat{u}'$ lies between 
$\hat{u}_{N-1}$ and $\hat{u}_N$ (not including $\hat{u}_{N-1}$).
For the first property, 
suppose to the contrary that 
$\hat{u}$ does not lie between $\hat{u}_1$ and $\hat{u}_2$
or $\hat{u} = \hat{u}_2$.
This means, that there exists $i \in \{2, \dots, N-1\}$
such that $\hat{u}$ lies between $\hat{u}_i$ and $\hat{u}_{i+1}$, 
including  the case that $\hat{u}$
coincides with $\hat{u}_i$. Note that, by Claim E, $\hat{u}$ cannot coincide with $\hat{u}_N$, 
since there is yet another
node $\hat{u}'$ between  $\hat{u}$ and $\hat{u}_N$.

By definition of $\hat{u}_i$ and $\hat{u}_{i+1}$, there is a $p \in [2n+3;m]$, 
such that both $\hat{u}_i$ and $\hat{u}_{i+1}$ cover $a'_p$. 
Then, by the connectedness condition,  $\hat{u}$ covers the following (sets of) 
vertices: 
\begin{itemize}
\item $a'_p$, since $\hat{u}$ is on the path from $\hat{u}_{i}$ 
to $\hat{u}_{i+1}$
(or 
$\hat{u}$ coincides with $\hat{u}_{i}$), 
\item $a_1$, since $\hat{u}$ is on the path from  $\hat{u}_1$ to $u_A$,
\item $S \cup \{z_1,z_2\}$, since $\hat{u}$ is on the path from  $u_A$ to 
$u'_A$.
\end{itemize}
However, by Lemma~\ref{lem:compl_edge}, we know that, 
to cover $S \cup \{ z_1, z_2 \}$ with weight $\leq 2$,
we are only allowed to put 
non-zero weight on pairs of complementary edges.
Hence, it is impossible to achieve also weight $\geq 1$ on $a'_p$ and $a_1$
at the same time.

It remains to show that $\hat{u}'$ lies between 
$\hat{u}_{N-1}$ and $\hat{u}_N$ (not including $\hat{u}_{N-1}$).
Suppose to the contrary that it does not. Then, 
analogously to the above considerations for $\hat{u}$, 
it can be shown that there exists some $p \in [2n+3;m]$, such that $\hat{u}'$
covers the vertices
$S \cup \{z_1,z_2, a_p, a'_1\}$. Again, this is impossible
by Lemma~\ref{lem:compl_edge}. $\hfill \diamond$

\medskip
By Claim C, the decomposition $\mcF$ contains a path 
$\hat{u}_1 \cdots \hat{u}_N$ that covers the edges 
$e_{\min \ominus 1}, e_{\min}$, $e_{\min \oplus 1}$, \dots, $e_{\max \ominus 
1}$, $e_{\max}$
in this order. We next strengthen this  
property by showing that every node $\hat{u}_i$ covers exactly one edge $e_p$.

\medskip
{\sc  Claim H.}{\it \clmH}

\medskip
{\sc Proof of Claim H.}
We prove this property for the ``outer nodes'' $\hat{u}_1$, $\hat{u}_N$ 
and for the ``inner nodes'' $\hat{u}_2 \cdots \hat{u}_{N-1}$ separately.
We start with the ``outer nodes''.
The proof for $\hat{u}_1$ and $\hat{u}_N$ is symmetric. We thus only work out 
the details for 
$\hat{u}_1$. Suppose to the contrary that $\hat{u}_1$ not only covers $e_{\min 
\ominus 1}$ 
but also $e_{\min}$. We distinguish two cases according to the position of 
node $\hat{u}$ in Figure~\ref{fig:u-and-uprime} (b):

\begin{itemize}
 \item Case~(1): $\hat{u} = \hat{u}_1$. 
Then,  $\hat{u}_1$ has to cover the following (sets of) vertices: 
\begin{itemize}
\item $S \cup \{z_1,z_2\}$, since $\hat{u}$ is on the path from $u_A$ to 
$u'_A$ and we are assuming $\hat{u} = \hat{u}_1$. 

\item $a_1$, since $\hat{u}_1$ covers $e_{\min \ominus 1}$,

\item $a'_{\min}$, since we are assuming that $\hat{u}_1$ also covers $e_{\min}$. 

\end{itemize}

By applying Lemma~\ref{lem:compl_edge}, we may conclude that 
the set $S \cup \{z_1,z_2,a_1,a'_{\min}\}$
cannot be covered by a fractional edge cover of weight $\leq 2$.

\item Case~(2): $\hat{u} \neq \hat{u}_1$. Then 
$\hat{u}$ is on the path from $\hat{u}_1$ to 
$\hat{u}_{2}$. Hence, 
$\hat{u}$ has to cover the following (sets of) vertices: 
\begin{itemize}
\item $S \cup \{z_1,z_2\}$, since $\hat{u}$ is on the path from $u_A$ to $u'_A$,

\item $a_1$, since $\hat{u}$ is on the path from $u_A$ to $\hat{u}_1$, 

\item $a'_{\min}$, since $\hat{u}$ is on the path from $\hat{u}_1$ to 
$\hat{u}_2$.
\end{itemize}
As in Case (1) above, 
$S \cup \{z_1,z_2,a_1,a'_{\min}\}$
cannot be covered by a fractional edge cover of weight $\leq 2$
due to Lemma~\ref{lem:compl_edge}.
\end{itemize}

It remains to consider the ``inner'' nodes $\hat{u}_i$ with  $2 \leq i \leq 
N-1$. 
Each such $\hat{u}_i$ has to cover $S \cup \{z_1,z_2\}$ since all these nodes 
are on the path from $u_A$ to $u'_A$ by Claim G. Now suppose that $\hat{u}_i$ 
covers 
$e_p = A'_p \cup \overbar{A_p}$
for some $p \in \{e_{\min}, \dots, e_{\max \ominus 1}\}$. 
By Lemma~\ref{lem:covering}, covering all of the vertices $A'_p \cup 
\overbar{A_p} \cup S \cup \{z_1,z_2\}$ by a fractional edge cover of weight 
$\leq 2$ 
requires that we put total weight $1$ on  the
edges $e^{k,0}_p$ and total weight $1$ on the edges $e^{k,1}_p$
with $k \in \{ 1,2,3\}$.
However, then it is 
impossible to cover also $e_{p'}$ for some $p'$ with $p' \neq p$.
This concludes the proof of Claim F.$\hfill \diamond$

\medskip
We can now associate with each $\hat{u}_i$ for $1 \leq i \leq N$ the 
corresponding 
edge $e_p$ and write $u_p$ to denote the node that covers the edge $e_p$.
By Claim G, we know that all of the nodes $u_{\min} \dots, u_{\max\ominus 1}$ 
are on the path from $u_A$ to $u'_A$. Hence, by the connectedness condition, 
all these nodes cover $S \cup \{z_1,z_2\}$.

We are now ready to construct a satisfying truth assignment $\sigma$ of 
$\varphi$. For each $i \leq 2n+3$, let $X_i$ be the set $B_{u_{(i,1)}} \cap (Y 
\cup Y')$. As 
  $Y \subseteq B_{u_A}$ and $Y' \subseteq B_{u'_A}$, the sequence $X_1 \cap 
  Y, \ldots, X_{2n+3}\cap Y$ is non-increasing and the sequence $X_1 \cap Y', 
  \ldots,
  X_{2n+3}\cap Y'$ is non-decreasing. Furthermore, as all edges $e_{y_i} = \{ 
  y_i, y'_i \}$ must be covered by some node in $\mcF$, we conclude that for 
  each 
  $i$ and $j$, $y_j \in X_i$ or $y'_j \in X_i$.
  Then, there is some $s \leq 2n+2$ such that $X_s = X_{s+1}$. Furthermore,
  all nodes between $u_{(s,1)}$ and $u_{(s+1,1)}$ cover $X_s$. 
  We derive a truth assignment for $x_1, \ldots, x_n$ from $X_s$ as follows. For
  each $\ell \leq n$, we set $\sigma(x_{\ell}) = 1$ if $y_{\ell} \in X_s$ and otherwise
  $\sigma(x_{\ell}) = 0$. Note that in the latter case $y'_{\ell} \in X_s$.

\medskip
{\sc  Claim I.}
{\it \clmI}

\medskip
{\sc Proof of Claim I.}
We have to show that every clause $c_j =   L_j^1   \vee L_j^2 \vee L_j^3$ of $\varphi$ 
is true in $\sigma$. 
Choose an arbitrary $j \in \{1, \dots, m\}$.
We have to show that there exists a literal in $c_j$ which is true in $\sigma$.
To this end, we inspect the node $u_{(s,j)}$, which, by construction, 
lies between $u_{(s,1)}$ and $u_{(s+1,1)}$. 
Let $p = (s,j)$.
Then we have $A'_{p} \cup \overbar{A_{p}} \cup S   \cup  \{  z_1, z_2\}
\subseteq B_{u_{p}}$. Moreover, by the definition of $X_s$, we also have 
$X_s \subseteq B_{u_{p}}$.
  By Lemma~\ref{lem:covering}, the only way to cover
  $B_{u_{p}}$ with weight $\leq 2$ is by using exclusively the edges 
  $e^{k,0}_p$ and $e^{k,1}_p$ with $k \in \{1,2,3\}$.
More specifically, we have 
 $\sum_{k=1}^3 \gamma_{u_{p}}(e^{k,0}_p) = 1$ and $\sum_{k=1}^3 \gamma_{u_{p}}(e^{k,1}_p) = 1$.
Therefore, $\gamma_{u_{p}}(e^{k,0}_p) > 0$ for some $k$. 
We distinguish two cases depending on the form of literal $L^k_j$:

\begin{itemize}
\item Case (1): First, suppose $L^k_j = x_{\ell}$. By Lemma~\ref{lem:compl_edge}, 
complementary edges must have equal weight. Hence, 
from $\gamma_{u_{p}}(e^{k,0}_p) > 0$ it follows that 
also 
  $\gamma_{u_{p}}(e^{k,1}_p) > 0$ holds. Thus, the weight on $y'_{\ell}$ is less
  than $1$, which means that $y'_{\ell} \not\in B(\gamma_{u_{p}})$ and 
  consequently $y'_{\ell} \not\in X_s$. Since this implies that $y_{\ell} \in X_s$, we indeed 
have that 
  $\sigma(x_{\ell}) = 1$. 
  
\item Case (2):  Conversely, suppose $L^k_j =  \neg x_{\ell}$. Since $\gamma_{u_{p}}(e^{k,0}_p) > 0$, the weight on $y_{\ell}$ is 
  less
  than $1$, which means that $y_{\ell} \not\in B(\gamma_{u_{p}})$ and 
  consequently $y_{\ell} \not\in X_s$. Hence, we have $\sigma(x_{\ell}) = 0$. 
  
\end{itemize}  
In either case, literal  $L^k_p$ is satisfied by $\sigma$ and therefore, the $j$-th clause $c_j$ is 
satisfied by $\sigma$. Since $j$ was arbitrarily chosen, $\sigma$  indeed satisfies $\varphi$.$\hfill \diamond$

\medskip
\noindent
Claim~I completes the proof of Theorem~\ref{thm:npcomp}.
\end{proof}

\noindent
We conclude this section  by mentioning that the above reduction 
is easily extended to $k+\ell$ for arbitrary
$\ell \geq 1$: for integer values $\ell$, simply add a clique of $2 \ell$ fresh 
vertices 
$v_1, \dots, v_{2 \ell}$ to $H$ and connect each $v_i$ with each ``old'' vertex 
in $H$. Now assume a rational value $\ell \geq 1$, i.e., 
$\ell = r / q$ for natural numbers $r,q$ with $r >q > 0$. 
To achieve a rational bound $k + r/q$, we add %
$r$ fresh 
vertices and 
add hyperedges 
$\{v_i, v_{i\oplus 1},\dots, v_{i \oplus (q-1)}\}$ with $i \in \{1, \dots, 
r\}$
to $H$, where $a \oplus b$ denotes $a + b$ modulo $r$. Again, we connect 
each 
$v_i$ with each ``old'' vertex in $H$. With this construction we
can give NP-hardness proofs for any (fractional) $k \geq 3$. For all fractional values
$k < 3$ (except for $k = 2$) different gadgets and ideas might be needed to
prove NP-hardness of \rec{FHD,$k$}, which we leave  for future work.
 
\section{A Framework for Efficient Computation of Decompositions}
\label{sect:framework}

Before we move on to the easy cases for \rec{GHD,\,$k$} and \rec{FHD,\,$k$}, we will introduce a framework for a uniform presentation of the tractability proofs in the following sections. Conceptually, we can split the task of checking whether a decomposition of certain width exists into two parts: (1) deciding which sets of vertices are acceptable as bags (the \emph{candidate bags}) and (2) deciding if there is a tree decomposition made up of only acceptable bags. In this section, we focus on the second part and show that this task is indeed tractable as long as the decompositions satisfy a certain normal form. This will allow us to show the \textsc{Check} problem tractable for settings where we can compute an appropriate set of candidate bags in polynomial time.

To emphasize the generality of the approach we will focus on tree decompositions in this section. Recall that a generalized hypertree decomposition of width at most $k$ is simply a tree decomposition where every bag has an integral edge cover with weight at most $k$. The same is true for fractional hypertree decompositions and fractional edge covers.  Hence, the \rec{GHD,\,k} problem can be solved by computing appropriate sets $\mathbf{S}$ of candidate bags that can be covered by $k$ edges and then deciding whether there exists a TD using only bags from $\mathbf{S}$.  If such a TD exists, it is a witness for the existence of a GHD of width at most $k$. Of course, the same strategy also works for \rec{FHD,\,k}.

First, we will formally define the task we are interested in as the candidate tree decomposition problem.  We show that the problem is \np-complete even for acyclic graphs. Following that, we show that the problem becomes tractable if we limit our search to finding TDs that adhere to a certain normal form which is sufficient for our purposes.

\begin{definition}
  Let $H$ be a hypergraph and $\mcT=\defTD$ be a tree decomposition of $H$.  Let the \emph{candidate bags} ${\bf S}$ be a family of subsets
  of $V(H)$.  If for each $u \in T$ there exists an $S \in \mathbf{S}$ such
  that $B_u=S$, then we call $\mcT$ a \emph{candidate tree decomposition} of ${\bf S}$.  We
  denote by $\realization{\mathbf{S}}$ the
  set of all candidate tree decompositions of $H$.
\end{definition}

\begin{theorem}
  Let $H$ be a hypergraph and $\mathbf{S} \subseteq 2^{V(H)}$. It is \NP-complete to decide 
  whether $\realization{\mathbf{S}} \neq \emptyset$. The problem remains \NP-complete even if 
  we restrict the choice of $H$ to acyclic graphs.
\end{theorem}

\begin{proof} 
  The problem is clearly in \np.  We show \NP-hardness by reduction
  from the exact cover problem: Let $U=\{u_1, \dots, u_n\}$ be the
  universe and let $X_1, \dots, X_m$ be subsets of $U$. The exact
  cover problem asks for a cover of $U$ by elements of
  $\{X_1, \dots, X_m\}$ such that the sets in the cover are pairwise
  disjoint.

  We define an acyclic graph $G$ as follows: $G$ is a tree with vertices
  $v,v_1, \dots. v_n,u_1, \dots, u_n$.  The edges of $G$  are
  $\{v,v_i\}$ and $\{v_i,u_i\}$ for each $1 \leq i \leq n$.  Let
  ${\bf S}=\{S,S_1, \dots, S_m\}$ where $S=\{v,v_1, \dots v_n\}$ and
  each $S_i = X_i \cup \{v_j \mid u_j \in X_i\}$, i.e., by taking
  $X_i$ and adding $v_j$ for each $u_j$ contained in $X_i$.

  We claim that then $\realization{\mathbf{S}} \neq \emptyset$ iff there is
  an exact cover of $U$ by $X_1, \dots, X_m$.  One direction is
  easy. Let ${\bf X} \subseteq \{X_1, \dots, X_m\}$ be an exact 
  cover of $U$. Denote the elements of ${\bf X}$ by
  $X'_1, \dots, X'_q$.  For $1 \leq i \leq q$, let $S'_i$ be obtained
  from $X'_i$ by adding $v_i$ for each $u_i \in X'_i$. Clearly
  $S'_i \in \{S_1, \dots, S_m\}$ and there exists a candidate tree decomposition $\defTD$ of $G$ where $V(T)=\{t,t_1, \dots, t_q\}$,
  $E(T)=\{\{t,t_1\}, \dots, \{t,t_q\}\}$, $B_t=S$ and for each
  $1 \leq i \leq q$, $B_{t_i}=S'_i$.

  The other direction is more complicated. Let $\defTD$ be a
  smallest (w.r.t. the number of nodes of $T$) candidate TD of
  $G$.  In particular, this means that the bags of all the nodes are distinct.
  The proof proceeds in several steps:

  \begin{enumerate}
  \item There is $t \in V(T)$ such that $B_t=S$.
    Indeed, otherwise, $v$ is not covered.
  \item For each $1 \leq i \leq n$, there is $t' \in V(T)$ such that
    $u_i \in B_{t'}$ and $t$ is adjacent to $t'$.  Indeed, assume the
    opposite and let $t' \in V(T)$ be a node with $u_i \in B_{t'}$
    such that $t$ is not adjacent to $t'$.

    Since we assume a candidate TD, we have $B_{t'} = S_j$ for some $j$ and hence
    $v_i \in B_{t'}$.  Now consider the path between $t$ and $t'$ and let
    $s$ be the node next to $t$ on this path.  By our minimality
    assumption we know that $B_s \neq B_t$ and thus $B_s = S_k$ with
    $k \neq j$.  Since we assume the claim to be false we have
    $u_i \notin S_k$ and hence $v_i \notin S_k$. As $v_i \in B_{t'}$ and
    $v_i \in B_t$, the connectedness condition of the tree decomposition
    is violated.
  \item Let $t_1, \dots, t_q$ be the neighbors of $t$ in $T$.  We claim
    that $T$ has no other nodes. Indeed, all the vertices of $G$ are
    covered by the bags of $t,t_1, \dots t_q$ by the previous two items.
    Each edge $\{v,v_i\}$ is contained in $B_t$. Also, each edge
    $\{u_i,v_i\}$ is covered by some $B_{t_j}$ containing $u_i$ (existing by
    the previous item).  It follows that $T[\{t,t_1, \dots, t_q\}]$
    together with the corresponding bags form a tree decomposition of
    $G$. By the minimality assumption, $T$ does not have other nodes.
  \item We claim that for any $1 \leq i \neq j \leq q$,
    $B_{t_i} \cap B_{t_j} \neq \emptyset$.  Indeed, otherwise, there is
    $u_k \in B_{t_i} \cap B_{t_j}$. However, $u_k \notin B_t$ in
    contradiction to the connectedness condition.  It follows that
    $B_{t_1} \cap U, \dots, B_{t_q} \cap U$ are disjoint elements of
    $\{X_1, \dots, X_m\}$ covering all of $U$ as required.
  \end{enumerate}
  
\vspace{-14pt}

\end{proof}

To obtain a tractable version of the problem we will introduce a generalization of the normal form that was used in the tractability proof for HDs in~\cite{2002gottlob}. 

\begin{definition}
  \label{def:compnf}
  A tree decomposition $\defTD$ of a hypergraph $H$ is in {\em
    component normal form (\compnf)\/} if for each node $r \in T$, and
  for each child $s$ of $r$ there is {\em exactly one\/} \comp{$B_r$}
  $C_s$ such that $\VTs = C_s \cup (B_r \cap B_s)$ holds. We say $C_s$
  is the component \emph{associated with} node $s$.
\end{definition}

By the above definition, we know that in in a \compnf TD, every child $s$ of a node $r$ is associated with at most 
one \comp{$B_r$} $C_s$. By the connectedness condition, also the converse is true, i.e., 
every \comp{$B_r$} $C_s$ is associated with \emph{at most one} child $s$. Note that (by applying the ideas of the 
transformation of HDs into the normal form of~\cite{2002gottlob}), every TD can be transformed in polynomial time into a TD in 
\compnf without increasing the width (more precisely, the bags in the resulting TD are subsets of the bags in the 
original TD).

\begin{definition}
  Let $H$ be a hypergraph and let ${\bf S}$ be a family of subsets of
  $V(H)$.  Let $\mcT \in \realization{\mathbf{S}}$ be a tree
  decomposition in \compnf.  We say $\mcT$ is a \emph{\compnf
    candidate tree decomposition} of ${\bf S}$.  We denote by
  $\blockreal{\mathbf{S}}$ the set of all \compnf candidate tree
  decompositions of $H$.
\end{definition}

\begin{theorem}
  \label{thm:frameworkalg}
  Let $H$ be a hypergraph and $\mathbf{S} \subseteq 2^{V(H)}$. There
  exists an algorithm that takes $H$ and $\mathbf{S}$ as an input and
  decides in polynomial time whether
  $\blockreal{\mathbf{S}} \neq \emptyset$, and if so, return a $\mcT \in \blockreal{\mathbf{S}}$.
\end{theorem}

The intuition behind a polynomial-time algorithm for this problem is simple. Every parent/child relationship in a tree decomposition corresponds to a separator $S$ (the bag of the parent) and an $[S]$-component. It should therefore be enough to first enumerate all pairs $(S,C)$ of separators $S$ in $\mathbf{S}$ and [$S$]-components $C$, 
and then check  if these pairs, which we will call \emph{blocks}, can be combined to form a valid tree decomposition. Through the restriction to a specific set $\mathbf{S}$, that is part of the input, the number of blocks we have to consider is only polynomial in the input.

We are not aware of a proof of Theorem~\ref{thm:frameworkalg} in the
literature. In light of the hardness result for the general case, we
choose to present a full proof of the theorem in
Appendix~\ref{sec:algappendix} even though it could be considered
folklore.
  In particular, methods for subedge-based decompositions~\cite{2009gottlob} as well as tree projections and their associated algorithms, see e.g.,~\cite{DBLP:journals/jcss/GoodmanS84,DBLP:journals/jcss/GottlobGS18,DBLP:journals/jal/LustigS99}, are closely related to the \compnf CTD problem.

The restriction to component normal form will ultimately not restrict us in the
following sections. For generalized and fractional hypertree width,
\compnf can be enforced without increasing the width.  Such a
transformation can be found, e.g, in~\cite{2002gottlob} and as part of
the proof of Lemma~\ref{lem:nfghd}.  Still, some care will be required
in the enumeration of the candidate bags to guarantee that they allow
for a decomposition in component normal form.

\section{Tractable Cases of GHD Computation}
\label{sect:bmip}
\label{sect:ghd}

As discussed in Section \ref{sect:introduction} we are interested in finding a realistic and non-trivial criterion on hypergraphs that 
makes the \rec{GHD,\,$k$} problem tractable for fixed $k$.
We thus propose here such a simple property, namely the bounded intersection of two or more edges.

\begin{definition}\label{def:bip}
The {\em intersection width} $\iwidth{\HH}$ of a hypergraph $\HH$ is the 
maximum 
cardinality of any intersection $e_1\cap e_2$ of two distinct edges $e_1$ and 
$e_2$ 
of $\HH$.
We say that a hypergraph $H$ has
the {\em $i$-bounded intersection property ($i$-BIP)} if 
$\iwidth\HH\leq i$ holds. 

Let $\classC$ be a class of hypergraphs. 
We say that $\classC$ 
has the  
{\em bounded intersection property (BIP)} if there exists some 
integer 
constant $i$ such that
every hypergraph $H$ in $\classC$
 has the $i$-BIP.
Class $\classC$ has the 
{\em logarithmically-bounded intersection property (LogBIP)}
if for each of its elements $\HH$, $\iwidth{\HH}$ is $\calO(\log n)$, where 
$n = ||\HH||$ denotes the size of $\HH$. 
\end{definition}

The BIP criterion properly generalizes bounded arity and is indeed non-trivial in the sense that there exist classes of unbounded $\ghw$ that enjoy the BIP. Among others this includes the classes of graphs, regular hypergraphs, and linear hypergraphs.
Moreover, 
a recent empirical study \cite{pods/FischlGLP19} 
suggests that the overwhelming number of CQs enjoys the $2$-BIP
(i.e., one hardly joins two relations over more than 2 attributes).
To allow for a yet bigger class of hypergraphs, the BIP can be 
relaxed as follows.

\begin{definition}\label{def:bmip}
The {\em $c$-multi-intersection width} $\cmiwidth{c}{\HH}$ of a 
hypergraph 
$\HH$ 
is the maximum cardinality of any intersection $e_1\cap\cdots\cap e_c$  of 
$c$ distinct edges $e_1, \ldots,  e_c$ of $\HH$.  
We say that a hypergraph $H$ has
the 
{\em $i$-bounded $c$-multi-intersection property (\icBMIP)} if
$\cmiwidth{c}{\HH}\leq i$ holds.

Let $\classC$ be a class of hypergraphs. 
We say that $\classC$ 
has the  
{\em bounded multi-intersection property (BMIP)} 
if there exist constants $c$ and $i$ 
such that
every hypergraph $H$ in $\classC$
has the \icBMIP.
Class $\classC$ of hypergraphs has the 
{\em logarithmically-bounded  multi-intersection property (LogBMIP)} if there 
is 
a constant $c$ such that 
for the hypergraphs $\HH\in \classC$, $\cmiwidth{c}{\HH}$ is 
$\calO(\log n)$, 
where $n$ denotes the size of the
hypergraph $\HH$.
\end{definition}

  \begin{figure}[h]
    \centering
     \includegraphics{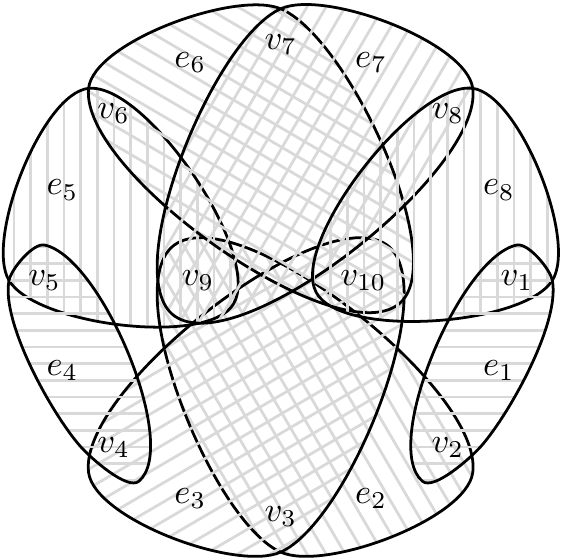}
    \caption{Hypergraph $\HH_0$ from Example \ref{ex:ghw1}} 
    \label{fig:GHDvsHD_graph}
  \end{figure}

  \begin{figure}[h]
  \centering 
  \footnotesize
    \tikzstyle{htmatrix}=[%
      matrix,
      matrix of nodes,
      nodes={draw=none},
      column 1/.style={nodes={fill=gray!20,align=right},anchor=base east,minimum width=2em},
      column 2/.style={anchor=base west},
      ampersand replacement=\&
    ]
      \begin{tikzpicture}[sibling distance=13em,
     every node/.style = {shape=rectangle,draw,inner sep=1,align=center}]
  \node[htmatrix] { \htnode{$v_1,v_2,v_3,v_6,v_7,v_9,v_{10}$}{$e_1,e_2,e_6$}}
    child { node[htmatrix] {\htnode{$v_3,v_4,v_5,,v_6,v_7,v_9,v_{10}$}{$e_3,e_5$}} }
    child { node[htmatrix] {\htnode{$v_1,v_7,v_8,v_9,v_{10}$}{$e_7,e_8$} }
    };
\end{tikzpicture}
\caption{HD of hypergraph $H_0$ in Figure~\ref{fig:GHDvsHD_graph}}
\label{fig:H0_HD}
\end{figure}
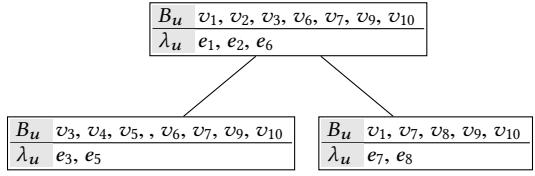

\begin{example}
  \label{ex:ghw1}
  
  Figure~\ref{fig:GHDvsHD_graph} shows the hypergraph $H_0 = (V_0,E_0)$ with $\ghw(\HH_0)=2$ but $\hw(\HH_0)$=3. (This example is from~\cite{2009gottlob}, which, in turn, was inspired by work of Adler~\cite{adler2004marshals}). 
    Figure \ref{fig:H0_HD}
  shows an HD of width 3 and Figure \ref{fig:H0_GHD} shows GHDs of width 2 for the hypergraph $H_0$. 
  The \emph{iwidth} and the \emph{$3$-miwidth} of $H_0$ is 1. Starting from $c=4$, the \emph{$c$-miwidth} is 0. 
  \hfill$\Diamond$
\end{example}

The LogBMIP is  
the most liberal restriction on classes of hypergraphs introduced in 
Definitions~\ref{def:bip} and \ref{def:bmip}. 
The main result in this section will be that the
\rec{GHD,$\,k$} problem with fixed $k$ 
is tractable for any class of hypergraphs satisfying this 
criterion.

\begin{theorem}
  \label{theo:LogBMIP}
For every hypergraph class $\classC$ that enjoys the LogBMIP, and for every 
constant 
$k\geq 1$, the \rec{GHD,\,$k$} problem is tractable, i.e., 
given a hypergraph $\HH$, it is feasible in polynomial time to 
check $\ghw(\HH)\leq k$ and, if so, to compute a GHD of 
width $k$ of $\HH$.%
\end{theorem}

Our plan is to make use of Theorem~\ref{thm:frameworkalg} by computing
appropriate sets $\mathbf{S}$ of candidate bags such that
$\blockreal{\mathbf{S}} \neq \emptyset$ if and only if there exists a
GHD with width at most $k$. Example~\ref{ex:ghw2} illustrates the
main challenge that needs to be tackled to compute such sets of
candidate bags.  A bag $B_u$ in a GHD can be any subset of
$B(\lambda_u)$ and choosing smaller subsets can decrease the width.
At the same time, enumerating all subsets of $B(\lambda_u)$
is not an option if we are interested in classes of hypergraphs with
unbounded rank. The main reason why the \textsc{Check} problem
is tractable for HDs is that the additional special condition
severely restricts the possible choices of $B_u$ for given $B(\lambda_u)$.

\begin{example}[Example \ref{ex:ghw1} continued]
  \label{ex:ghw2}
  In Figure~\ref{fig:H0_GHD}, we have two GHDs of width 2 of the hypergraph $H_0$ from 
  Figure~\ref{fig:GHDvsHD_graph}. In the root $u_0$ of both GHDs, we have
  $v_2 \in B(\lambda_{u_0})$ since $v_2 \in e_2$ 
  but $v_2 \not \in B_{u_0}$.  Hence, both GHDs violate the special condition in node $u_0$.
  However, if $v_2$ were added to $B_{u_0}$, then it can be seen that covering the edges $e_1, e_2, e_7, e_8$
  below $u_0$ is no longer possible in a width 2 GHD.
  That is why the HD in Figure~\ref{fig:H0_HD} has width 3.
  \hfill$\Diamond$
\end{example}

\begin{figure}[b]
    \centering 
  \footnotesize
    \tikzstyle{htmatrix}=[%
      matrix,
      matrix of nodes,
      nodes={draw=none},
      column 1/.style={nodes={fill=gray!20,align=right},anchor=base east,minimum width=2em},
      column 2/.style={anchor=base west},
      ampersand replacement=\&
    ]

    \begin{minipage}[c]{0.48\textwidth}
       \centering
       \begin{tikzpicture}[sibling distance=12em,
     every node/.style = {shape=rectangle,draw,inner sep=1,align=center}
     ]
      
  \node (root) [htmatrix] {\htnode{$v_3,v_6,v_7,v_9,v_{10}$}{$e_2,e_6$ }}
        child { node (c1) [htmatrix]{\htnode{$v_3,v_7,v_8,v_9,v_{10}$}{ $e_3,e_7$}} 
           child { node (c2) [htmatrix]{\htnode{$v_1,v_2,v_3,v_8,v_9,v_{10}$}{$e_2,e_8$ }}
                 }
              }
        child { node (c1a) [htmatrix]{\htnode{$v_3,v_6,v_9,v_{10}$}{$e_3,e_5$}}
     child { node (c3) [htmatrix]{\htnode{$v_3,v_4,v_5,v_6,v_9,v_{10}$}{$e_3,e_5$}}
        }
     };
     
     \node [above=of c1a.north west, draw=none, style={color=red}, anchor=south west, yshift=-1cm] { $u'$: };
     
     \node [above=of root.north west, draw=none, style={color=red}, anchor=south west, yshift=-1cm] { $u_0 = u$: };

     \node [above=of c1.north west, draw=none, style={color=red}, anchor=south west, yshift=-1cm] { $u_1$: };
     \node [above=of c2.north west, draw=none, style={color=red}, anchor=south west, yshift=-1cm] { $u_2 = u^*$: };

\end{tikzpicture}
     \centering\medskip\noindent (a)
     \end{minipage}
     \begin{minipage}[c]{0.48\textwidth}
       \centering
       \begin{tikzpicture}[sibling distance=12em,
     every node/.style = {shape=rectangle,draw,inner sep=1,align=center}
     ]
     
  \node (root) [htmatrix] {\htnode{$v_3,v_6,v_7,v_9,v_{10}$}{$e_2,e_6$}}
        child { node (c1) [htmatrix]{\htnode{$v_3,v_7,v_8,v_9,v_{10}$}{$e_3,e_7$}} 
           child { node (c2) [htmatrix]{\htnode{$v_1,v_2,v_3,v_8,v_9,v_{10}$}{$e_2,e_8$}}
                 }
              }
     child { node (c3) [htmatrix]{\htnode{$v_3,v_4,v_5,v_6,v_9,v_{10}$}{$e_3,e_5$}}
     };
     
     \node [above=of root.north west, draw=none, style={color=red}, anchor=south west, yshift=-1cm] { $u_0 = u$: };

     \node [above=of c1.north west, draw=none, style={color=red}, anchor=south west, yshift=-1cm] { $u_1$: };
     \node [above=of c2.north west, draw=none, style={color=red}, anchor=south west, yshift=-1cm] { $u_2 = u^*$: };
     
     \node (circ1) [shape=circle, style={color=red}, thick] at ([yshift=-5,xshift=-19]root) {  \;\;\; };
     \node (circ2) [shape=circle, style={color=red}, thick] at ([yshift=5,xshift=-13]c2) {  \;\;\; };
\end{tikzpicture}
     \centering\medskip\noindent (b)
     \end{minipage}

   \caption{(a) non bag-maximal vs. (b) bag-maximal GHD of hypergraph $H_0$ in Figure~\ref{fig:GHDvsHD_graph}}
   \label{fig:H0_GHD}
  \end{figure}
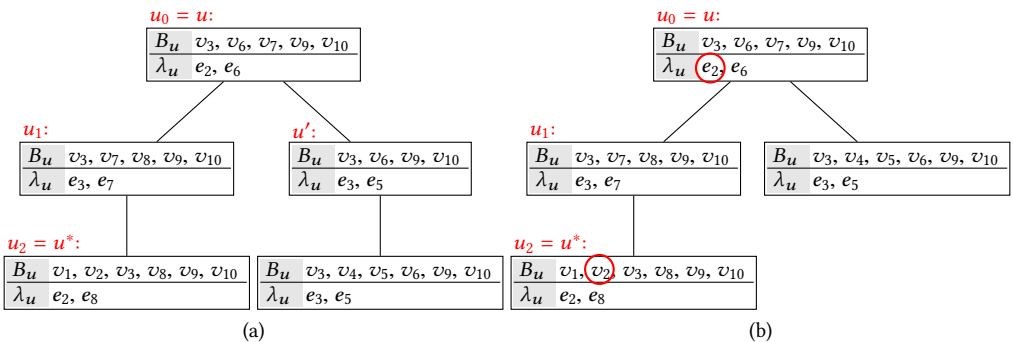

We start by introducing a useful property of GHDs, which we will call 
{\em bag-maximality\/}. 
Let $\mcH=\defGHD$ be a GHD of some hypergraph $H = (V(H), E(H))$. 
For each node $u$ in $T$, we have $B_u \subseteq B(\lambda_u)$ by definition of GHDs
and, in general, $B(\lambda_u) \setminus B_u$ may be non-empty.
We observe that it is sometimes possible to take some vertices from 
$B(\lambda_u) \setminus B_u$ and add them to $B_u$ 
without violating the connectedness condition. Of course, such an addition of vertices to $B_u$ 
does not violate any of the other conditions of GHDs. Moreover, it does not increase the width.

\begin{definition}
 Let $\mcH=\defGHD$ be a GHD of some hypergraph $H = (V(H), E(H))$.
 We call $\mcH$ {\em bag-maximal\/}, if for every node $u$ in $T$,
adding a vertex $v \in B(\lambda_u) \setminus B_u$ to $B_u$ would violate the connectedness condition. 
\end{definition}

It is easy to verify that if $H$ has a GHD of width $\leq k$, then it
also has a bag-maximal GHD of width $\leq k$. However, since we want
to build on the algorithm from Section~\ref{sect:framework}, we need to
show that this also holds for bag-maximal \compnf GHDs. The problem
here is that adding vertices to a bag $B_u$, to make it maximal, can change the set of
$[B_u]$-components. Fortunately, we can reuse existing arguments on the existence of hypertree decompositions
to show that if $H$ has a GHD of width $\leq k$, then it indeed
also has a bag-maximal \compnf GHD of width $\leq k$.

We now carry over several properties of HDs from \cite{2002gottlob}. An inspection of the corresponding proofs in \cite{2002gottlob} reveals that these properties hold also in the generalized case. We thus state the following results below without explicitly ``translating'' the proofs of \cite{2002gottlob} to the generalized setting. 
Note that \cite{2002gottlob} deals with HDs and, therefore, in all decompositions 
considered there, the special condition holds. However, 
in Lemmas~\ref{lem:52} and \ref{lem:53} below, 
the special condition is not needed.

  We briefly recall the crucial notation for the following lemmas. For a set $V' 
  \subseteq V(H)$, we define $\nodes(V') = 
  \{ u \in T \mid B_u \cap V' \neq \emptyset \}$.
  If we want to make explicit the decomposition $\mcG$, 
  we also write $\nodes(V', \mcG)$ synonymously with $\nodes(V')$.
  By further overloading the $\nodes$ operator, we also write 
  $\nodes(T_u)$ or $\nodes(T_u, \mcG)$ to denote the 
  nodes in a subtree $T_u$ of $T$, i.e.,  $\nodes(T_u) = \nodes(T_u,\mcG) =\{ v \mid v \in T_u \}$.

\begin{lemma}[Lemma~5.2 from \cite{2002gottlob}]
\label{lem:52}
Consider an arbitrary GHD $\mcH = \defGHD$ of a hypergraph $H$. 
Let $r$ be a node in $T$, 
let $s$ be a child of $r$ and let $C$ be a \comp{$B_r$} of $H$ such that
$C \cap \VTs \neq \emptyset$. Then, $\nodes(C,\mcH) 
\subseteq \nodes(T_s)$.
\end{lemma}
 
\begin{lemma}[Lemma~5.3 from \cite{2002gottlob}]
\label{lem:53}
Consider an arbitrary GHD $\mcH = \defGHD$ of a hypergraph $H$. 
Let $r$ be a node in $T$ and let 
$U \subseteq V(H) \setminus B_r$ such that
$U$ is [$B_r$]-connected. Then
$\nodes(U,\mcH)$ induces a (connected) subtree of $T$.
\end{lemma}

\begin{lemma}
  \label{lem:nfghd}
  For every GHD $\mcH=\defGHD$ of width $k$ of a hypergraph $H$,
  there exists a bag-maximal \compnf GHD $\mcH'$ of $H$ of width $\leq k$.
\end{lemma}
\begin{proof}
  Start with a GHD $\mcH=\defGHD$ of width $k$ of $\HH$. As long as
  there exists a node $u \in T$ and a vertex
  $ v\in B(\lambda_u) \setminus B_u$, such that $v$ can be added to
  $B_u$ without destroying the GHD properties, select such a node $u$
  and vertex $v$ arbitrarily and add $v$ to $B_u$. By exhaustive
  application of this transformation, a bag-maximal GHD of width $k$
  of $\HH$ is obtained. 
  
We proceed by restating a procedure from~\cite{2002gottlob} that fixes violations of the \compnf condition:.
for the bag-maximal GHD obtained above, assume that there exist two nodes $r$ and $s$ such that $s$ is a child of $r$, and the \compnf condition is violated for the pair, i.e., there does not exist a single  $[B_r]$-component $C_s$ such that $V(T_s) = C_s \cup (B_r \cap B_s)$.
  Let $C_1,\ldots,C_h$ be all the $[B_r]$-components containing some
  vertex occurring in $V(T_s)$. Hence,
  $V(T_s) \subseteq \left(\bigcup_{j=1}^h C_j \cup B_r\right)$. For
  each $[B_r]$-component $C_j$ ($1 \leq j \leq h$), consider the set
$\nodes(C_j,\mcH)$.
  By Lemma~\ref{lem:53}, $\nodes(C_j,\mcH)$ induces a subtree of $T$, and by Lemma~\ref{lem:52}, $\nodes(C_j,\mcH) \subseteq \nodes(T_s)$. 
  Hence $\nodes(C_j,\mcH)$ induces in fact a subtree of $T_s$.
 
\newcommand{\new}{\ensuremath{u}}
For each node $n \in \nodes(C_j,\mcH)$ define a new node
$\new_{n,j}$ and let $\lambda_{\new_{n,j}}=\lambda_n$ and
$B_{\new_{n,j}} = B_n \cap (C_j \cup B_r)$. Note that
$B_{\new_{n,j}} \neq \emptyset$, because by definition of
$\nodes(C_j,\mcH)$, $B_n$ contains some vertex belonging to $C_j$. Let
$N_j$ = $\{ u_{n,j} \mid n \in \nodes(C_j,\mcH)\}$ and, for any
$C_j$ ($1 \leq j \leq h$), let $T_j$ denote the (directed) graph
$(N_j,E_j)$ such that $u_{p,j}$ is a child of $u_{q,j}$ if and
only if $p$ is a child of $q$ in $T$. $T_j$ is clearly isomorphic to
the subtree of $T_s$ induced by $\nodes(C_j,\mcH)$, hence $T_j$ is a
tree as well. 
 
 Now transform the GHD $\mcH$ as follows: delete the subtree  $T_s$ from $T$
and attach to $r$ every tree $T_j$ for
$1 \leq j \leq h$. In other words, we replace the subtree $T_s$ by a set
 of trees $\{ T_1, \ldots, T_h \}$. By construction, $T_j$ contains a
 node $u_{n,j}$ for each node $n$ belonging to
 $\nodes(C_j,\mcH)\ (1 \leq j \leq h)$. Then, if we let
 $\ffo{children}(r)$ denote the set of children of $r$ in the new tree
 $T$ obtained after the transformation above, it holds that for any
 $s' \in \ffo{children}(r)$, there exists a $[B_r]$-component $C$ of
 $H$ such that $\nodes(T_{s'})=\nodes(C,\mcH)$, and
 $V(T_{s'}) \subseteq (C\cup B_r)$.
We want to show $V(T_{s'})  = C \cup (B_{s'} \cap B_r)$.  
For the ``$\supseteq$''-direction, we observe that 
$C \subseteq T_{s'}$ clearly holds, since we have $C \subseteq T_{s}$ 
and the bags $B_{u_{n,j}}$ in $T_{s'}$ were obtained from $B_n$ in $T_s$ 
as $B_{u_{n,j}} = B_n \cap (C_j \cup B_r)$ and we are considering the component $C = C_j$ here. 
Moreover, $B_{s'} \cap B_r \subseteq B_{s'} \subseteq V(T_{s'})$ clearly holds. Hence, 
we have $C \cup (B_{s'} \cap B_r) \subseteq V(T_{s'})$.
For the ``$\subseteq$''-direction,
we conclude from $V(T_{s'}) \subseteq C \cup B_r$ that also 
$V(T_{s'}) \subseteq C \cup (V(T_{s'}) \cap B_r)$ holds. Hence, it suffices to show that
$(V(T_{s'}) \cap B_r) \subseteq B_{s'}$ holds.
By connectedness, we have $V(T_{s}) \cap B_r \subseteq B_{s} \cap B_r$ and, therefore, 
also  $V(T_{s'}) \cap B_r \subseteq B_{s} \cap B_r$.
Moreover, by construction, we have $B_{s'} = B_s \cap (C \cup B_r)$ and, therefore 
$B_{s} \cap B_r \subseteq B_{s'}$. We thus also arrive at 
$V(T_{s'})  \subseteq C \cup (B_{s'} \cap B_r)$.

It remains to show that $T_{s'}$ is bag-maximal.
Assume to the contrary that there exists a node $\new_{n,j} \in T_{s'}$ and a vertex
 $ v\in B(\lambda_{\new_{n,j}}) \setminus B_{\new_{n,j}}$ such
 that $v$ can be added to $B_{\new_{n,j}}$ without destroying the
 GHD properties.  Recall that
 $B_{\new_{n,j}} = B_n \cap (C_j \cup B_r)$ and
 $\lambda_{\new_{n,j}} = \lambda_n$. Since $T$ was bag-maximal
 initially and from the construction (which only makes bags smaller),
 the only candidates for such a $v$ are those vertices
 $B_n \setminus (C_j \cup B_r)$ that got removed from the
 bag. However, all neighboring nodes are either $r$ or in $T_{s'}$,
 which means their bags are subsets of $C_j \cup B_r$.  Hence, no
 $v \in B_n \setminus (C_j \cup B_r)$ is contained in a neighbor of
 ${\new_{n,j}}$ and adding it would break connectedness. Our newly
 constructed $T_{s'}$ is therefore also bag-maximal.

 Iterating this procedure for all  \compnf  violations will eventually produce a new GHD that is still bag-maximal and in \compnf.
\end{proof}

\begin{example}[Example \ref{ex:ghw2} continued]
\label{ex:ghw3}
Clearly, the GHD in Figure \ref{fig:H0_GHD}(a) violates bag-maximality in node $u'$, since 
the vertices $v_4$ and $v_5$ can be added to $B_{u'}$ without violating any GHD properties. 
If we add $v_4$ and $v_5$ to $B_{u'}$, then bag $B_{u'}$ at node $u'$ and the bag at its child node are the same, which allows us to delete one of the nodes. This results in the GHD given in Figure \ref{fig:H0_GHD}(b), which is bag-maximal. In particular, the vertex $v_2$ cannot be added to $B_{u_0}$: indeed, adding $v_2$ to $B_{u_0}$ would violate the connectedness condition, since $v_2$ is not in $B_{u_1}$ but in $B_{u_2}$.  
\hfill$\Diamond$
\end{example}

  For the following arguments, the reader is advised to be careful in distinguishing between the bag $B_u$ of a node $u$ and the set of vertices $B(\lambda_u)$ that are covered by the integral edge cover $\lambda_u$.
Before we prove a crucial lemma, 
we introduce some useful notation: 
\begin{definition}
\label{def:critical_ghd}
Let $\mcH=\defGHD$  
be a GHD
of a hypergraph $\HH$. Moreover, let 
$u$ be a node in $\mcH$ and 
let $e \in \lambda_u$ such that 
$e \setminus B_u \neq \emptyset$ holds. 
Let $u^*$ denote the node closest to $u$, such that $u^*$ covers $e$, i.e., 
$e \subseteq B_{u*}$. Then,
we call the path $\pi = (u_0,u_1,\ldots,u_\ell)$ with $u_0 = u$ and $u_\ell = u^*$ the
{\em critical path\/} of $(u,e)$ denoted as $\critp(u,e)$.
\end{definition}

\begin{lemma}
\label{lem:crit}
Let $\mcH=\defGHD$  be a bag-maximal GHD of a hypergraph 
$\HH = (V(H),$ $E(H))$, 
let $u \in T$, $e \in \lambda_u$, and $e \setminus B_u \neq \emptyset$. 
Let $\pi = (u_0, u_1, \dots, u_\ell)$ with $u_0 = u$  be the critical path of $(u,e)$.
Then the following equality holds.
$$e\cap B_u=\ \ \ \  e\,\cap\bigcap_{j=1}^\ell B(\lambda_{u_j})
$$
\end{lemma}

\begin{proof}
  ``$\subseteq$'':
Given that $e \subseteq B_{u_\ell}$ and by the connectedness 
condition, 
$e \cap B_u$ must be a subset of $B_{u_j}$ for every $j \in \{1,\dots,\ell\}$. 
Therefore, $e\cap B_u \subseteq e \cap \bigcap_{j=1}^\ell B(\lambda_{u_j})$ holds.

\smallskip

\noindent
``$\supseteq$'': Assume to the contrary that there exists some vertex 
$v \in e$ with  $v \not\in B_u$ but 
$v \in \bigcap_{j=1}^\ell B(\lambda_{u_j})$. 
By $e \subseteq B_{u_\ell}$, we have $v \in B_{u_\ell}$.
By the connectedness condition, along the path $u_0, \dots, u_\ell$ with $u_0 =u$, there exists $\alpha \in \{0, \dots, \ell-1\}$, s.t.\
$v \not \in B_{u_{\alpha}}$ and $v \in B_{u_{\alpha+1}}$. However, by the assumption, $v \in \bigcap_{j=1}^\ell B(\lambda_{u_j})$ holds. 
In particular, $v \in B(\lambda_{u_{\alpha}})$. Hence, we could 
safely add $v$ to $B_{u_{\alpha}}$ without violating the connectedness condition nor any other 
GHD condition. This contradicts the bag-maximality of $\mcH$. 
\end{proof}
\begin{example}[Example \ref{ex:ghw2} continued]
\label{ex:ghw4}
Consider root node $u$ of the GHD in Figure \ref{fig:H0_GHD}(b). 
We have $e_2 \in \lambda_u$ and $e_2 \setminus B_u = \{v_2\} \neq \emptyset$. 
As $e_2$ is covered by $u_2$, 
the critical path of $(u,e_2)$ is $\pi = (u,u_1,u_2)$.
It is easy to verify that 
$e_2 \cap B_u = e_2 \cap (e_3 \cup e_7) \cap (e_8 \cup e_2) = \{v_3,v_9\}$ indeed holds. 
\hfill$\Diamond$
\end{example}

Lemma~\ref{lem:crit} characterizes the overlap of an edge with a bag as an intersection of unions. However, to utilize the proposed intersection constraints, we would prefer unions of intersections instead. A straightforward transformation from an intersection of unions to a union of intersections may introduce certain redundant terms that we would like to avoid for technical reasons. We therefore employ a particular transformation, via the $\cupcap$-trees defined below, that avoids such redundant terms in the union.

\begin{definition}
  \label{def:uitree}
  Let $H$ be a hypergraph, $e$ an edge of $H$ and let
  $Q_1, \dots, Q_{\ell}$ be sets of edges.
  The $\cupcap$-tree $T$ of $e, Q_1 \dots, Q_\ell$ is the output of Algorithm~\ref{alg:uitree} with inputs
  $e, Q_1, \dots Q_\ell$. We refer to the set of all leaves of $T$ as $leaves(T)$.
\end{definition}

\begin{algorithm}[ht]
\SetKwData{Left}{left}\SetKwData{This}{this}\SetKwData{Up}{up}
\SetKwFunction{Union}{Union}\SetKwFunction{FindCompress}{FindCompress}

\SetKwData{N}{N}

\SetKwInput{Output}{output}
\SetKwInput{Input}{input}

\Input{An edge $e \in E(H)$, sets $Q_1, \dots, Q_\ell$ of edges of $E(H)$}
\Output{$\cupcap$-tree $T$ of $e,Q_1, \dots Q_\ell$}
\BlankLine
\tcc{Initialization: compute $(N,E)$ for $T_0$ }
$N \la \{r\}$\;
$E \la \emptyset$\;
$\lab(r) \la \{e \}$\;
$T \la (N,E)$\;
\BlankLine
\tcc{Compute $T_j$ from $T_{j-1}$ in a loop over $j$}
\For{$j\leftarrow 1$ \KwTo $\ell$}{
  \ForEach{leaf node $p$ of $\;T$}{
    \If{$\lab(p) \cap Q_{j}  = \emptyset$}{
          Let $Q_{j} = \{e_{j 1}, \dots,  e_{j h_j}\}$\;
Create new nodes $\{p_1,  \dots,  p_{h_j}\}$\;
           \lFor{$\alpha \la 1$ \KwTo $h_j$}{$\lab(p_\alpha) \la \lab(p_\alpha) \cup \{ e_{j \alpha} \}$}
           $N \la N \cup \{p_1,  \dots,  p_{h_j}\}$\;
           $E \la E \cup \{(p,p_1),  \dots,  (p,p_{h_j})\}$\;
}
}
$T\la(N,E)$\;
}
\caption{\uitree}\label{alg:uitree}
\end{algorithm}%

\begin{lemma}
  \label{lem:uitree}
  Let $H$ be a hypergraph, $e$ an edge of $H$ and let $Q_1, \dots, Q_\ell$ be sets of edges. Let $T$ be the $\cupcap$-tree of $e,Q_1,\dots,Q_\ell$, then
   \[
     e \cap \bigcap_{j=1}^\ell \bigcup Q_j = \bigcup_{p \in leaves(T)} \bigcap \lab(p)
   \]
\end{lemma}
\begin{proof}
  Proof is by induction over $\ell$. For $\ell = 0$, we have
  $leaves(T) = \{r\}$ and $\lab(r)=e$ and the statement trivially
  holds. For $0 \leq j \leq \ell$, let $T_j$ denote the $\cupcap$-tree of
  $e,Q_1,\dots,Q_j$. Suppose the statement is true for $\ell-1$,
  then we observe the following equality
  \[
    e \cap \bigcap_{j=1}^{\ell} \bigcup Q_j =  \left(e \cap \bigcap_{j=1}^{\ell-1} \bigcup Q_j\right) \cap \bigcup Q_\ell =
        \bigcup_{p \in leaves(T_{\ell-1})} \left(\bigcap \lab(p) \cap \bigcup Q_\ell\right)
  \]
  where the right equality follows from the induction hypothesis and distribution of $\bigcup Q_\ell$ over the union over the leaves of $T_{\ell-1}$.
  Now, consider a leaf $p \in leaves(T_{\ell-1})$. The construction of $T_\ell$ either adds new leaves $new_p$ as children of $p$, or $p$ remains a leaf in $T_\ell$. We claim that in the first case $\bigcap \lab(p) \cap \bigcup Q_\ell = \bigcup_{p' \in new_p} \bigcap \lab(p')$ and in the second case, $\bigcap \lab(p) \cap \bigcup Q_\ell = \bigcap \lab(p)$. If the claim holds, we have the following equality and the statement follows immediately.
\[
    \bigcup_{p \in leaves(T_{\ell-1})} \left( \bigcap \lab(p)  \cap \bigcup Q_\ell\right)  = \bigcup_{p \in leaves(T_{\ell})} \bigcap \lab(p)
  \]  

  What is left is to verify the claim. The case where new children are added to $p$ is straightforward by distributivity as
  the label of each new child corresponds to a term of the union
$\left(\bigcap \lab(p) \cap e_{\ell 1} \right) \cup \cdots \cup \left(\bigcap \lab(p) \cap e_{\ell h_\ell} \right) = \bigcap \lab(p)  \cap \bigcup Q_\ell$.
  If $p$ remains a leaf, then we have $\lab(p) \cap Q_\ell \neq \emptyset$. Thus, $\bigcap \lab(p) \subseteq \bigcup Q_\ell$ and therefore $\bigcap \lab(p) \cap \bigcup Q_\ell = \bigcap \lab(p)$.

\end{proof}

Throughout the rest of this paper, we will
be interested in how bags can be represented as combinations of
edges. In particular, we will see that, under the various restrictions
introduced at the beginning of this section, we are able to bound the
representation of bags as \emph{unions of intersections} of
edges. After introducing some notation for such unions of intersections we
can show how the LogBMIP allows for a bounded representation of bags
for GHDs.  The main result then follows by using this representation
to compute an appropriate set of candidate bags, to which Algorithm~\ref{alg:bre}
from
Section~\ref{sect:framework} can then be applied.

\begin{definition}
  Let $H$ be a hypergraph. A  \emph{$(q, p)$-set} $X \subseteq V(H)$ is a set of the form
  $X = X_1 \cup \cdots \cup X_{q'}$ with $q' \leq q$ and where every $X_i$ is the intersection
  of at most $p$ edges. We will use $q$-set as shorthand for $(q,1)$-set.
\end{definition}

We will repeatedly make use of the fact that, by the idempotence of
union and intersection, we can w.l.o.g. assume a $(q,p)$-set to be the
union of exactly $q$ terms, each consisting of the intersection of
exactly $p$ edges.

\begin{lemma}
  \label{lem:mainintersect}
  Let $H$ be a hypergraph with $\cmiwidth{c}{H}\leq i$ and let
 $B_1, \dots, B_{\ell}$ be $(q, 1)$-sets.
 For each $e \in E(H)$, there exists a $(q^{c-1}, c)$-set $I$ and a subedge $e' \subseteq e$ with $|e'| \leq i\,q^{c}$,
 such that 
 \[
   e \cap \bigcap_{j=1}^{\ell} B_j =  I \cup e'
 \]
   Furthermore, $e'$ is the union of at most $q^{c}$ subsets of intersections of exactly $c$ edges.

\end{lemma}
\begin{proof}
  For $j \in [\ell]$, fix a set of edges $Q_j = \{e_{j1},\dots,e_{jq} \}$ such that
  $B_j = e_{j1} \cup \cdots \cup e_{jq}$. Let $T$ be the $\cupcap$-tree of $e,Q_1,\dots,Q_\ell$. For a node $p$ of $T$, we refer to the number of edges in the path from the root to $p$ as the depth of $p$, or $depth(p)$.
  Note that by construction, $|\lab(p)| = depth(p)+1$ for each node $p$ in $T$.
  We consider the following partition of $leaves(T)$: let $SMALL$ contain all the leaves of $T$ at depth at most $c-1$ and, conversely, let $FULL$ be the set of leaves at depth at least $c$.

  \sloppy
  By Lemma~\ref{lem:uitree},
  $e \cap \bigcap_{j=1}^\ell B_j = \left( \bigcup_{p \in SMALL} \bigcap
    \lab(p) \right) \cup \left(\bigcup_{p \in FULL} \bigcap \lab(p)
  \right)$.  Hence, to prove the lemma, it suffices to show that
  $\bigcup_{p \in SMALL} \bigcap \lab(p)$ is a $(q^{c-1}, c)$-set and
  that $|\bigcup_{p \in FULL} \bigcap \lab(p)| \leq i q^{c}$ holds.

  \medskip
  \textsc{Claim A.} 
  \emph{$\bigcup_{p \in  SMALL}\bigcap \lab(p)$ is a $(q^{c-1}, c)$-set.} 
  
  \smallskip
  \textsc{Proof of Claim A.}
  Since each of the sets $Q_j$ has at most $q$ members, every node in $T$ has at most $q$ children. Hence, there are at most $q^{c-1}$ leaves at depth $\leq c-1$ and therefore $|SMALL| \leq q^{c-1}$. Furthermore, we have $|\lab(p)|=depth(p)+1 \leq c$, i.e., each intersection has at most $c$ terms.
    $\hfill\diamond$

  \medskip
  \textsc{Claim B.} \emph{$|\bigcup_{p \in FULL} \bigcap \lab(p)| \leq i q^{c}$.}

  \smallskip
  \textsc{Proof of Claim B.}
  First, observe that for each $p \in FULL$, there exists a node $p'$
  in $T$ at depth $c$ such that $\lab(p) \supseteq \lab(p')$ and
  therefore also $\bigcap \lab(p) \subseteq \bigcap \lab(p')$. Note that there are at most $q^{c}$ nodes $p'$ at depth $c$.
  Furthermore, because $|\lab(p')|=c$ and we assume $\cmiwidth{c}{H}\leq i$, it holds that $|\bigcap \lab(p')| \leq i$.
  In total, we thus have that  $\bigcup_{p \in FULL} \bigcap \lab(p)$ is a union of sets $X_p$ such that each 
$X_p$ is the subset of one out of at most $q^{c}$ vertex sets, and 
each of these vertex sets has cardinality at most $i$. 

\end{proof}

For a given edge cover $\lambda_u$ and arbitrary edge $e \in E(H)$ with $\lambda_u(e) = 1$, 
Lemma \ref{lem:crit} gives us a representation of $e \cap B_u$ of the form 
$e\,\cap\bigcap_{j=1}^\ell B(\lambda_{u_j})$. Clearly, the sets $B(\lambda_{u_j})$ are $(k,1)$-sets, i.e., 
unions of (up to) $k$ edges. We can therefore apply 
Lemma~\ref{lem:mainintersect} by taking $B_j = B(\lambda_{u_j})$ and $q = k$ to get a representation of the 
form $I \cup e'$ for each of the possible subedges $e\,\cap B_u$ that may
ever be used in a bag-maximal \compnf GHD. This idea is formalized in the following lemma,
where
we identify a polynomially big family of vertex sets $\mathbf{S} \subseteq 2^{V(H)}$, such that the 
bags of any bag-maximal \compnf GHD of $H$ must be a member of this family.

\begin{lemma}
  \label{lem:ghdreal}
  Let $H$ be a hypergraph with $\cmiwidth{c}{H}\leq b = a \log(|V(H)| + |E(H)|)$ and fix an integer $k > 0$.
  There exists a set $\mathbf{S} \subseteq 2^{V(H)}$, which 
can be computed in polynomial time (for fixed $a$, $c$, and $k$), such that 
$\blockreal{\mathbf{S}} \neq \emptyset$ if and only if $\ghw(H) \leq k$.

  Furthermore, for any bag-maximal \compnf GHD $\defGHD$ of $H$ of width $\leq k$, we have $\defTD \in \blockreal{\mathbf{S}}$.
\end{lemma}
\begin{proof}
  Let $n = ||H||$ refer to the size of $H$ and $m = |E(H)|$.
  We define the following sets:

  \smallskip
  $\mathcal{I} = \{ I \mid I \mbox{ is a } (k^{c-1},c)\mbox{-set}\}$

  $\mathcal{C} = \{ e' \mid \mbox{there exist distinct } e_1,\dots,e_c \in E(H)$, 
  such that $e' \subseteq e_1 \cap \cdots \cap e_c \}$
  
  $Sub = \{ I \cup \bigcup_{j=1}^{k^{c}} C_j \mid I \in \mathcal{I},\, C_1,\dots, C_{k^{c}} \in \mathcal{C},$  
  and $I \cup \bigcup_{j=1}^{k^{c}} C_j \subseteq e$ for some $e \in E(H) \}$.
  \medskip
  
By construction, $Sub$ contains only $E(H)$ and subedges of $H$.
  There are no more than $m^{ck^{c-1}}$ possible $(k^{c-1},c)$ sets. 
Also, by the condition $\cmiwidth{c}{H} \leq b$, we have
  $|\mathcal{C}| \leq 2^bm^c$ and, therefore, $Sub$ has at most
  $m^{ck^{c-1}} m^{ck^c} 2^{bk^c}$ elements. In our concrete case, where $b=a \log n$, we have $2^b = n^{a}$,
  and $Sub$ can be computed in
  $O(n^{f_1(a,c,k)})$ time for some function $f_1$ 
by straightforward enumeration. We can now
  construct our desired set $\mathbf{S}$ as the set of all unions of
  up to $k$ elements of $Sub$. It follows that the construction of $\mathbf{S}$ is
  possible in $O(n^{f_2(a,c,k)})$ time for some function $f_2$.

It remains to show that this $\mathbf{S}$ indeed has the
  property that (1) $\blockreal{\mathbf{S}} \neq \emptyset$ if and only if
  $\ghw(H) \leq k$ and (2) for any bag-maximal \compnf GHD $\defGHD$ of $H$ of width $\leq k$, we have $\defTD \in \blockreal{\mathbf{S}}$.

First, assume $\blockreal{\mathbf{S}} \neq
  \emptyset$. Then there exists a TD of $H$ where each bag is in
  $\mathbf{S}$ and therefore a union of $k$ subedges of
  $H$. Hence,  every bag of the TD can also be covered by $k$ edges
  of $H$ and thus can clearly be turned into a GHD of width at most
  $k$.

Now, assume $\ghw(H) \leq k$.  Let $\mcH=\defGHD$ be a bag-maximal
  \compnf GHD of width at most $k$ and let $u$ be a node of $T$. By
  Lemma~\ref{lem:nfghd}, such a GHD always exists if $\ghw(H) \leq
  k$. W.l.o.g. we assume that
  $B(\lambda_u) = e_1 \cup \cdots \cup e_k$ and, therefore, also
  \[
    B_u = B_u \cap B(\lambda_u) = (B_u \cap e_1) \cup \cdots \cup (B_u \cap e_k)
  \]
  We will show that $B_u \cap e_{j} \in Sub$ for each
  $j \in [k]$ and therefore also $B_u \in \mathbf{S}$. The case where
  $e_{j} \cap B_u = e_{j}$ is trivial as
  $e_{j} \in E(H) \subseteq Sub$. So, let $e_{j}$ be any edge from
  this representation of $B(\lambda_u)$ where
  $e_{j} \cap B_u \neq e_{j}$.  By Lemma~\ref{lem:crit} the following
  equality holds for the critical path $(u, u_1,\dots,u_\ell)$ of
  $(u, e_j)$
  \[
    e_j\cap B_u=\ \ \ \  e_j\,\cap\bigcap_{j=1}^\ell B(\lambda_{u_j})
  \]
  By assumption, every such $B(\lambda_{u_j})$ is a $k$-set and
  hence, by Lemma~\ref{lem:mainintersect}, we know that
  $e_{j} \cap \bigcap^\ell_{j=1} B(\lambda_{u_j})$ is precisely
  the union of an element of $\mathcal{I}$ and at most $k^{c}$ sets from $\mathcal{C}$, i.e.,
  $e_{j} \cap B_u \in Sub$.  Since the choice of $u$ was
  arbitrary, every bag of $\mcH$ is contained in $\mathbf{S}$ and we
  have $\mcH \in \blockreal{\mathbf{S}}$.
\end{proof}

\emph{Proof of Theorem~\ref{theo:LogBMIP}}: We assume that $\classC$ enjoys
the LogBMIP, i.e., for every $H \in \classC$, we have
$\cmiwidth{c}{H}\leq a \log n$. To solve the \rec{GHD,\, k} problem, we can then
simply compute, in polynomial time, the set $\mathbf{S}$ from
Lemma~\ref{lem:ghdreal} and decide whether
$\blockreal{\mathbf{S}} \neq \emptyset$. By
Theorem~\ref{thm:frameworkalg} this is also tractable and therefore, so is
the whole procedure.
\qed

We have defined 
in Section~\ref{sect:introduction}
the degree $d$ of a hypergraph $H$. 
We now consider hypergraphs of bounded degree.

\begin{definition}\label{def:bdegree}
We say that a hypergraph $H$ has
the 
{\em $d$-bounded degree property (\dBDP)} if \linebreak 
$\ddegree{\HH}\leq d$ holds.

Let $\classC$ be a class of hypergraphs. 
We say that $\classC$ 
has the  
{\em bounded degree property (BDP)} 
if there exists a constant $d$
such that
every hypergraph $H$ in $\classC$
has the \dBDP.
\end{definition}

The class of 
hypergraphs of bounded degree is an interesting
special case of the class of hypergraphs enjoying the BMIP.
Indeed, suppose that each vertex in a hypergraph $H$ 
occurs 
in at most $d$ edges for some constant $d$. 
Then the intersection of $d+1$ hyperedges is always empty. 
The following corollary is thus immediate.

\begin{corollary}
\label{corol:degreeGHD}
For every class $\classC$ of hypergraphs of bounded degree, for each constant 
$k$, 
the problem  \rec{GHD,\,$k$} is tractable.
\end{corollary}

In case of the BMIP,
the upper bound on $\mathbf{S}$ in the proof of Lemma~\ref{lem:ghdreal},
is $2^{f(i,k,c)}m^{g(k,c)}$ for some function $g$.
Recall from Theorem~\ref{thm:frameworkalg}
that \blockreal{\textbf{S}}$\neq \emptyset$ can be decided in time complexity
  that is polynomial in $|\mathbf{S}| \cdot |V(H)|$.
We thus get the following parameterized complexity result.

\begin{theorem}\label{theo:BMIPftp}
For constants $k$ and $c$,
the \rec{GHD,\,$k$} problem is fixed-parameter tractable 
w.r.t.\ the parameter $i$ for hypergraphs enjoying the \icBMIP, i.e., 
in this case, 
\rec{GHD,\,$k$} can be solved in time $\calO(h(i) \cdot \poly(n))$, where 
$h(i)$ is a function depending  on the intersection width $i$ only
and 
$\poly(n)$ is a function that depends polynomially on the size $n=||H||$ of 
the given hypergraph $H$.
\end{theorem}

For practical purposes, this is of particular interest in case of the BIP, i.e., $c=2$.
Here, the construction of $Sub$ can be significantly simplified since,  for any
critical path $\critp(u,e) = (u, u_1, \dots, u_\ell)$, it is enough to consider only one step, i.e.,
  $e \cap B_u \subseteq e \cap B(\lambda_{\hat{u}}) = e \cap (e_{\hat{u} 1} \cup \cdots \cup e_{\hat{u} k})$,
for some node  $\hat{u}$, 
such that all edges $e_{\hat{u} j}$ are distinct from $e$.
This is the case for the following reasons: first, by connectedness, 
$e \cap B_u \subseteq e_{u_j 1} \cup \cdots \cup e_{u_j k}$ for every node $u_j$ on this critical path. 
And second, there is at least one node $\hat{u}$ on this critical path such that $e \not\in \lambda_{\hat{u}}$. Indeed, if $e$ were present in $\lambda_{u_j}$ for every $j$, then we could add $e$ to every bag $B_{u_j}$ on the critical path
without violating the connectedness condition. This would contradict the assumption of bag-maximality. 

By the BIP, we have $|e \cap e_{\hat{u} j}| \leq i$ for every $j$ and, therefore, 
we can simply compute all the subsets of $e \cap B(\lambda_{\hat{u}})$ in polynomial time to construct 
$Sub$ as follows: 
$$Sub= E(H) \cup \bigcup_{e\in E(\HH)}   
\Big({\bigcup_{e_1,\ldots,e_j\in (E(\HH) \setminus \{e\}),\, j\leq k }} 
2^{(e\cap (e_1\cup\cdots\cup e_j))}\Big),$$
i.e., $Sub$ contains $E(H)$ and all subsets of intersections of edges
$e \in E(H)$ with unions of $\leq k$ edges of $H$ different from $e$.
Here, the BIP is only used implicitly in the sense that we can be sure that the above expression is 
computable in polynomial time, because $|e\cap (e_1\cup\cdots\cup e_j)| \leq ji \leq ki$ holds by the BIP.
The explicit use of the BIP allows for a yet simpler (more coarse-grained) way to define $Sub$ 
as follows: 
$$
Sub= E(H)  \cup \{ e' \mid e' \subseteq e \mbox{ for some $e \in E(H)$ and } |e'| \leq ki\}.
$$

\section{Tractable Cases of FHD Computation}
\label{sect:fhd-bip}
\label{sect:fhd-exact}
\label{sect:fhd}

\subsection{Initial Considerations and Lemmas}

In Section \ref{sect:ghd}, we have shown that under certain conditions (with the BIP and BDP as most specific and the LogBMIP as most general conditions) the problem of computing a GHD of width $\leq k$ can be reduced to finding a \compnf candidate tree decomposition for an appropriate set of candidate bags.  The key to this problem reduction was to enumerate a set of subedges, which allowed us to enumerate all bags of possible bag-maximal GHDs of width $\leq k$. 
When trying to 
carry over these ideas from GHDs to FHDs, we encounter {\em two major challenges\/}:
Can we adapt our approach for GHDs to work with FHDs? And is it possible to find bounded representations of all the sets of vertices that can be \emph{fractionally} covered with weight $\leq k$?

For the second challenge, recall from the GHD-case that the possible bags $B(\lambda_u)$ 
could be easily computed from the given set of subedges, since each $\lambda_u$ can choose at most $k$ subedges.
In contrast, for a fractional cover $\gamma_u$, we do not have such a bound on the number of edges with non-zero weight.  It is easy to exhibit a family $(H_n)_{n \in \mathbb{N}}$ of hypergraphs where it is advantageous to have unbounded $\cov(\gamma_n)$ even if $(H_n)_{n \in \mathbb{N}}$ enjoys the BIP, as the following example illustrates:

\begin{example}
\label{ex:fhwLongEdge}
{\em 
Consider the family $(H_n)_{n \in \mathbb{N}}$ of hypergraphs with $H_n 
=(V_n,E_n)$ defined as follows:

\smallskip
$V_n = \{v_0, v_1, \dots, v_n\}$

$E_n= \{ \{v_0,v_i\} \mid 1 \leq i \leq n\} \cup \{\{v_1, \dots, v_n\}\}$

\smallskip

\noindent
Clearly $\iwidth{H_n} = 1$, but an optimal fractional edge cover of $H_n$ is 
obtained by the following mapping $\gamma$ with 
$\cov(\gamma) = E_n$:

\smallskip

$\gamma(\{v_0, v_i\}) = 1/n$ for each $i \in \{1, \dots, n\}$ and

$\gamma(\{v_1,\dots,  v_n\}) = 1 - (1/n)$ 

\smallskip
\noindent
such that $\weight(\gamma) = 2 - (1/n)$, which is 
optimal in this case. 
}
\end{example}

\nop{************************
We refer the reader to a recent article \cite{DBLP:journals/corr/FischlGP16} for 
most of the basic definitions needed also in this work.
Below, we recall some crucial 
definitions and introduce some 
additional ones mainly to fix the terminology.

\smallskip

\begin{definition}
Let $H = (V(H),E(H))$ be a hypergraph and let $\gamma \colon E(H) \ra [0,1]$ be an
edge-weight function for $H$. For $v\in V(H)$, 
we write $\gamma(v)$ to denote the total weight that $\gamma$ assigns to $v$. Moreover, 
we write by $B(\gamma)$ to denote  the set of all 
vertices {\em covered\/} by $\gamma$, i.e.:
\[ \gamma (v) = \sum_{e\in E(H), v\in e} \gamma(e).
\]
\[ B(\gamma) = \left\{ v\in V(H) \mid \sum_{e\in E(H), v\in e} \gamma(e) \geq 1 \right\} = 
\left\{ v\in V(H) \mid \gamma(v) \geq 1 \right\}
\]
\end{definition}

\begin{definition}
 \label{def:FHD}
Let $H=(V(H),E(H))$ be a  hypergraph. 
A {\em fractional hypertree decomposition\/} (FHD) of $H$
is a tuple 
$\left< T, (B_u)_{u\in N(T)}, (\gamma)_{u\in N(T)} \right>$, such that 
$T = \left< N(T),E(T)\right>$ is a rooted tree and 
the 
following conditions hold:

\begin{enumerate}
 \item for each $e \in E(H)$, there is a node $u \in N(T)$ with $e \subseteq 
B_u$;
 \item for each $v \in V(H)$, the set $\{u \in N(T) \mid v \in B_u\}$ is 
connected 
in $T$;
 \item for each $u\in N(T)$, $\gamma_u$ is an edge-weight function $\gamma_u \colon E(H) 
\ra 
[0,1]$
with 
$B_u \subseteq  B(\gamma_u)$.
\end{enumerate}
\end{definition}

\medskip
\noindent
The width of an FHD is the maximum weight of the functions $\gamma_u$,
over all nodes $u$ in $T$. Moreover, the fractional hypertree width of
$H$ (denoted $\fhw(H)$) is the minimum width over all FHDs of $H$.
Condition~2 is often called the ``connectedness condition'' and the
set $B_u$ is usually referred to as the ``bag'' at node $u$. It is
easy to see that the existence of an FHD of width $\leq k$ is the same
as the existence of a TD where every bag has fractional cover number
$\leq k$.

Note that, in contrast to hypertree decompositions (HDs),
the underlying tree $T$ of an FHD does not need to be rooted. For the
sake of uniformity, we assume that also the tree underlying an FHD is
rooted with the understanding that the root is arbitrarily chosen.
Finally, by slight abuse of notation, we shall write $u \in T$ short
for $u \in N(T)$. Hence, FHDs will be referred to as $\defFHD$ or,
simply as $\defFHDohne$.

We sometimes identify sets of edges with hypergraphs. If a set of edges $E$ is used, where instead a hypergraph is expected, then we mean the hypergraph $(V,E)$, where $V$ is simply the union of all edges in $E$. 
Finally, for a set $E$ of edges, it is convenient to write $\bigcup E$ (and $\bigcap E$, respectively)  to denote the set of vertices obtained by taking the union
(or the intersection, respectively) of the edges in $E$.
************************} %

Below we show that, for cases where the second challenge can be resolved (i.e., we can establish an upper bound on 
$\cov(\gamma_u)$ for all nodes $u$ in an FHD), the check-problem of FHDs can be essentially reduced to the GHD case. The following lemma is thus the crucial tool for the remainder of this paper. It tells us that, in case of the BMIP, it suffices to 
determine a collection of vertex sets $\mathbf{Q} \subseteq 2^{V(H)}$, such that the bags $B_u$ of an FHD can be taken from subsets of the sets in $\mathbf{Q}$.

\begin{lemma}
  \label{lem:fhwreal}
Let $H$ be a hypergraph with $\cmiwidth{c}{H}\leq i$ for constants $c$ and $i$,
and fix
$k,\,q \geq 1$.
Suppose there exists a FHD \mcF, where for each bag $B_u$, there exists a $q$-set $Q_u$ such that $B_u \subseteq Q_u$ and $\rho^*(Q_u) \leq k$.
  Then, there exists a
  polynomial-time computable set $\mathbf{S}$ such that
  $\blockreal{\mathbf{S}} \neq \emptyset$ if and only if $\fhw(H) \leq k$.

\end{lemma}
\begin{proof}
We first define $\mathbf{S}$ by making use of results from the GHD-case: 
let $\mathbf{S}_{ghd}$ be the set from Lemma~\ref{lem:ghdreal} such that
$\blockreal{\mathbf{S}_{ghd}} \neq \emptyset$ if and only if
  $ghd(H) \leq q$.  We  know that
  such a set exists and can be computed in polynomial time.  We can
  then obtain the required set
  $\mathbf{S} = \{ S \in \mathbf{S}_{ghd} \mid \rho^*(S) \leq k\}$ by
  solving a linear program for every element of $\mathbf{S}_{ghd}$.
  Hence, also $\mathbf{S}$ can be computed in polynomial time.

We now argue that $\blockreal{\mathbf{S}} \neq \emptyset$ if and only if $\fhw(H) \leq k$.
 Suppose $\blockreal{\mathbf{S}} \neq \emptyset$, then clearly there
  is a TD where every bag has fractional cover number at most $k$ and, therefore,
  $\fhw(H) \leq k$.

For the other direction, suppose $\fhw(H) \leq k$, and let
$\mcF=\defFHD$ be the FHD which exists by the assumption of the lemma. 
That is, for each node $u \in T$, $B_u$ is a subset of 
some $Q_u$ with $\rho^*(Q_u) \leq k$.
Moreover, since $Q_u$ is a $q$-set, each 
$Q_u$ is of the form $Q_u = e_{u1} \cup \cdots \cup e_{uq}$. 
Now let $\lambda_u = \{e_{u1}, \dots, e_{uq}\}$ for each
  $u \in T$. It is easy to see that $\mcH = \defGHD$ is a GHD of
  $H$ with $ghw(\mcH) \leq q$.  Note that $\mcH$ is not necessarily
  bag-maximal or in \compnf. However, following the procedure
  described in Lemma~\ref{lem:nfghd}, there exists a bag-maximal
  \compnf $\mcH' = \defGHDprime$ with $ghw(\mcH') \leq q$ and
  therefore $\defTDprime \in \blockreal{\mathbf{S}_{ghd}}$.
  
  Now recall that the 
  transformation into a bag-maximal
  \compnf $\mcH'$ from  Lemma~\ref{lem:nfghd} never increases bags. That is,  
every
  bag of $\mcH'$ is a subset of a bag of $\mcH$ and, hence,  
  still a subset of some $Q$ with $\rho^*(Q) \leq k$. Therefore, \defTDprime is in \compnf and each of its bags is in $\mathbf{S}_{ghd}$ and can be fractionally covered with weight $k$, i.e., $\defTDprime \in \blockreal{\mathbf{S}}$.
\end{proof}

Similarly as in the GHD-case, we will have to deal with unions of intersections of edges also in the FHD-case. 
The following definition and the accompanying two lemmas will be convenient for this purpose.

\begin{definition}
  For a hypergraph $\HH$ we write $\HH^\cap$ for the closure of $H$
  under intersection of edges.
\end{definition}

It will be important to observe that adding subedges does not change the fractional hypertree width of a hypergraph.
Every FHD of $\HH$ is still an FHD of $\HH^\cap$ and every
FHD of $\HH^\cap$ can be easily transformed into an FHD of $\HH$. 
It follows that we always have $\fhw(\HH) = \fhw(\HH^\cap)$.

\begin{lemma}
  \label{lem:boundsupp}
    Let $\HH$ be a hypergraph and $\gamma$ a fractional edge cover of a set $S$ of vertices.
  Then $B(\gamma)$ is a $2^{|\supp(\gamma)|}$-set w.r.t. $\HH^\cap$.
\end{lemma}
\begin{proof}
  Let us call a subset $E \subseteq \supp(\gamma)$ full if
  $\sum_{e \in E} \gamma(e) \geq 1$ and let $\mathbf{FS}$ contain all the full subsets. 
For every $E \in \mathbf{FS}$, we have
  $\bigcap E \subseteq B(\gamma)$ and it is straightforward to verify
  the following equality:
  $$B(\gamma) = \bigcup_{E \in \mathbf{FS}} \bigcap E$$
  It is then enough to observe that $\bigcap E \in E(\HH^\cap)$ and that there are at most $2^{|\supp(\gamma)|}$ full subsets.
\end{proof}

\begin{lemma}
  \label{lem:bmipclosure}
  Let $\HH$ be a hypergraph with $\cmiwidth{c}{\HH}\leq i$. 
Then $\HH^\cap$ can be computed in polynomial time for fixed $c$ and $i$.
  Moreover, $2^c$\emph{-miwidth}$(H^{\cap}) \leq i$ holds.
\end{lemma}
\begin{proof}
First of all note that if $|E(H^\cap)| \leq 2^c$, then 
also $|E(H)| \leq 2^c$ and we can trivially compute $H^\cap$ in polynomial time 
by simply computing all possible intersections of edges in $E(H)$.
Moreover, in this case, 
the condition
$2^c$\emph{-miwidth}$(H^{\cap}) \leq i$ is void, since there are no distinct 
$2^c$ edges in $H^{\cap}$.

Now let $\alpha = 2^c$ and consider an intersection $I$ of $\alpha$
distinct edges of $H^\cap$ of the form
$I = e'_1 \cap \cdots \cap e'_{\alpha}$, where $e'_j\in E(H^\cap)$ for
$j \in [\alpha]$.
Each $e'_j \in E(H^\cap)$ is
an intersection of edges from $E(H)$, i.e., there exists 
$\mathcal{E}_j \subseteq E(H)$ such that
$e_j' = \bigcap_{e\in \mathcal{E}_j} e$. Clearly,
$I = \bigcap \big( \bigcup_{j=1}^\alpha \mathcal{E}_j\big)$.
  
We claim that $|\bigcup_{j=1}^\alpha \mathcal{E}_j| \geq c$, i.e., 
$I$ is the intersection of at least $c$ edges from $E(H)$. 
Indeed,
less than $c$ distinct edges, there could only be less than $2^c-1$ non-empty sets
$\mathcal{E}_j$. It thus follows that $I$ is a subset of
an intersection of at least $c$ edges of $E(H)$. Thus, by $\cmiwidth{c}{H}\leq i$, 
we conclude that  $|I|\leq i$ holds. Hence, also in the case $|E(H^\cap)| > 2^c$, 
the condition $2^c$\emph{-miwidth}$(H^{\cap}) \leq i$ holds.

It remains to show that $E(H^{\cap})$ can be computed from $E(H)$ in polynomial time: let $m = |E(H)|$. 
Then there are 
less than  $m^c$ intersections of less than $c$ distinct edges from $E(H)$ and these 
intersections can clearly be computed in polynomial time. 
In order to compute also the set of intersections of at least $c$ edges from $E(H)$, we proceed as follows: 
we first compute the set $\mathcal{I}_0$ of intersections of $c$ edges. 
Then, for every $j \in [m-c]$, we compute the set $\mathcal{I}_j$ of
intersections of $\alpha$ edges from $E(H)$ with $\alpha \in \{c, \dots, c +j\}$. 
By $\cmiwidth{c}{H}\leq i$, we know that $|\mathcal{I}_j| \leq 2^im^c$ holds for every $j \in [m-c]$.
Hence, 
all these intersections can clearly be computed in polynomial time (for fixed $c$ and $i$).
\end{proof}

In order to apply Lemma~\ref{lem:fhwreal}, we will have to prove that an 
appropriate constant $q$ indeed exists. The following two subsections will show that
such a constant indeed exists for classes of either bounded
intersection width or bounded degree. Furthermore, 
in Section \ref{sect:approx},
we will also use
the lemma to derive an approximation result for
classes that exhibit the BMIP.

\subsection{Computing FHDs Under Bounded Degree}
\label{sect:fhw-bd}

First, we will investigate the situation for classes of hypergraphs
whose degree is bound by a constant $d$.  We show that a fractional
edge cover $\gamma$ of a hypergraph $H$ has support $\supp(\gamma)$
bounded by a constant that depends only on $weight(\gamma)$ and
$d$. From there, tractability follows from the combination of
Lemma~\ref{lem:fhwreal} and Theorem~\ref{thm:frameworkalg}.

\begin{theorem}\label{thm:exactFHDBD}
  For every hypergraph class $\classC$ that has bounded degree, and
  for every constant $k \geq 1$, the \rec{FHD,\,$k$} problem is
  tractable, i.e., given a hypergraph $\HH \in \classC$, it is
  feasible in polynomial time to check $\fhw(\HH)\leq k$ and, if so,
  to compute an FHD of width $k$ of $\HH$.%
\end{theorem}

For this result, we need to introduce, analogously to edge-weight functions 
and edge covers in Section \ref{sect:prelim},
the notions of vertex-weight functions and vertex covers.

\begin{definition}
\label{def:cover-and-support}
A vertex-weight function $w$ for a hypergraph $\HH$ assigns a weight $w(v) \geq 0$ to each vertex $v$ of $\HH$. 
We say that $w$ is a 
{\em fractional vertex cover\/}
of $\HH$  if for each edge  $e\in E(\HH)$, $\Sigma_{v\in e} w(v)\geq 1$ holds.
For a vertex-weight function $w$ for hypergraph $\HH$, we denote by 
$\weight(w)$ its total weight, i.e. $\Sigma_{v\in V(\HH)}w(v)$.
The {\em fractional vertex cover number\/}  $\tau^*(\HH)$ is defined as the minimum 
$\weight(w)$ where $w$ ranges over all fractional  vertex covers of $\HH$.
The {\em vertex support\/} $\vsupp(w)$ of a hypergraph $\HH$ under a vertex-weight function $w$  is defined as $\vsupp(w)=\{v\in V(\HH) \, | \, w(v)>0\}$.
\end{definition}

For our result on bounded support, we will exploit the well-known dualities $\rho^*(\HH)=\tau^*(\HH^d)$ and $\tau^*(\HH)=\rho^*(\HH^d)$, where  $\HH^d$ denotes the dual of $\HH$. 
To make optimal use of this, we make, for the moment,  several assumptions.
First of all, we will assume that (1) hypergraphs have no isolated vertices and (2) no empty edges. 
Furthermore, we assume that 
(3) hypergraphs never have two distinct vertices of the same ``edge-type'' 
(i.e., we exclude that two vertices occur in precisely the same edges) and 
(4) they never have two distinct edges of the same ``vertex-type'' (i.e., we exclude duplicate edges). 

Assumptions (1) -- (4) can be safely made.
Recall that we are ultimately interested in the computation of an FHD of width $\leq k$ for 
given $k$. As mentioned above, without assumption (1), the computation of an edge-weight function and, hence, of an FHD of width $\leq k$ makes no sense. Assumption (2) does not restrict the search for a specific FHD since 
we would never define an edge-weight function with non-zero weight on an empty edge.
As far as assumption (3) is concerned, suppose that a hypergraph $H$ has groups of multiple vertices of identical edge-type. Then it is sufficient to consider 
the reduced hypergraph $H^-$ resulting from $H$ by ``fusing'' each such group to a single vertex. Obviously $\rho^*(H)=\rho^*(H^-)$, and each edge-weight function for $H^-$ 
can be extended in the obvious way to an edge-weight function of the same total weight to $H$.
Finally, assumption (4) can also be made w.l.o.g., since we can again 
define a reduced hypergraph $H^-$ resulting from $H$ by retaining only one edge from each group of identical edges. Then every edge cover of $H^-$ is an edge cover of $H$. Conversely, every 
edge cover of $H$ can be turned into an edge cover of $H^-$ by assigning to each edge $e$ in $H^-$ the sum of the weights of $e$ and all edges identical to $e$ in $H$.

Under our above assumptions  (1) -- (4), 
for every hypergraph $\HH$, the property $\HH^{dd}=\HH$ holds 
and there is an obvious one-to-one correspondence between the edges (vertices) of $\HH$ and the vertices (edges)  of $\HH^d$.
Moreover, there is an obvious one-to-one correspondence between the fractional edge covers of 
$\HH$ and the fractional vertex covers of $\HH^{d}$. In particular, if there is a fractional edge cover $\gamma$ for $\HH$, then its corresponding ``dual'' $\gamma^d$  assigns to each vertex $v$ of $\HH^d$ the same weight as to the edge in $\HH$ that is represented by this vertex and vice versa. 

Note that if we do not make assumptions (3) and (4), then there are hypergraphs $H$ with $H^{dd}\neq H$. For instance, consider the hypergraph $\HH_0$ with $V(\HH_0)=\{a,b,c\}$ and $E(\HH_0)=\{\,e=\{a,b,c\}\,\}$, i.e., property (3) is violated.
The hypergraph $\HH_0^d$ has a unique vertex $e$ and a unique hyperedge $\{e\}$. 
Hence, $H_0^{dd}$ is (isomorphic to) the hypergraph with a unique vertex $a$ and 
a unique hyperedge $\{a\}$, which is clearly different from the original hypergraph~$H_0$.

To get an upper bound on the support $\supp(\gamma)$ of a 
fractional edge cover of a hypergraph $H$, we make use of the following result for fractional vertex covers.
This result is due to Zolt{\'a}n F{\"u}redi~\cite{furedi1988}, who 
extended earlier results by Chung et al.~\cite{chung1988}. 
Below, we appropriately reformulate F{\"u}redi's result for our purposes: 

\begin{proposition}[\cite{furedi1988}, page 152, \ Proposition 5.11.(iii)]\label{prop:furedi}
For every hypergraph $\HH$ of rank (i.e., maximal edge size) $r$, and every fractional vertex cover $w$ for $\HH$  satisfying  $\weight(w)= \tau^*(\HH)$, the 
property $|\vsupp(w)| \leq r\cdot \tau^*(\HH)$ holds.
\end{proposition}

By  duality, exploiting the relationship $\rho^*(\HH)=\tau^*(\HH^d)$ and by recalling that the degree of $\HH$ corresponds to the rank of $\HH^d$, we immediately get the following corollary:

\begin{corollary}\label{cor:bsupp}
For every hypergraph $\HH$ of degree $d$, and every fractional edge cover $\gamma$ for $\HH$ satisfying  $\weight(\gamma) = \rho^*(\HH)$,   
the property
$|\supp(\gamma)|\leq d\cdot \rho^*(\HH)$ holds.
\end{corollary}

\emph{Proof of Theorem~\ref{thm:exactFHDBD}}: Let $\classC$ be a class of
hypergraphs with degree at most $d$ and let $\HH \in \classC$.
Observe that we can extend Corollary~\ref{cor:bsupp} to any vertex subset $V' \subseteq V(H)$ because $\rho^*(V')$ in $H$ is the same as $\rho^*(H[V'])$. The degree of the induced subhypergraph can not become greater than the degree of $H$.

Therefore, for every set of vertices
$V' \subseteq V(H)$, if $\rho^*(V') \leq k$, then $V'$ can be fractionally covered using at most $kd$ edges.
From Lemma~\ref{lem:boundsupp}
we see that any such $V'$ is a $2^{kd}$-set w.r.t. $H^\cap$
and so is any bag of an FHD of $H^\cap$ with width $\leq k$.
Recall that $\ddegree{\HH}\leq d$ also means $\cmiwidth{(d+1)}{H}=0$. By Lemma~\ref{lem:bmipclosure}, 
we have that $H^\cap$ can be computed in polynomial time and that $2^{d+1}$\emph{-miwidth}$(H^{\cap}) \leq 0$.

We can then use Lemma~\ref{lem:fhwreal} to obtain a set $\mathbf{S}$ in polynomial time such that
$\blockreal{\mathbf{S}} \neq \emptyset$ if and only if $fhw(H^\cap) \leq k$.
As discussed before, we have $\fhw(H^\cap) = \fhw(H)$.
So, all that is left to do is to decide $\blockreal{\mathbf{S}} \neq \emptyset$.
Theorem~\ref{thm:frameworkalg} shows us that this is also 
feasible in polynomial time and therefore, so is the whole procedure.
\qed

\subsection{Computing FHDs Under Bounded Intersection}
\label{sect:fhw-bip}

We will now follow the same strategy used in the previous section on
bounded degree to prove tractability of checking $\fhw$ for hypergraph
classes enjoying the BIP. However, bounded intersection width
has different structural implications than bounded
degree. Example~\ref{ex:fhwLongEdge} illustrated our main challenge in the
fractional setting, namely the potentially unbounded size of the support
of an optimal fractional edge cover.
In that example, the growth of the support
is linked to the increase in the degree of $v_0$ whereas the intersection
width remains $1$, no matter how large $n$ becomes. A combinatorial
result that bounds the size of the support in terms of the optimal weight of the cover and the 
intersection width is therefore impossible.

Instead, we will show that every vertex set $B(\gamma)$ for some fractional cover $\gamma$
can be expressed by a
combination of edges and vertices, where the number of both is bounded
by functions of $\weight(\gamma)$ and the intersection width.
We can then again make use of Lemma~\ref{lem:fhwreal} to derive our main result.

\begin{theorem}\label{theo:exactFHDBIP}
  \label{thm:fhwbip}
  For every hypergraph class $\classC$ that enjoys the BIP, and
  for every constant $k \geq 1$, the \rec{FHD,\,$k$} problem is
  tractable, i.e., given a hypergraph $\HH \in \classC$, it is
  feasible in polynomial time to check $\fhw(\HH)\leq k$ and, if so,
  to compute an FHD of width $k$ of $\HH$.%
\end{theorem}

Throughout this subsection we consider a hypergraph $H$ with $\iwidth H \leq i$ for some constant $i$.
We will investigate fractional edge covers $\gamma$ of vertices $B(\gamma) \subseteq V(H)$.
We write $\Eh$ and $\El$ to denote the \emph{heavy} and \emph{light-weight} edges under $\gamma$, respectively. For given $k \geq 1$, 
we choose $1 - \frac{1}{2k}$ as the boundary between heavy and light-weight edges. More precisely, 
let $\cov(\gamma)$ denote the support of $\gamma$;  then we define  $\Eh$ and $\El$ as follows:

\smallskip

$\El = \{ e \in  \cov(\gamma) \mid \gamma(e) < 1 - \frac{1}{2k}\}$,

$\Eh = \{ e \in \cov(\gamma) \mid \gamma(e) \geq 1 - \frac{1}{2k}\}$,

\medskip

We do not require $\gamma$ to be optimal but we require it to be \emph{redundancy-free} in the following sense: 
if $\gamma'$ is a fractional cover  with  $\gamma'(e) < \gamma(e)$ for some $e \in \cov(\gamma)$ and 
$\gamma'(e) = \gamma(e)$ for all other edges $e \in E$, then $B(\gamma') \subset B(\gamma)$.
The set $\Bg$ has the following \emph{split and canonical representation} in terms of heavy and light-weight edges:

\begin{definition}
\label{def:canonicalRepresentation}
Let $\gamma$ be an edge-weight function of some hypergraph $H$ with $\weight(\gamma) \leq k$ for some $k \geq 1$. 
Then 
$e'_1 \cup \dots \cup e'_\ell \cup U$ with 
$\Bg = e'_1 \cup \dots \cup e'_\ell \cup U$ is 
a  {\em split representation\/} of $\Bg$ if 
the following property (1) holds:
\begin{itemize}
\item[(1)] for every $\alpha \in \{1, \dots, \ell\}$, $e'_\alpha = \Bg \cap e_\alpha$ for some $e_\alpha \in \Eh$;
\end{itemize}
If, additionally, the following property (2) holds, we call $\gamma$ the \emph{canonical representation} of $\Bg$:
\begin{itemize}
 \item[(2)] $U = \{ v \in \Bg \mid \forall e \in \Eh$: $v \not\in e\}$;
 \end{itemize}
\end{definition}

Recall that we are assuming $\gamma$ to be redundancy-free. Hence, 
also the union $e'_1 \cup \dots \cup e'_\ell$ is non-redundant in the
sense that, for every $e'_\alpha \in \{1, \dots, \ell\}$, we have
$(e'_1 \cup \dots \cup e'_\ell) \setminus e'_\alpha \subset (e'_1 \cup
\dots \cup e'_\ell)$. Moreover, for each $\alpha$, the edge
$e_\alpha \in \Eh$ with $e'_\alpha = \Bg \cap e_\alpha$ is in fact
unique. The reason for the uniqueness is that 
the weight of every
heavy edge is greater than $0.5$. Therefore, if $e'_i = e'_j$ for some 
indices $i \neq j$, then the weight put by $\gamma$ on the vertices in $e'_i$ (and, hence, also in $e'_j$) is 
greater than 1. We could thus safely reduce the weight of one of the edges $e_i$ or $e_j$ without decreasing 
$B(\gamma)$, which contradicts the irredundancy of $\gamma$.

In this section, we are considering hypergraphs satisfying the
BIP. Hence, 
we
can show that the number of vertices in $\Bg$ which are only covered
by light-weight edges, is bounded by a constant that exclusively depends on
$k$ and $i$.

\begin{lemma}
  \label{lem:boundedU}
Let $ k \geq 1$ and $i \geq 0$ be constants, let $H$ be a hypergraph with $\iwidth H \leq i$ 
and let $\gamma$ be an edge-weight function of $H$ with $\weight(\gamma) \leq k$. Moreover, 
let $e'_1 \cup \dots \cup e'_\ell \cup U$ be the 
canonical representation of $\Bg$. Then $|U| < 2ik^3$ holds.
\end{lemma}

\begin{proof}
By definition, $U$ is only covered by edges from $\El$. Let $e$ be an arbitrary edge in $\El$ and let $m = |\Bg \cap e|$. 
We first show that $m < 2ik^2$. Indeed, by definition of $\El$, $e$ puts weight $< 1 - \frac{1}{2k}$ on each vertex in $\Bg \cap e$. 
Hence, weight $> \frac{1}{2k}$ has to be put on each vertex in $\Bg \cap e$ by the other edges. In total, the other edges thus have to put
weight $>  \frac{m}{2k}$ on the $m$ vertices $\Bg \cap e$.

By the BIP, whenever $\gamma$ puts weight $w$ on some edge
$e'$ different from $e$, then, in total,  at most weight $i w$ is put on the vertices in $e$. Hence, since we are assuming $\weight (\gamma) \leq k$, 
the total weight of all edges in $\El$ (even the total weight of all edges in $E(H)$) is $\leq k$. Hence, in total
at most weight $ik$ can be put on the vertices in $\Bg \cap e$ by the edges different from $e$. We therefore have 
$ik > \frac{m}{2k} $ or, equivalently,  $m < 2ik^2$.

Now let $m$ be the maximum size of any edge in $\El$ and suppose that, for an arbitrary edge $e \in \El$, $\gamma(e) = w$ holds. Then, in total, $e$ puts weight $\leq m  w$ on the vertices in $U$. 
Hence, the total weight put by all edges of $\El$ on all vertices in $U$ is $\leq m k$. Moreover, recall that every vertex in $U$ receives weight at least $1$. Together with the above bound $2ik^2$ on the size of the edges in $\El$, we thus get  
$|U| < 2ik^3$.
\end{proof}

Our goal now is to show that every fractional cover $\gamma$ can be replaced by a fractional cover that is, in a sense, very close to an
integral cover. To formalize this closeness to an integral cover, 
we introduce the notion of {\em $c$-bounded fractional part\/}.
For $\gamma : E(H) \ra [0,1]$  and  $S \subseteq \cov(\gamma)$, we
write $\gamma |_S$ to denote the {\em restriction of $\gamma$ to $S$\/}, i.e., 
$\gamma |_S (e) = \gamma(e)$ if $e \in S$ and  $\gamma|_S (e)= 0$  otherwise.

\begin{definition}
\label{def:c-bounded}
Let 
$\calF= \left< T, (B_u)_{u\in T}, (\gamma_u)_{u\in T} \right>$ be 
an  FHD of some hypergraph $H$
and let $c \geq 0$. 
We say that $\calF$ has {\em $c$-bounded fractional part\/} if
in every node $u \in T$, the following property holds: 
Let $R = \{e \in \cov(\gamma_u) \mid \gamma_u(e) < 1\}$;  
then $|B(\gamma_u|_{R})| \leq c$.
\end{definition}

A naive approach towards our goal of reaching an FHD with $c$-bounded
fractional part will be to simply take the fractional edge cover
$\gamma_u$ at each node $u$ and set the weight of the heavy edges to
1. Of course, this will, in general, increase the width. However, as
will be illustrated below, it will not increase the width a
lot. Moreover, we will establish conditions under which the increase
of the width can be neglected.  We first give a formal definition of
the \emph{naive cover}:

\begin{definition}
  \label{def:naive-cover}
Let $\gamma$ be an edge-weight function of some hypergraph $H$ with
$\weight(\gamma) \leq k$ for some $k \geq 1$ and let
$e'_1 \cup \dots \cup e'_\ell \cup U$ be a split representation
of $\Bg$. Then we call the edge-weight function $\nu$ a
{\em naive cover\/} if the following properties hold:
\begin{itemize}
\item[(1)] for every $\alpha \in \{1, \dots, \ell\}$, let
  $e_\alpha \in \Eh$ with $e'_\alpha = \Bg \cap e_\alpha$; then we set
  $\nu(e_\alpha) = 1$.
\item[(2)] let $U' = U \setminus (e_1 \cup \cdots \cup e_\ell)$. $\nu$ is an optimal fractional edge cover of $U'$, i.e., let
  $S = \{e \mid e \in \cov(\nu)$ and $e \cap U' \neq \emptyset\}$; then
  $\weight(\nu|_S) = \rho^*(U')$.
\end{itemize}
\end{definition}

If we consider the naive cover of a canonical representation, we have $U'=U$ in the above definition.
Intuitively, a naive cover $\nu$ is close to an integral cover in that
it assigns weight 1 to some of the edges and the 
fractional part of $B(\nu)$ (i.e., the vertices in $B(\nu)$ which are
outside these edges with weight 1) are covered optimally. Ultimately, we will
show that we can always find a naive cover where the number vertices in the fractional part
is bounded by a constant.
Indeed, for the naive cover of a canonical representation, we have the bound $|U| \leq 2ik^3$ by
Lemma \ref{lem:boundedU}, provided that $\iwidth{H} \leq i$.
  However, the naive cover depends not only on $\gamma$
  but also on the split representation it is based on. The naive cover based on the canonical representation is not necessarily the one with the least total weight.

Recall that, since we are assuming fractional covers to be
redundancy-free, for each $\alpha$ in the definition above, the
edge $e_\alpha \in \Eh$ is unique. But of course, there may be several
optimal fractional edge covers of $U$.
By our definition of ``heavy edges''
we immediately get the
inequality $\weight(\nu) - \weight(\gamma) \leq 0.5$, since $\nu$
increases the weight of each heavy edge by at most $\frac{1}{2k}$ and
there cannot be more than $k$ heavy edges under $\gamma$. Below we
illustrate that this gap between $\nu$ and $\gamma$ might be even
smaller.

\begin{example}
\label{ex:big-edge} 
Recall from Example~\ref{ex:fhwLongEdge}
the hypergraph $H_n = (V,E)$ with 
$V = \{v_0, v_1, \dots, v_n\}$ and $E = \{e_0, \dots, e_n\}$, where
$e_0 = \{v_1, \dots, v_n\}$ and, for $\alpha \in \{1, \dots, n\}$, $e_\alpha = \{v_0, v_\alpha\}$. 
That is, $H$ contains a big edge $e_0 = \{v_1, \dots, v_n\}$, a single vertex $v_0$ outside this edge and small edges connecting 
each of the vertices in $e_0$ with the outside vertex $v_0$. 

Now let $V^* \subseteq V$ with $v_0 \in V^*$ and $|V^* \cap e_0| = d$ for some integer $d \geq 1$. For the sake of simplicity, 
suppose that
$V^* = \{v_1, \dots, v_d\}$. 
Then an optimal fractional edge cover 
$\gamma$ of $V^*$ would set $\gamma(e_\alpha) = \frac{1}{d}$ for each $\alpha \in \{1, \dots, d\}$ and $\gamma(e_0) = 1 -\frac{1}{d}$.
We thus get $\weight(\gamma) = 2 -\frac{1}{d}$. A naive cover of the canonical representation of $\gamma$ would set $\nu(e_0) = 1$ and  
$\nu(e_\alpha) = 1$
for a single (arbitrarily chosen) $\alpha \in\{1,\dots, n\}$
and $\nu(e_\beta) = 0$ for all other edges $e_\beta$.
\hfill$\Diamond$
\end{example}

In the above example, we observe that $\weight(\nu) - \weight(\gamma) \leq \frac{1}{d}$ holds. This means that, the bigger $d = |V^* \cap e_0|$ gets, the smaller the possible improvement over a naive cover will be. Of course, the above example is very simple in that $\Eh$ consists of a single 
edge $e_0$, $\Bg$ contains a single vertex $v_0$ outside $e_0$, and the light-weight edges containing  $v_0$ cover a single vertex in~$e_0$. The following lemma generalizes the observation 
that $\weight(\nu) - \weight(\gamma)$ decreases as the contribution of the heavy edges to $\Bg$ increases.

\begin{lemma}
\label{lem:max-gap}
Let $d, k \geq 1$ and $i \geq 0$ be constants and let $H$ be a
hypergraph with $\iwidth H \leq i$.  Moreover, let $\gamma$ be an
edge-weight function of $H$ with a canonical representation
$\Bg = e'_1 \cup \dots \cup e'_\ell \cup U$, s.t.\ $\ell \leq k$ and for
every $\alpha \in \{1, \dots, \ell\}$, the following properties hold:
\begin{itemize}
\item[(1)] $e'_\alpha = \Bg \cap e_\alpha$ for some $e_\alpha \in \Eh$;
\item [(2)] $|e'_\alpha| \geq d + ki$.
\end{itemize}
Then $\weight(\nu) - \weight(\gamma) < \frac{ik^2}{d}$ holds, where
$\nu$ denotes a naive cover corresponding to the canonical representation
$\Bg = e'_1 \cup \dots \cup e'_\ell \cup U$.
\end{lemma}

\begin{proof}
We first partition  $\cov(\gamma)$ and $\cov(\nu)$ 
into $R \cup S$ and $R' \cup S$, respectively,  with 
$S = \{e_\alpha \in E(H) \, | \, 1 \leq \alpha \leq \ell$ and $e'_\alpha = \Bg \cap e_\alpha\}$ and 
$R = \cov(\gamma) \setminus S$ and $R' = \cov(\nu) \setminus S$.

Clearly, $\weight(\gamma) = \weight(\gamma|_R) + \weight(\gamma|_S)$
and $\weight(\nu) = \weight(\nu|_{R'}) + \weight(\nu|_S) = \weight(\nu|_{R'}) + \ell$
hold.
Moreover, since 
a naive cover $\nu$ is an optimal fractional cover on $U$, we have 
$\weight(\nu|_{R'})  = \rho^*(U) \leq \weight(\gamma|_R)$. 
Hence, in order to prove the lemma, 
it suffices to show that $\ell - \weight(\gamma|_S) \leq \frac{ik^2}{d}$.

Consider $e'_\alpha$ for some $\alpha \in \{1, \dots, \ell\}$.
By $\iwidth H \leq i$, we have $|e'_\alpha \cap e'_\beta| \leq i$ for each of the $\ell-1$ $\beta$'s with  $\beta \neq \alpha$. 
Hence, less than $k i$ vertices in $e'_\alpha$ are also contained in one of the other heavy edges $e_\beta$.
Now let $e''_\alpha = \{ v \in e'_\alpha \mid v \not\in e'_\beta$ for every $\beta \neq\alpha\}$.
Then $|e''_\alpha| > d$ holds by the assumption 
$|e'_\alpha| > d + ki$.
Since $e''_\alpha \subseteq \Bg$, the cover $\gamma$ must put weight $\geq 1$ on each 
of the vertices in $e''_\alpha$. The edges 
$e_\beta$ with  $\beta \neq\alpha$ do not put any weight on the vertices in $e''_\alpha$. 
It remains to consider the edges in $R$: whenever an edge in $R$ has weight $w$ in $\gamma$, 
then it can put at most weight $wi$ in total on the vertices in $e''_\alpha$. Since $\weight(\gamma|_R) \leq \weight(\gamma) \leq k$, all of the edges in $R$ taken together can only put 
$\leq ki$ weight in total on the vertices in $e''_\alpha$. By 
$|e''_\alpha| > d$, there exists at least one vertex in $e''_\alpha$ that receives weight 
$< \frac{ki}{d}$ by $\gamma|_R$. Hence, since all vertices of $e''_\alpha$ are contained in $\Bg$, 
$\gamma(e_\alpha) > 1-\frac{ki}{d}$ must hold. We therefore get the inequality
$\weight(\gamma|_S) > \ell \cdot (1-\frac{ki}{d})$ and, thus, also 
$\ell - \weight(\gamma|_S) < \ell \cdot \frac{ki}{d} \leq k \cdot \frac{ki}{d} = \frac{k^2i}{d}$.
\end{proof}

As in Example~\ref{ex:big-edge}, we again observe that $\weight(\nu) - \weight(\gamma)$
decreases as the contribution of the heavy edges to $B(\gamma)$ increases. Nevertheless, no matter how big the heavy edges get, this gap may remain greater than 0. However, for our purposes, a slightly weaker condition than 
$\weight(\nu) = \weight(\gamma)$ suffices, namely: for a sufficiently big lower bound on the size of the heavy edges, if $\weight(\gamma) \leq k$, then also $\weight(\nu)  \leq k$ holds. 
Towards this goal, we
will show next that, for relevant values of $\weight(\nu)$, the difference 
$\weight(\nu) - k$ is bounded from below by some constant which only depends on 
$i$ and $k$. By relevant we mean that $\weight(\nu)$ is in the interval $(k,k+0.5]$. The reason for the irrelevance of the values outside this interval is that, if 
$\weight(\nu) \leq k$,  then we may simply replace $\gamma$ by $\nu$ without further ado. 
And if $\weight(\nu) > k + 0.5$, then,  by  $\weight(\nu) - \weight(\gamma) \leq 0.5$,
also $\weight(\gamma) > k$ would hold, which contradicts the assumption that 
$\weight(\gamma) \leq k$ holds. 

The following definitions are crucial: 
\begin{definition}
\label{def:HC}
For constants $c,k \geq 1$, we define:
\begin{eqnarray*}
\calH(c)  & = & \{ H = (V,E) \, \mid \, |V| \leq c\} \\
M(k,c) & = & \{ \rho^*(H) + j - k \, \mid \, H \in  \calH(c),\, j \in \mathbb{N}, \mbox{ and } 
\rho^*(H) + j - k \in (0, 0.5] \} \\
\mu(k,c) & = & \min (M(k,c)), \mbox{ if } M(k,c) \neq \emptyset \mbox{ and undefined otherwise}. 
\end{eqnarray*}
\end{definition}
Consider a fractional edge cover $\gamma$ with 
$\Bg = e'_1 \cup \dots \cup e'_\ell \cup U$, s.t.\  $|U| \leq c$. Moreover, we assume that 
$\weight(\gamma) \leq k$ holds. If $M(k,c) = \emptyset$, then $\weight(\nu) \leq k$ and we are done. 
Note that this case arises, for instance, in Example~\ref{ex:big-edge} for $k = 2$. There we have
$\ell = 1$ and $|U| = c = 1$. Hence, the fractional cover number of the induced subhypergraph
$H[U]$ is $1$ and the 
naive cover which sets $\nu(e_0) = 1$ and  $\nu(e_\alpha) = 1$ for a single $\alpha \in\{1,\dots, n\}$
clearly satisfies $\weight(\nu) \leq k$. 

Based on Lemmas \ref{lem:boundedU} and \ref{lem:max-gap}, we can now prove the central combinatorial result:

\begin{theorem}
  \label{cor:c-bounded-cover}
  Let $H$ be a hypergraph with $\iwidth{H} \leq i$ and let $\gamma : E(H) \to [0,1]$ with $\weight(\gamma) \leq k$.
  Then there exist a constant $c=f(k, i)$ for some function $f$ and an edge-weight function 
  $\nu : E(H) \to [0,1]$ with $\weight(\nu) \leq k$ and 
  $B(\gamma) \subseteq B(\nu)$ such that $\nu$ has $c$-bounded fractional part.
\end{theorem}
\begin{proof}
Let a canonical representation of $\gamma$ be of the form 
$\Bg = e'_1 \cup \dots \cup e'_\ell \cup U$ with 
$e'_\alpha = \Bg \cap e_\alpha$ such that $e_\alpha \in \Eh$ for every $\alpha \in \{1, \dots, \ell\}$.
By Lemma \ref{lem:boundedU}, we have $|U| < c_0 =  2ik^3$. 
By our definition of heavy edges, we know that $\weight(\nu) \leq \weight(\gamma) + 0.5$
holds for any naive cover $\nu$. Hence, together with the assumption 
$\weight(\gamma) \leq k$, we conclude that 
$\weight(\nu) \leq k+0.5$ must hold. 
We now distinguish the following cases:

\smallskip
\noindent
{\em Case 1.} Suppose that $U = \emptyset$. Let $\Bg = e'_1 \cup \dots \cup e'_\ell$. 
If $\gamma (e_\alpha) = 1$ for every $\alpha \in \{1, \dots, \ell\}$, then we are done. Otherwise,
let $\epsilon = \max (\{1-\gamma (e_\alpha) \, \mid \, 1 \leq \alpha \leq \ell\}$.
We  consider two subcases:

\smallskip
\noindent
{\em Case 1.1.} Suppose that $|e'_\alpha| > 2ki$ for every $\alpha \in \{1, \dots, \ell\}$. 
As in the proof of Lemma \ref{lem:max-gap}, 
we define $e''_\alpha = \{ v \in e'_\alpha \mid v \not\in e'_\beta$ for every $\beta \neq\alpha\}$.
By $\iwidth H \leq i$ and $\weight(\gamma) \leq k$,
we have $|e'_\alpha \setminus e''_\alpha| \leq  ki$.
Hence, by $|e'_\alpha| > 2ki$, 
we conclude that $|e''_\alpha| > ki$ for every $\alpha$.
By the definition of $\epsilon$, there exists $\alpha$ with 
$\epsilon = 1-\gamma (e_\alpha)$. Hence, by $e''_\alpha \subseteq\Bg$, the edges
outside $\Eh$ must put weight at least $\epsilon$ on each of the vertices in $e''_\alpha$. 
By $|e''_\alpha| > ki$, the edges outside $\Eh$ must put total weight $> ki\epsilon$ on all of the 
vertices in $e''_\alpha$. By $\iwidth H \leq i$, this requires that 
$\weight(\gamma|_R) > k\epsilon$ for $R = \cov(\gamma) \setminus \Eh$ must hold.
Hence, $\weight (\gamma) > \weight (\gamma|_{\Eh}) + k\epsilon$. 
On the other hand, the naive cover $\nu$ increases (compared with $\gamma$) 
the weight of each edge $e_1, \dots, e_\ell$
by at most $\epsilon$. That is, $\weight (\nu) \leq \weight (\gamma|_{\Eh}) + k\epsilon$.
Hence, $\weight(\nu) \leq \weight(\gamma)$ and we may replace 
$\gamma$ by $\nu$.

\smallskip
\noindent
{\em Case 1.2.} Suppose that there exists $\alpha \in \{1, \dots, \ell\}$ with $|e'_\alpha| \leq 2ki$.
W.l.o.g., assume that $\alpha = \ell$. Then we may represent $\Bg$ as
$\Bg = e'_1 \cup \dots \cup e'_{\ell-1} \cup U_1$  with
$U_1 = e'_\ell$. Clearly $|U_1| \leq 2ki \leq c_0$
and we move on to Case 2 or 3 below.

\smallskip
\noindent
{\em Case 2.} Suppose that $U \neq \emptyset$ and $M(k,c_0) = \emptyset$. 
The latter condition implies that 
$\rho^*(G) + j - k \not\in (0, 0.5]$ for every hypergraph $G \in \calH(c_0)$
and every natural number $j$. 
In particular, the induced subhypergraph $H[U]$ is in $\calH(c_0)$. Moreover, 
choose $j = \ell$. Then also 
$\rho^*(H[U]) + \ell - k \not\in (0, 0.5]$, i.e.,
$\weight(\nu) - k \not\in (0, 0.5]$. On the other hand, as argued above, 
$\weight(\nu) \leq k+0.5$ or, equivalently,  
$\weight(\nu) -k  \leq 0.5$
The only possibility to satisfy both conditions is that 
$\weight(\nu) \leq k$ holds. We may therefore again replace 
$\gamma$ by a naive cover $\nu$ to get $\Bg \subseteq B(\nu)$ and $\weight(\nu) \leq k$.

\smallskip
\noindent
{\em Case 3.} Suppose that $U \neq \emptyset$ and $M(k,c_0) \neq \emptyset$. Then we define the following 
value $d_0$ to distinguish between small and big heavy edges: 
$$d_0 =  \frac{ik^2}{\mu(k,c_0)}
$$
We distinguish two subcases: 

\smallskip
\noindent
{\em Case 3.1.} 
If $|e'_\alpha| > d_0 + ki$ for every $\alpha \in \{1, \dots, \ell\}$, 
then, by Lemma \ref{lem:max-gap},
we have
$$\weight(\nu) - \weight(\gamma) < 
\frac{ik^2}{d_0} =
\frac{ik^2}{ik^2/\mu(k,c_0)} = \mu(k,c_0).$$ 
Hence, together with the assumption $\weight(\gamma) \leq k$, we have 
$\weight(\nu) - k < \mu(k,c_0)$.
Moreover, 
as argued above, 
we may assume  $\weight(\nu) \leq k + 0.5$ or, equivalently, 
$\weight(\nu) - k \leq 0.5$.

Note that $\weight(\nu) = \ell + \rho^*(H[U])$ holds for the subhypergraph $H[U]$ of $H$ induced by $U$. 
By $|U| \leq c_0$, we thus have $H[U] \in \calH(c_0)$.
Hence, by $\weight(\nu) - k \leq 0.5$, 
either $\weight(\nu) - k \leq  0$ or $\weight(\nu) - k \in M(k,c_0)$ holds. 
The latter case can be ruled out because $\mu(k,c_0) = \min (M(k,c_0))$
and $\weight(\nu) - k < \mu(k,c_0).$ 
Hence, we conclude that $\weight(\nu) - k \leq 0$ or, equivalently, $\weight(\nu) \leq k$ holds. 
We may therefore again replace 
$\gamma$ by a naive cover $\nu$ to get $\Bg \subseteq B(\nu)$ and $\weight(\nu) \leq k$.

\smallskip
\noindent
{\em Case 3.2.} Now suppose that 
there exists $\alpha \in \{1, \dots, \ell\}$ with 
$|e'_\alpha| \leq d_0 + ki$. W.l.o.g., suppose that $\alpha = \ell$. 
Then we may in fact represent $\Bg$ as
$\Bg = e'_1 \cup \dots \cup e'_{\ell-1} \cup U_1$  with 
$U_1 = U \cup e''_\ell$ 
and $e''_\ell = \{ v \in e'_\ell \, \mid \, v \not\in e'_\beta$ for every $\beta \neq \ell\}$. 
Then, in particular, $|U_1| \leq c_1 = c_0 + d_0 + ki$ holds, i.e., 
the size of the fractional part $U_1$ is 
bounded by a constant that depends only on $k$ and $i$. 
Note that then still Case 3 applies since $U_1 \supseteq U\neq \emptyset$ and 
$M(k,c_1) \supseteq M(k,c_0) \neq \emptyset$ clearly hold. 
Hence, after at most $\ell$ iterations of Case 3, 
eventually Case 3.1 applies and we may replace $\gamma$ by
a naive cover, such that the size of the fractional part is bounded by 
a constant that depends only on $k$ and $i$. 
\end{proof}

  A cover that has $c$-bounded fractional part also has bounded support as the $c$ fractionally covered vertices require at most one edge each as support. We thus immediately obtain the following corollary which may be of independent interest.

\begin{corollary}
  Let $H$ be a hypergraph with $\iwidth{H} \leq i$ and let $\gamma : E(H) \to [0,1]$ with $\weight(\gamma) \leq k$.
  Then there exists an assignment
$\nu : E(H) \to [0,1]$ such that $\weight(\nu) \leq k$,
$B(\gamma) \subseteq B(\nu)$ and $\nu$ has bounded support (depending on $i$ and $k$).
\end{corollary}

The theorem gives us the desired representation of the bags in an FHD of a hypergraph of 
intersection width bounded by some constant $i$.
We are now ready to prove the main theorem of this subsection, namely the tractability of the 
\rec{FHD,\,k} problem for hypergraph classes with bounded intersection.

\medskip

\emph{Proof of Theorem~\ref{thm:fhwbip}}: We show that \rec{FHD,\,k}
is tractable for a class $\classC$ of hypergraphs with intersection
width at most $i$. Let $\HH \in \classC$ and let $\HH^1$ be the
hypergraph obtained by adding all edges of size 1 to $\HH$. Adding
these edges has no effect on any fractional edge covers and, in
particular, $\fhw(\HH) = \fhw(\HH^1)$.

Let $c$ be the constant from Theorem~\ref{cor:c-bounded-cover} for our
$k$ and $i$.  We define $\mathbf{R}$ as the set of all
$(k+c)$-sets w.r.t. $H^1$ with fractional cover number at most $k$. By Theorem~\ref{cor:c-bounded-cover}, for
every fractional cover $\gamma$ with $\weight(\gamma)\leq k$ we have that $B(\gamma)$ has $c$-bounded fractional part.
Thus, it is a subset of $k$ edges plus $c$ vertices, i.e., a $(k+c)$-set w.r.t. $H^1$.
Then, just like for bounded degree, we can combine Lemma~\ref{lem:fhwreal} and Theorem~\ref{thm:frameworkalg} to check $\fhw(\HH^1)\leq k$.
\qed

\medskip\noindent{\bf Deciding the {\sc Check} Problem for Hypergraphs of Bounded Rank.}
We conclude this section by discussing an important special case of bounded intersection width: bounded rank. It is easy to see that for every hypergraph $\HH$, we have $\iwidth{\HH}\leq \rarity{H}$. The tractability of the \rec{FHD,$k$} problem for classes of bounded rank is therefore already a consequence of Theorem~\ref{thm:fhwbip}. However, the complexity of the algorithm described above may be prohibitive for practical use. We therefore choose to briefly present a significantly simpler method for hypergraph classes with bounded rank.

\begin{definition}\label{def:brank}
We say that a hypergraph $H$ has
the {\em $r$-bounded rank property ($r$-BRP)} if 
$\rarity{\HH}\leq r$ holds. 
For a class $\classC$ of hypergraphs,
we say that $\classC$ 
has the  
\emph{bounded rank property (BRP)} 
if there exists a constant $r$
such that
every hypergraph $H$ in $\classC$
has the $r$-BRP.
\end{definition}

For every hypergraph with $\rarity{\HH}\leq r$ for some constant $r$, the following lemma is immediate:

\begin{lemma}
\label{lem:brp}
Let $H$ be a hypergraph whose rank is bounded by some constant $r$. Then for every (fractional) edge weight function $\gamma$ for $H$ satisfying $\weight(\gamma) \leq k$, the property 
$|B(\gamma)| \leq r \cdot k$ holds.  
\end{lemma}

We can now use the above lemma to directly generate the candidate bags
for the algorithm from Theorem~\ref{thm:frameworkalg}. For any
hypergraph $H$, we simply compute the set $\mathbf{S}$ of all sets
with at most $r\cdot k$ vertices that can be fractionally covered with
weight $k$. Clearly, if $\blockreal{\mathbf{S}} \neq \emptyset$, then there
exists an FHD of width $k$. On the other hand, if $\fhw(\HH)\leq k$, then there exists a \compnf FHD
$\defFHD$ of width $k$. This is easy to see from the proof of Lemma~\ref{lem:nfghd}. The construction of a \compnf GHD there translates directly to FHDs.
Hence, $\defTD$ is a \compnf TD where every bag
is in $\mathbf{S}$, i.e., $\blockreal{\mathbf{S}} \neq \emptyset$.  We thus get the following special case of Theorem~\ref{thm:fhwbip}:

\newcommand{\thmRankFHD}{%
For every hypergraph class $\classC$ 
that has bounded rank, and for every 
constant $k \geq 1$, 
the 
\rec{FHD,\,$k$} problem is tractable, i.e., 
given a hypergraph $\HH \in \classC$, it is feasible in polynomial time to 
check  $\fhw(\HH)\leq k$ and, if so, to compute an FHD of 
width $k$ of $\HH$.%
}

\begin{corollary}\label{theo:exactFHD}
\thmRankFHD
\end{corollary}

\section{Approximation of $\fhw$ Under Constraints}
\label{sect:approx}

We now turn our attention to approximations of the $\fhw$. 
It is known from~\cite{DBLP:journals/talg/Marx10}
that 
a tractable cubic approximation of the $\fhw$ always 
exists, i.e.: for $k \geq 1$,
there exists a polynomial-time algorithm that, given a hypergraph $H$ with 
$\fhw(H) \leq k$, finds an FHD of $H$ of  width $\calO(k^3)$. 
In this section, we search for conditions which guarantee a better 
approximation 
of the $\fhw$.

The natural first candidate for restricting hypergraphs is the BMIP.
For the \rec{GHD,\,$k$} problem, this restriction guarantees tractability.
Furthermore, the BIP and the BDP also lead to tractability
of the \rec{FHD,\,$k$} problem for fixed $k \geq 1$.
Both of these properties are generalized by the BMIP, making it the natural next step. For the fractional case, we will show that a significantly better polynomial-time approximation 
of the $\fhw$ than in the general case is possible for hypergraphs enjoying the BMIP. It is noteworthy, that we can reuse the strategy from the exact cases without modifications for this task. Finally, we show that classes with bounded Vapnik-Chervonenkis (VC) dimension --  an even more general property than the BMIP -- also allow for a polynomial-time approximation of the $\fhw$, but with slightly weaker guarantees.

\subsection{Approximation of $\fhw$ Under BMIP}
\label{sect:fhw-bip}

We first inspect the case of the bounded multi-intersection property. We will show that the BMIP allows for 
an arbitrarily close approximation of the $\fhw$ in polynomial time. 
Formally, the main result of this section is 
as follows:

\begin{theorem}\label{theo:fhwbmip}
  Let $\classC$ be a hypergraph class that enjoys the BMIP and let
  $k, \epsilon$ be arbitrary constants with $k \geq 1$ and
  $\epsilon > 0$.  Then there exists a polynomial-time algorithm that,
  given a hypergraph $H \in \classC$ with $\fhw(H) \leq k$, finds an
  FHD of $H$ of width $\leq k (1+ \epsilon)$.%
\end{theorem}

Even though we are now interested in approximations we can still
employ the same strategy that was used in the previous section.
However, it is not clear how to find a bounded representation of the sets
that can be covered with weight $k$. Yet, if we introduce an
arbitrarily small constant $\epsilon>0$, then it is possible to find supersets 
with weight $k (1+\epsilon)$ which are indeed representable in the desired fashion,
as we will show in Lemma \ref{lem:bmip} below. The proof of this lemma will make use of the 
following combinatorial result:

\begin{lemma}
  \label{lem:bmipcomb}
  Fix an integer $c\geq 1$.
  Let $X= \{x_1, \dots, x_n \}$ be a set of positive numbers $\leq \delta$
  and fix $w$ such that $\sum_{j=1}^{n}x_j \geq w \geq \delta c$.
  Then we have $\sum x_{i_1}\cdot x_{i_2} \cdot \,\cdots\, \cdot x_{i_c} \geq (w-\delta c)^c$,
  where the sum is over all $c$-tuples $(i_1, \dots, i_c)$ of \emph{distinct} integers
  from $[n]$.

\end{lemma}
\begin{proof}
  We proceed by induction over $c$. For $c=1$ the statement holds by the requirement that $\sum_{j=1}^{n}x_j \geq w$.
  Suppose the statement holds for $c-1$. Let us consider only the tuples with a fixed value for $i_c$ and let us consider only
  the sum over those tuples, which we can write as $x_{i_c}\cdot \sum x_{i_1} \cdot x_{i_2} \cdot  \cdots \cdot x_{i_{c-1}}$. Notably, the sum now ranges only over all $(c-1)$-tuples of distinct integers from $[n]\setminus \{i_c\}$. Therefore, we need to apply the induction hypothesis for the values $X \setminus x_{i_c}$. This also shifts the lower bound of the sum over all values from $w$ to $w - \delta$, since 
  we are assuming $x_j \leq \delta$ for every $j$. 
  Applying the induction hypothesis in this way then gives us the following inequality

  $$x_{i_c}\cdot \sum x_{i_1} \cdot x_{i_2} \cdot  \cdots \cdot x_{i_{c-1}} \geq x_{i_c} ((w-\delta)-\delta(c-1))^{c-1} = x_{i_c} (w-\delta c)^{c-1}$$

  Now, take the sum over all values in $[n]$ for $i_c$ to again
  consider all the $c$-tuples, and we arrive at the statement:
  $$\sum x_{i_1}\cdot  x_{i_2} \cdot \,\cdots\, \cdot x_{i_c} \geq (w-\delta c)^{c-1} \sum_{j=1}^n x_j  \geq w (w-\delta c)^{c-1} \geq (w-\delta c)^{c}$$
  
  \vspace{-20pt}
\end{proof}

\begin{lemma}
  \label{lem:bmip}
  Let $H$ be a hypergraph with $\cmiwidth{c}{H} \leq i$ for some $c \geq 1$ and $i \geq 0$ such that 
  $H$ is closed
  under intersection of edges and contains all edges of size 1. Moreover, fix a $k \geq 1$ and an $\epsilon$ with 
  $0 < \epsilon \leq 1$. 
  Then there exists an integer $q$, depending only on $k$, $c$, $i$, and $\epsilon$,
  such that for any set of vertices $S$ with $\rho^*(S) \leq k$, there exists a
  $q$-set $S'$ with $S \subseteq S'$ and $\rho^*(S') \leq k (1+\epsilon)$.
\end{lemma}
\begin{proof}
  Let $\gamma$ be a fractional edge cover of some $S \subseteq V(H)$
  with $weight(\gamma) \leq k$. Similar to the proof for \rec{FHD,\,k}
  under BIP, we will partition the edges with non-zero weight into heavy and light edges. This time the
  threshold for heavy edges is chosen to be $\epsilon/(4c)$, i.e., the sets $\El$ and $\Eh$ are defined as follows:

\smallskip

$\El = \{ e \in  \cov(\gamma) \mid \gamma(e) < \epsilon/(4c)\}$,

$\Eh = \{ e \in \cov(\gamma) \mid \gamma(e) \geq \epsilon/(4c)\}$.

\medskip

Let $X$ be the set of all vertices for which the weight of
$\gamma|_{\Eh}$ is at least $1-\epsilon/2$ and let $S' = S \cup X$.
We will show that $S'$ has the desired properties by showing that $X$ and $S\setminus X$ are 
$q_1$ and $q_2$ sets for appropriately chosen $q_1$ and $q_2$, both of which depend only on $k$, $c$, $i$, and $\epsilon$.

First, we can increase the weight of each edge in $\gamma$ by the factor 
$(1+\epsilon)$ to obtain a new $\gamma'$ with $S' \subseteq B(\gamma')$
and $weight(\gamma')\leq k(1+\epsilon)$. Now, in $\gamma'$, $X$ is fully covered by the 
edges in $\Eh$, i.e., $X \subseteq B(\gamma'|_{\Eh})$. This is due to the fact that
$(1+\epsilon) (1-\epsilon/2) = 1 + \epsilon/2 - \epsilon^2/2 \geq 1$ holds by the assumption 
$0 < \epsilon \leq 1$, i.e., if $\gamma|_{\Eh}$ put at least $1-\epsilon/2$ weight on $X$, then $\gamma'|_{\Eh}$ will put at least weight 1 on $X$.
Since we are assuming $\weight(\gamma) \leq k$ and $\gamma(e)  \geq \epsilon/(4c)$ for every $e \in \Eh$, 
there can be no more than $4ck/\epsilon$ edges in $\Eh$. From
Lemma~\ref{lem:boundsupp} we then know that $X$ is a
$\left(2^{4ck/\epsilon}\right)$-set. 

To complete the proof, we show that $S\setminus X$ contains at most
$f(c,k,i, \epsilon)$ vertices and can therefore be represented as a
union of that many edges of size 1. By the definition of $X$, the vertices in $S \setminus X$
need to receive at least $\epsilon/2$ weight from edges in $\El$. As every light edge contributes 
$\leq \epsilon/(4c)$ weight, there need to be at least $2c$ light edges incident to every vertex in $S \setminus X$.

We proceed with a counting argument.
Imagine a bipartite graph $G = (S \setminus X, T, E(G))$ where $T$ is the set of all $c$-tuples of distinct light edges.
There is an edge from $v \in S \setminus X$ to $(e_1, \dots, e_c) \in T$ iff $v$ is in $e_1 \cap \cdots \cap e_c$.
Furthermore, we assign weight $\prod_{j=1}^c \gamma(e_j)$ to every edge in $E(G)$ 
incident to a $(e_1, \dots, e_c) \in T$.

\sloppy
We now count the total weight in $G$ from both sides. First observe that on the $T$ side, we have degree at most $i$ because $\cmiwidth{c}{H}\leq i$. Therefore, the total weight is at most $i \cdot \sum_{(e_1, \dots, e_c) \in T} \prod_{j=1}^c \gamma(e_j)$.
Observe that $\sum_{(e_1, \dots, e_c) \in T} \prod_{j=1}^c \gamma(e_j) \leq \left(\sum_{e_1 \in \El} \gamma(e_1)\right) \cdot \cdots \cdot \left(\sum_{e_c \in \El} \gamma(e_c)\right)$ as, by distributivity, all the terms of the left-hand side sum are also present on the right-hand side of the inequality. Furthermore, we have $\sum_{e \in \El} \gamma(e) \leq k$ and thus, by putting it all together, we see that the total weight in $G$ is at most $k^ci$.

From the $S \setminus X$ side, consider an arbitrary vertex $v \in S \setminus X$ and let $e_1, \dots, e_n$ be the
light edges in $E(H)$ containing $v$.  As $v$ is not in $X$, the heavy edges of $\gamma$ contribute 
less than $1-\epsilon/2$ weight to $v$. Hence, we have $\sum_{j=1}^n\gamma(e_j) \geq \epsilon/2$ and $\gamma(e_j)< \epsilon/(4c)$ for each $j \in [n]$. We can
apply Lemma~\ref{lem:bmipcomb} for $X = \{\gamma(e_1), \dots, \gamma(e_n)\}$, $\delta = \epsilon/(4c)$, and $w = \epsilon/2$ to get the inequality
$\sum \gamma(e_{j_1}) \cdot \, \cdots \, \cdot \gamma(e_{j_c})  
\geq (\epsilon/2 - c\epsilon/(4c))^c = (\epsilon/4)^c$, where
the sum ranges over all $c$-tuples 
$(e_{j_1}, \dots, e_{j_c})$
of distinct edges from  $\{e_1, \dots, e_n\}$, which denote the 
light edges in $E(H)$ containing $v$. 

We conclude that $v$ (now considered as a vertex in $G$) is incident to edges whose total weight is 
$\geq (\epsilon/4)^c$ in $E(G)$.
Since we have seen above that the total weight of all edges in $E(G)$  is $\leq k^ci$,
there can be no more than $i(4k/\epsilon)^c$ vertices in $S \setminus X$.
\end{proof}

\emph{Proof of Theorem~\ref{theo:fhwbmip}}:
We proceed analogously to the proofs of the two tractability results for $\fhw$ in Section~\ref{sect:fhd}.
Recall that, for every $\HH \in \classC$, we are assuming  $\cmiwidth{c}{\HH}\leq i$ for constants $c$ and $i$.
Let $\HH$ be a hypergraph in $\classC$ and let $\HH'$ be the new hypergraph obtained by computing $H^\cap$ and adding every edge of size 1. As argued before, $\fhw(\HH') = \fhw(\HH)$ and, by Lemma~\ref{lem:bmipclosure}, we can compute $\HH'$ in polynomial time.

From the combination of Lemma~\ref{lem:bmip} with Lemma~\ref{lem:fhwreal} we have that there exists a polynomial time computable set $\mathbf{S}$ with $\blockreal{\mathbf{S}} \neq  \emptyset$ if and only if 
$\fhw(\HH') \leq k (1+\epsilon)$. We can then use Theorem~\ref{thm:frameworkalg} to decide $\blockreal{\mathbf{S}}\neq \emptyset$ and, in consequence, $\fhw(\HH') \leq k (1+\epsilon)$ and, therefore, also 
$\fhw(\HH)\leq k(1+\epsilon)$.
\qed

\medskip
\noindent
{\bf A polynomial time approximation scheme for finding optimal FHDs.}
Recall the definition of the \boundedopt\ problem from Section~\ref{sect:introduction}, i.e.: given a hypergraph $H$, we are interested in $\fhw(H)$,
but only if $\fhw(H) \leq K$. 
We will now show that, with the algorithm from Theorem~\ref{theo:fhwbmip} at our disposal, we are able to give a polynomial time approximation scheme for 
the bounded optimization problem. More precisely, we aim at an approximation algorithm with the following properties: 

\begin{definition}[PTAS \cite{DBLP:books/lib/Ausiello99,DBLP:books/daglib/0004338}]
 \label{def:PTAS}
 Let $\Pi$ be an (intractable) minimization problem with positive objective function $f_\Pi$. 
 An algorithm \texttt{Alg} is called an {\it approximation scheme} for $\Pi$ if on input $(I,\epsilon)$, where $I$ is an instance of $\Pi$ and $\epsilon > 0$ is an error parameter, it outputs a solution $s$ such that:
   \[ f_\Pi(I,s) \leq (1+\epsilon) \cdot f_\Pi(I,s^*) \]
 where $s^*$ is an optimal solution of $I$, i.e. for all other solutions $s'$ of $I$ it holds that $f_\Pi(I,s^*) \leq f_\Pi(I,s')$.
 
 \texttt{Alg} is called a \emph{polynomial time approximation scheme} (PTAS), if for every fixed $\epsilon > 0$, its running time is bounded by a polynomial in the size of instance $I$.
\end{definition}

We now show that, in case of the BMIP, the 
$K$-\textsc{Bounded-FHW-Optimization} problem indeed allows for a PTAS:

\begin{algorithm}[t]
\SetKwData{Left}{left}\SetKwData{This}{this}\SetKwData{Up}{up}
\SetKwFunction{Union}{Union}\SetKwFunction{FindCompress}{FindCompress}
\SetKw{not}{not}
\SetKwData{N}{N}

\SetKwInput{Output}{output}
\SetKwInput{Input}{input}

\Input{Hypergraph $H$ with $\cmiwidth{c}{H} \leq i$, numbers $K \geq 1$, $\epsilon \geq 0$}
\Output{Approximation of $\fhw(H)$, i.e.,  FHD $\mcF$ with $\width(\mcF) \leq \fhw(H) + \epsilon$ if $\fhw(H) \leq K$}
\BlankLine
\tcc{Check upper bound}
\If{\not ($\mcF =$ \findFHD ($H$, $K$, $\epsilon$, $c$, $i$))}{
   \Return{\tt fails} \tcc*[r]{$\fhw(H) > k$}
}
\BlankLine
\tcc{Initialization}
$L = 1$\;
$U = K + \epsilon$\;
$\epsilon'$ = $\epsilon / 3$\;
\BlankLine
\tcc{Main computation}
\Repeat{$U-L<\epsilon$}{
   \If{$\mcF' =$ \findFHD ($H$, $L + (U-L)/2$, $\epsilon'$, $i$)}{
      $U = L + (U-L)/2 +\epsilon'$\;
      $\mcF = \mcF'$\;
   }
   \Else{
      $L = L + (U-L)/2$\;
   }
}
\Return{$\mcF$}\;
\caption{{\tt FHW-Approximation}} \label{alg:ptaas}
\end{algorithm}%

\begin{theorem}\label{theo:PTAS}
For every hypergraph class $\classC$ that enjoys the BMIP, there exists a PTAS for the \boundedopt\ problem.
\end{theorem}
\begin{proof}
By Theorem \ref{theo:fhwbmip}, there exists a function \findFHD ($H$, $k$, $\epsilon$, $c$, $i$) with the following properties: 

\begin{itemize}
\item \findFHD\  takes as input a hypergraph $H$ with $\cmiwidth{c}{H} \leq i$ and numbers $k \geq 1$, $\epsilon \geq 0$; 
\item \findFHD returns an FHD $\mcF$ of $H$ of width $\leq k+\epsilon$ if $\fhw(H) \leq k$ 
  and \texttt{fails} otherwise (i.e., $\fhw(H) > k$ holds). Note that finding a FHD of width $\leq k + \epsilon$ is equivalent
  to finding an FHD of width $\leq k(1+\epsilon')$ for $\epsilon' = \epsilon/k$.
\item \findFHD\  runs in time polynomial in the size of $H$, where $k$, $\epsilon$, $c$ and $i$ are considered 
as~fixed.
\end{itemize}
Then we can construct Algorithm \ref{alg:ptaas} ``{\tt FHW-Approximation}'', which uses \findFHD\ as subprocedure. We claim that 
{\tt FHW-Approximation} is indeed a PTAS for the \boundedopt\ problem. First we argue that the algorithm is correct; we will 
then also show its polynomial-time upper bound. 

As for the correctness, note that the algorithm first checks
if $K$ is indeed an upper bound on $\fhw(H)$. This is done via a call of function \findFHD. 
If the function call {\tt fails}, then we know that 
$\fhw(H) > K$ holds. Otherwise, we get an FHD $\mcF$ of width $\leq K + \epsilon$. In the latter case, 
we conclude that $\fhw(H)$ is in the interval $[L,U]$ with 
$L = 1$ and $U = K + \epsilon$.

The loop invariant for the repeat loop is, that $\fhw(H)$ is in the interval $[L,U]$ and 
the width of the FHD $\mcF$ is $\leq U$. 
To see that this invariant is preserved by every iteration of the loop, consider the 
function call \findFHD ($H$, $L + (U-L)/2$, $\epsilon'$, $i$): 
if this call succeeds then the function returns an FHD $\mcF'$ of width $\leq L + (U-L)/2 +\epsilon'$. Hence,
we may indeed set $U = L + (U-L)/2 +\epsilon'$ and $\mcF = \mcF'$ without violating the loop invariant.
On the other hand, suppose that the call of \findFHD\ fails. This means that 
$\fhw(H) > L + (U-L)/2$ holds. Hence,
we may indeed update $L$ to $L + (U-L)/2$ and the loop invariant still holds.
The repeat loop terminates when $U - L < \epsilon$ holds. 
Hence, together with the loop invariant, we conclude that, on termination, 
$\mcF$ is an FHD of $H$ with $\width(\mcF) - \fhw(H) < \epsilon$ or, equivalently, 
$\width(\mcF) < \fhw(H) + \epsilon$.
Note that we may assume w.l.o.g., that $\fhw(H) \geq 1$. Hence, $\fhw(H) + \epsilon \leq \fhw(H) (1+\epsilon)$ holds
and, therefore, also $\width(\mcF) < \fhw(H) (1+ \epsilon)$. 

It remains to show that algorithm {\tt FHW-Approximation} runs in polynomial time w.r.t.\ the size of $H$. 
By Theorem \ref{theo:fhwbmip}, the function \findFHD\ works in polynomial time w.r.t.\ $H$. 
We only have to show that the number of iterations of the repeat-loop is bounded by a polynomial in $H$. 
Actually, we even show that it is bounded by a constant (depending on $K$ and $\epsilon$, but not on $H$): 
let $K' := K + \epsilon - 1$. Then the size of the interval $[L,U]$ initially is $K'$. 
In the first iteration of 
the repeat-loop, we either set 
$U = L + (U-L)/2 +\epsilon'$ or $L = L + (U-L)/2$ holds. In either case, 
at the end of this iteration, we have $U - L \leq K' / 2 + \epsilon'$. 
By an easy induction argument, it can be verified that after $m$ iterations (with $m \geq 1$), 
we have 
$$U - L \leq \frac{K'}{2^m} + \epsilon' + \frac{\epsilon'}{2} + \frac{\epsilon'}{2^2} + \dots + \frac{\epsilon'}{2^{m-1}}
$$ 
Now set $m = \lceil \log (K' / \epsilon') \rceil$. Then we get
$$
\frac{K'}{2^m}  \leq \frac{K'}{2^{\log (K' / \epsilon')}} =
\frac{K'}{(K' / \epsilon')} = \epsilon'.$$
Moreover, $\epsilon' + \frac{\epsilon'}{2} + \frac{\epsilon'}{2^2} + \dots + \frac{\epsilon'}{2^{m-1}} < 2 \epsilon'$
for every $m \geq 1$.
In total, we thus have that, after  $m = \lceil \log (K' / \epsilon')\rceil$ iterations of the repeat loop, 
$U - L < 3 \epsilon' = \epsilon$ holds, i.e., the loop terminates. 
\end{proof}

\subsection{Approximation of $\fhw$ Under Bounded VC Dimension}
\label{sect:fhw-bmip}

We now present a polynomial-time approximation of the $\fhw$ for
classes of hypergraphs enjoying bounded Vapnik-Chervonenkis (VC)
dimension. This bounded VC dimension is yet again more general property than BMIP but the approximation guarantees in this class will be slightly weaker.
We will combine some classical results 
on the VC dimension with some novel observations.
This will yield an approximation of the $\fhw$ 
up to a logarithmic factor. We first recall the definition of the 
VC-dimension of hypergraphs.

\begin{definition}[\cite{1972sauer,1971vc}]
\label{def:vc}
Let $\HH=(V(H),E(H))$ be a hypergraph and $X\subseteq V(H)$ a set of vertices. 
Denote by $E(H)|_X$ the set 
$E(H)|_X =\{X \cap e\, |\, e\in E(H)\}$. 
The vertex set $X$ is called {\em shattered} if 
$E(H)|_X=2^X$.
The {\em Vapnik-Chervonenkis dimension (VC dimension) $\vc(\HH)$} of $\HH$ is 
the maximum cardinality of a shattered subset of $V(H)$. 
\end{definition}

We now provide a link between the VC-dimension and 
our approximation of the $\fhw$.

\begin{definition}
\label{def:transversality}
Let $H = (V(H),E(H))$ be a hypergraph. A {\em transversal}  (also known as {\em 
hitting set\/})
of $H$ is a subset $S \subseteq V(H)$ that has a non-empty intersection with 
every edge of $H$. 
The {\em transversality} $\tau(H)$  of $H$  is the 
minimum cardinality of all transversals of $H$.

Clearly, $\tau(H)$ corresponds to the minimum of the following integer linear 
program: 
find a mapping $w: V\rightarrow \{0,1\}$ 
which minimizes $\Sigma_{v\in V(H)}w(v)$ under the condition that
$\Sigma_{v\in e}w(v)\geq 1$ holds for each hyperedge $e\in E$.

The {\em fractional transversality}  $\tau^*$ of $H$ is defined as the minimum 
of
the above linear program when dropping the integrality condition,
thus allowing mappings $w: V\rightarrow \mathbb{R}_{\geq 0}$.
Finally, the {\em transversal integrality gap} $\tigap{\HH}$ of $\HH$ is the 
ratio $\tau(\HH)/\tau^*(\HH)$.
\end{definition}

Recall that computing the mapping $\lambda_u$ for 
some 
node $u$ in a GHD can be seen as searching for a minimal edge cover $\rho$ of 
the vertex set $B_u$, whereas computing 
$\gamma_u$ in an FHD 
corresponds to the search for a minimal fractional edge cover $\rho^*$
\cite{2014grohemarx}. Again, 
these problems can be cast as linear programs where the first problem has the 
integrality condition and the second one has not. 
Further, we can define the {\em cover integrality gap} $\cigap{H}$ of $H$ as 
the 
ratio $\rho(\HH)/\rho^*(\HH)$.
With this, we state the following approximation result for 
$\fhw$.

\newcommand{\thmApproxVC}{%
Let $\classC$ be a class of hypergraphs with VC-dimension bounded by some 
constant $d$
and let $k \geq 1$.
Then there exists a polynomial-time algorithm that, given a hypergraph $H \in 
\classC$ with 
$\fhw(H) \leq k$, finds an FHD of $H$ of width $\calO(k \cdot \log k)$.%
}

\begin{theorem}\label{theo:ApproxVC}
\thmApproxVC
\end{theorem}

\newcommand{\ggnew}[1]{#1}
\newcommand\hcancel[2][black]{\setbox0=\hbox{$#2$}%
\rlap{\raisebox{.45\ht0}{\textcolor{#1}{\rule{\wd0}{2pt}}}}#2} 
\newcommand{\ggkill}[1]{\hcancel[violet]{#1}}

\begin{proof}
The proof proceeds in several steps.

\smallskip

\noindent
{\em Reduced hypergraphs.} 
Recall the notion of reduced hypergraphs from Section~\ref{sect:fhw-bd}. There, it was established that we can assume, w.l.o.g., that every hypergraph is reduced. For the rest of this proof we therefore consider only reduced hypergraphs. This ensures that 
$(H^d)^d = H$  holds.
It is well-known and easy to verify that the following relationships
between $H$ and $H^d$ hold for any reduced hypergraph $H$,  
(see, e.g., \cite{duchet1996hypergraphs}):

\smallskip

(1) The edge coverings of $\HH$ and the transversals of $\HH^d$
coincide.

(2) The fractional %
edge coverings of $\HH$ and the fractional transversals of 
$\HH^d$
coincide.

(3) $\rho(\HH)=\tau(\HH^d)$, $\rho^*(\HH)=\tau^*(\HH^d)$, and
$\cigap{\HH}=\tigap{\HH^d}$.

\smallskip

\noindent
{\em VC-dimension.} By a classical result (\cite{ding1994} Theorem (5.4), see also \cite{Bronnimann1995} for related results), for every 
hypergraph $H = (V(H),E(H))$  \ggnew{with at least two edges} we have:
$$\tigap{H}= \tau(H)/\tau^*(H) \leq 2\vc(H)\log(11\tau^*(H)).$$
\ggnew{For hypergraphs $H$ with a single edge only, $\vc(H)=0$, and thus the above inequation does not hold. However, for such hypergraphs 
$\tau(H)=\tau^*(H)=1$. By putting this together, we get:}  
\ggnew{$$\tigap{H}= \tau(H)/\tau^*(H) \leq \max(1,2\vc(H)\log(11\tau^*(H))).$$}
Moreover, in \cite{assouad1983}, it is shown that $\vc(\HH^d)<2^{\vc(\HH)+1}$ 
always
holds.
In total, we thus get 

\begin{align*}
\cigap{\HH}=\tigap{\HH^d} & \leq \ggnew{\max(1,}2\vc(\HH^d)\log(11\tau^*(\HH^d))\ggnew{)} \\
& \leq \ggnew{\max(1,}2^{\vc(\HH)+2}\log(11\rho^*(\HH))\ggnew{)}\\
& \ggnew{\leq \max(1,2^{d+2}\log(11\rho^*(\HH)))},\ 
\ggnew{\mbox{which is\ } O(\log\rho^*(H))}.
\end{align*}

\noindent
{\em Approximation of $\fhw$ by $\ghw$.} 
Suppose that $H$ has an FHD $\left< T, (B_u)_{u\in V(T)}, (\lambda)_{u\in V(T)} 
\right>$
of width $k$. Then there exists a GHD of $H$ of width %
$\calO(k \cdot \log k)$. Indeed, we can find such a GHD by leaving the 
tree 
structure $T$ and the bags $B_u$ for every node $u$ in $T$ unchanged and 
replacing each fractional edge cover $\gamma_u$ of $B_u$ by an optimal integral 
edge cover $\lambda_u$ of $B_u$. By the above inequality, we thus increase the 
weight at each node $u$ only by a factor $\calO(\log k)$. Moreover, we know 
from 
\cite{DBLP:journals/ejc/AdlerGG07} that $\hw(H) \leq 3 \cdot \ghw(H) + 1$ holds.
In other words, we can compute an HD of $H$ (which is a special case of an FHD) 
in polynomial time,  whose width is $\calO(k \cdot\log k)$.
\end{proof}

To conclude, the following
Lemma~\ref{lemma:BMIPvsVC} establishes a relationship between BMIP and VC-dimension.

\newcommand{\lemBMIPvsVC}{%
If a class $\classC$ of hypergraphs has the BMIP then it has bounded 
VC-dimension.
However, there exist classes $\classC$ of hypergraphs with bounded VC-dimension 
that do not have 
the BMIP.}

\begin{lemma}\label{lemma:BMIPvsVC}
\lemBMIPvsVC
\end{lemma}

\begin{proof} \mbox{}
[BMIP $\Rightarrow$ bounded VC-dimension.]
Let $c \geq 1, i \geq 0$ and let 
$H$ be a hypergraph with $\cmiwidth{c}{\HH}\leq i$.
We claim that then $\vc(H) \leq c +  i$ holds.

Assume to the contrary that there exists a set  $X \subseteq V$, such 
that $X$ is  shattered and $|X| > c + i$. 
We pick $c$ arbitrary, pairwise distinct vertices $v_1, \dots, v_{c}$ from $X$ 
and 
define $X_j = X \setminus \{v_j\}$ for each $j$. 
Then $X = (X_1 \cap \dots \cap X_c) \cup \{v_1, \dots, v_c\}$ holds
and also $|X| \leq   |X^*| + c$ with 
$X^* \subseteq X_1 \cap \dots \cap X_c$.

Since $X$ is shattered, for each $1 \leq j \leq c$, there exists a distinct
edge $e_j \in E(H)$ with $X_j = X \cap e_j$. 
Hence, $X_j = X \setminus \{v_j\} \subseteq e_j$ and also 
$X^* \subseteq e_1 \cap e_2 \cap  \dots \cap e_c$ holds, i.e., 
$X^*$ is in the intersection of $c$ edges of $H$. 
By $\cmiwidth{c}{\HH}\leq i$, we thus get $|X^*| \leq i$.
In total, we have $|X| \leq |X^*| + c \leq i + c$, which 
contradicts our assumption that $|X| > c+i$ holds.

\smallskip
\noindent
[bounded VC-dimension $\not\Rightarrow$ BMIP.]
It suffices to exhibit a family $(H_n)_{n \in \mathbb{N}}$ 
of hypergraphs such that 
$\vc(H_n)$ is bounded whereas 
$\cmiwidth{c}{\mbox{$H_n$}}$ is unbounded for any constant $c$.
We define $H_n = (V_n,E_n)$ as follows:

\smallskip
$V_n = \{v_1, \dots, v_n\}$ 

$E_n = \{ V_n \setminus \{v_i\} \mid 1 \leq i \leq n\}$

\smallskip
\noindent
Clearly, $\vc(H_n) < 2$. Indeed, take an arbitrary set $X \subseteq V$ with 
$|X| \geq 2$. Then $\emptyset \subseteq X$ but
$\emptyset \neq X \cap e$ for any $e \in E_n$. 
On the other hand, let $c \geq 1$ be an arbitrary constant and let $X = e_{i_1} 
\cap \dots \cap e_{i_\ell}$ for some $\ell \leq c$ and edges $e_{i_j} \in E_n$. 
Obviously, $|X| \geq n - c$ holds. Hence, also $\cmiwidth{c}{\mbox{$H_n$}} \geq 
n - c$, i.e., 
it is not bounded by any constant $i \geq 0$.
\end{proof}

\noindent
In the first part of Lemma \ref{lemma:BMIPvsVC}, 
we have shown that $\vc(H) \leq c +  i$ holds.
For an approximation of an FHD by a GHD, we need
to approximate the fractional edge cover $\gamma_u$ of each bag $B_u$ 
by an integral edge cover $\lambda_u$, i.e.,
we  consider fractional vs.\ integral edge covers of the induced hypergraphs
$H_u = (B_u, E_u)$ with $E_u = \{e  \cap B_u \mid e \in E(H)\}$.
Obviously, the bound $\vc(H) \leq c +  i$ carries over to 
$\vc(H_u) \leq c +  i$.

\section{Conclusion and Future Work}
\label{sect:conclusion}
In this work,  we have settled the complexity of deciding $\fhw(H) \leq k$ for fixed 
constant $k\geq 2$ and $\ghw(H) \leq k$ for 
$k = 2$ by proving the $\NP$-completeness of both problems. 
This gives negative answers to two open problems.
On the positive side, we have identified rather 
mild restrictions such as the BDP (i.e., the bounded degree property), 
BIP (i.e., the bounded intersection property), LogBIP, 
BMIP (i.e., the bounded multi-intersection property), and LogBMIP, 
which give rise  to  a \ptime algorithm 
for the 
\rec{GHD,\,$k$} problem. Moreover, we have shown that the BDP and the BIP
ensure tractability also of the 
\rec{FHD,\,$k$} problem. For the BMIP, we have shown that an arbitrarily close approximation 
of the $\fhw$ in polynomial time exists. In case of bounded VC dimension, we have proposed a polynomial-time
algorithm for approximating the $\fhw$ up to a logarithmic factor.
As the empirical analyses reported 
in 
\cite{pods/FischlGLP19} show, these restrictions are very
well-suited for %
instances of CSPs and, even more so, of CQs. 
We believe that they deserve 
further attention. 

Our work does not finish here. We plan to explore several further issues 
regarding the computation and approximation of the fractional hypertree width. 
We find the following questions particularly appealing: (i)  Does the special 
condition defined by  Grohe and Marx~\cite{2014grohemarx} lead to tractable 
recognizability also for FHDs,  i.e., in case we 
define ``{\it sc-fhw$(H)$}'' as the smallest width an FHD of $H$ satisfying the 
special condition,  can $\mbox{\it sc-fhw}(H) \leq k$ be recognized efficiently?\ 
(ii)  Our tractability result in Section 5 for the \rec{FHD,\,$k$} problem is weaker than for \rec{GHD,\,$k$}. 
In particular, for the BMIP, we have only obtained efficient approximations of the $\fhw$. 
It is open if the BMIP suffices to ensure tractability of \rec{FHD,\,$k$}.
Note that the techniques applied in the tractability proof of \rec{FHD,\,$k$} in case of the BIP 
do not carry over in an obvious way to the BMIP. 
  An important property in our tractability proof is that 
edges with low weight can only contribute a {\em constant\/} number of vertices to $B(\gamma)$. This property crucially depends on the BIP, which implies that the additional weight put by other edges on these vertices is bounded by $ki$. 
There is no immediate analogue of this bound in case of the BMIP.
At any rate, we are currently working on an extension
of the tractability of \rec{FHD,\,$k$}  from BIP to BMIP and we conjecture that tractability indeed holds.

\begin{acks}
We are truly grateful to one of the anonymous referees, whose insightful and incredibly deep review comments have greatly helped to 
simplify the presentation and increase readability.

This work was supported by the Engineering and Physical
Sciences Research Council (EPSRC), Programme Grant EP{\slash}M025268{\slash} VADA:
Value Added Data Systems --- Principles and Architecture as well as by the Austrian
Science Fund (FWF):P30930 and Y698.
\end{acks}

\clearpage

\bibliographystyle{ACM-Reference-Format}
\bibliography{main} 

\appendix
\section{A Polynomial Time Algorithm for the \compnf CTD Problem}
\label{sec:algappendix}

In ths section we provide a proof of Theorem~\ref{thm:frameworkalg}. To do so, we present a bottom-up construction of \compnf CTDs, if they exist, using dynamic programming. Note that our presentation does not optimize for runtime, our goal is only to establish that the problem can be decided in polynomial time.

\begin{definition}
  A pair $(B,C)$ of \emph{disjoint} subsets of $V(H)$ is a
  \emph{block} if $C$ is a $[B]$-component of $H$ or $C =
  \emptyset$. Such a block is \emph{headed} by $B$.  Let $(B,C)$
  and $(X,Y)$ be two blocks.  We say that $(X,Y) \leq (B,C)$ if
  $X \cup Y \subseteq B \cup C$ and $Y \subseteq C$.

\end{definition}

\begin{definition}
For a block $(B, C)$ and vertex set $X \subseteq V(H)$ with $X \neq B$, we say that 
{\em $X$ is a basis of $(B, C)$} if 
the following conditions hold:
\begin{enumerate}
\item \label{cond:basis1} Let $(X,Y_1), \dots, (X,Y_\ell)$ be all
the blocks headed by $X$ that are less than or equal to 
$(B,C)$. Then $C \subseteq X \cup \bigcup_{i=1}^\ell Y_i$.
\item \label{cond:basis2} For each $e \in E(H)$ such that $e \cap C \neq \emptyset$,
$e \subseteq X \cup \bigcup_{i=1}^\ell Y_i$.
\item  \label{cond:basis3} For each $i \in [\ell]$, there exists a \compnf TD of $H[X \cup Y_i]$ where the root has precisely $X$ as its bag.
\end{enumerate}
\end{definition}

The existence of a basis $X$ intuitively corresponds to the existence of a
tree decomposition that covers the whole component $C$ (by $X$ together with the
[$X$]-components $Y_1, \dots, Y_\ell$)
and connects $X$ to
its parent bag $B$. The following lemmas confirm that this definition of
a basis for a block corresponds to such a TD in the expected way.

\begin{lemma}
  \label{lem:support}
  Let $H$ be a hypergraph, and $\defTD$ be a \compnf TD of $H$.  Let
  $r \in T$ be a non-leaf node. For each child $s$ of $r$, let $C_s$
  be the $[B_r]$-component associated with $s$.  The following two statements are true:
  \begin{itemize}
  \item $B_s$ is a basis
    of the block $(B_r, C_s)$.
    \item $(B_s, D) \leq (B_r, C_s)$ if and only
  if $D$ is either a component associated with a child of $s$ or if $D$ is empty.
  \end{itemize}
\end{lemma}
\begin{proof}
  We first observe that $(B_s, D) \leq (B_r, C_s)$ if and only
  if $D$ is either empty or a component associated with a child of $s$.
  Indeed, since we assume \compnf, we have that
  $V(T_s) = C_s \cup (B_r \cap B_s)$ and $C_s$ is the only such
  $[B_r]$-component.  Thus, $V(T_s) \subseteq B_r \cup C_s$.
  
  Moreover, $B_s \cup D \subseteq V(T_s)$ holds for every block 
  $(B_s, D)$ where $D$ is empty or a component associated with a
  child of $s$, which completes the proof of the ``if'' direction. 
  For the ``only if'' direction, recall that $(B_s, D) \leq (B_r, C_s)$
  requires $D \subseteq C_s$. Since there is only one node associated with $C_s$, we conclude 
  $D \subseteq \VTs$. Since
  $D$ is a \comp{$B_s$}, it must have its own associated child of $s$.

  Now that we know exactly which blocks headed by $B_s$ are relevant,
  we can show that they satisfy the conditions of a basis.
  Let $child(s)$ be the set of all children of $s$. For every $u \in child(s)$, let $D_u$ be the $[B_s]$-component
  associated with $u$.
  The
  vertices that occur in the subtree $T_s$ are precisely
  $V(T_s) = B_s \cup \bigcup_{u \in child(s)} D_u$. Since we assume
  \compnf, we also have $C_s \subseteq V(T_s)$ and therefore Condition~\ref{cond:basis1} of a
  basis is satisfied.
  
  For Condition~\ref{cond:basis2} it is enough to observe that if
  $e \cap C_s \neq \emptyset$, then it must be covered in the subtree
  $T_s$, i.e., $e \subseteq V(T_s)$. Otherwise, suppose $e$ were only
  covered in some node $u$ not in the subtree $T_s$. There is a vertex
  $v \in e \cap C_s$ that occurs in $V(T_s)$ but not in $B_r$ (recall
  $B_r$ and $C_s$ are disjoint). Any path from a node of $T_s$ to $u$
  must pass through $B_r$, which would break connectedness for 
  $v$.

  Finally, for Condition~\ref{cond:basis3} and each $u \in child(s)$,
  consider the subtree $T_u^*$ induced by $\{s\} \cup T_u$.  The root
  of $T_u^*$ has bag $B_s$ and from $V(T_u) = C_u \cup (B_s \cap B_u)$ also $B_s \cup V(T_u) = B_s \cup C_u$, i.e., $T_u^*$ is indeed a TD of $H[B_s \cup C_u]$.  To see that $T_u^*$ is
  in fact in \compnf, observe that $H[B_s \cup C_u]$ has only a single
  $[B_s]$-component $C_u$. Then \compnf of $T_u^*$ follows from the assumption that the original decomposition $\defTD$ is in
  \compnf.
\end{proof}

For the following arguments, it is convenient to introduce the notion of a
\emph{union of TDs} that have the same root bag.  Let
$\left<T_1, (B_{1,u})_{u\in T_1} \right>,\dots, \left<T_n, (B_{n,u})_{u \in T_n} \right>$ be rooted TDs and w.l.o.g. assume they have pairwise distinct nodes.
For each $i \in [n]$, let us denote the root of $T_i$ by $r_i$.
Furthermore, assume that $B_{r_i} = B_{r_j}$ for all $i,j \in [n]$.
We then define the union
$\defTD = \bigcup_{i =1}^n \left<T_i, (B_{i,u})_{u\in T_i} \right>$ as
the following structure:  
$T$ is a tree with a new root node $r$ and $B_r = B_{r_1}$. For each $i \in [n]$, all the nodes $u$ of $T_i$
except for the root $r_i$ are in $T$ and for each $u \neq r_i$ in $T_i $, we have $B_u = B_{i,u}$. 
Moreover, all edges of $T_i$ except for the ones adjacent to the root $r_i$ are also contained in $T$. 
Further, for every edge $[u,r_i]$ we introduce an edge $[u,r]$ in $T$.
The following lemma establishes a sufficienct condition such that this new structure
is indeed a TD.

\begin{lemma}
  \label{lem:tdunion}
  Let $(B, C_1), \dots, (B, C_\ell)$ be blocks of a hypergraph $H$. Assume for each $i \in [\ell]$ 
  that there
  exists a \compnf TD $\mcT_i$ of $H[B \cup C_i]$ where $B$ is the bag of the root. Then,
  $\mcT = \bigcup_{i=1}^\ell \mcT_i$ is a \compnf TD of $H[B\cup \bigcup_{i=1}^\ell C_i]$.
\end{lemma}
\begin{proof}
  We start by verifying that the connectedness condition is satisfied in $\mcT$.  By
  assumption, connectedness holds for each subtree rooted at a child
  of the root. The condition can then only be violated if vertices
  occur in more than one such subtree but not in the bag of the root. As all
  $C_i$ for $i \in [q]$ are $[B]$-components, or empty, they are also pairwise
  disjoint. So, the subtrees can only share variables in $B$ which is
  precisely the bag of the root.

  To see that every edge is covered, observe that for each edge $e$ in
  $H[B\cup \bigcup_{i \in [q]} C_i]$, there must be at least one
  $i \in [q]$ such that $e$ is also an edge in $H[B \cup C_i]$.
  Otherwise, $e$ would have to be part of more than one
  $[B]$-components, which is impossible as the components would then
  be $[B]$-connected.

  It remains to show that $\mcT$ is in \compnf: the decompositions
  $\mcT_i = \left< T_i, (B_{i,u})_{u \in T_i} \right>$ were already in
  \compnf and are left unchanged.  The only new parent/child
  relationships are those from the root $r$ to its children. Let $u$ with bag $B_u$ 
  be a child of $r$.  By the construction of the union of TDs, $u$ is obtained from some node 
  $u_i$ with bag $B_{i,u_i}$ in TD $\mcT_i$ for some $i$. That is, we have $B_u = B_{i,u_i}$. 
  Hence, the subtree rooted at $u$ covers a single \comp{$B_{r}$} by 
  the fact that $B_r = B_{r_i}$  and   $\mcT_i$ is a \compnf TD of $H[B \cup C_i]$.  
\end{proof}

\begin{lemma}
  \label{lem:suptotd}
  Let $H$ be a hypergraph and $B \subseteq V(H)$. Let $(B, C)$ be a block of $H$. If there exists $X \subseteq V(H)$ that is a basis of  $(B, C)$
  or if $C = \emptyset$, then there exists a
  \compnf TD of $H[B \cup C]$ where $B$ is the bag of the root.
\end{lemma}
\begin{proof}
  First, if $C = \emptyset$, the TD with a single node $u$ and $B_u = B$ is trivially a \compnf TD of $H[B]$.
  Otherwise, let $(X, Y_1),\dots,(X, Y_\ell)$ be all the blocks headed by $X$ that
  are less or equal $(B,C)$. For each $i \in [\ell]$, let
  $\mcT_i$ be a \compnf TD of
  $H[X \cup Y_i]$ where $X$ is the bag of the root node. 

  Let $\mcT=\defTD$ be the union $\bigcup_{i=1}^\ell \mcT_i$. By Lemma~\ref{lem:tdunion}, $\mcT$ is
  a \compnf TD of $H[X \cup \bigcup_{i=1}^\ell Y_i]$. Add a new root $r$ with $B_r = B$ to $\mcT$ as the parent
  of the previous root to obtain $\mcT'$. We claim that $\mcT'$ is the desired \compnf TD of $H[B \cup C]$. Note that we have
  $V(\mcT) = X \cup \bigcup_{i=1}^\ell Y_i$.

  Assume an edge $e \in H[B \cup C]$. If $e \cap C \neq \emptyset$,
  then $e \subseteq X \cup \bigcup_{i=1}^\ell Y_i$ because $X$ is a basis of $(B,C)$.  Therefore, $e$  occurs in
  $H[X \cup \bigcup_{i=1}^\ell Y_i]$ and must be covered in $\mcT$.
  Otherwise, if $e \cap C = \emptyset$, then
  $e \subseteq B$ and $e$ is covered by the root node of $\mcT'$.

  $\mcT$ satisfies the connectedness condition and has $X$
  as the bag of its root node. Hence, the only way the connectedness condition can fail in $\mcT'$
  is if there is a vertex in $B$ and $V(\mcT)$ but not in $X$.
  For each $i \in [\ell]$ we have $Y_i \subseteq C$ and
  because $B$ and $C$ are disjoint we have $B \cap Y_i = \emptyset$.
  Therefore,
  $B \cap \left(X \cup \bigcup_{i=1}^\ell Y_i\right) = B \cap X$,
  i.e., any vertex in $V(\mcT)$ and in $B$ is also in the bag $X$ at the 
  root of $\mcT$.

  It remains to show that $\mcT'$ is indeed in \compnf. We know that $\mcT$ is in
  \compnf. Furthermore, $\mcT$ is the single subtree of the root $r$ and
  there is only one $[B]$-component in $H[B \cup C]$, namely $C$. So,
  we need to show that $V(\mcT) = C \cup (B \cap X)$.
  Since $B \cup C$ is the set of all vertices in $H[B \cup C]$,  we have
  $V(\mcT) = (B \cup C) \cap V(\mcT) = (B \cap V(\mcT)) \cup (C \cap V(\mcT))$. By the connectedness
  shown above we have $B \cap V(\mcT) = B \cap X$. Moreover, by Condition~\ref{cond:basis1} of a basis, 
  we have $C \cap V(\mcT) = C$. In total, we thus get the desired equality
  $V(\mcT) = (B \cap X) \cup C$.
\end{proof}

\begin{lemma}
  \label{lem:totalsup}
  Let $H$ be a hypergraph and $B \subseteq V(H)$. If all blocks headed
  by $B$ have a basis, then $H$ has
  a \compnf tree decomposition where $B$ is the bag of the root.
\end{lemma}
\begin{proof}
  Let $(B, C_1), \dots, (B, C_\ell)$ be all the blocks headed by $B$. 
  According to Lemma~\ref{lem:suptotd} for each $i \in [\ell]$, since the block $(B, C_i)$ has a basis, there is a \compnf TD $\mcT_i$ of $H[B \cup C_i]$ with $B$ as the bag of
  the root. We can then take the union $\mcT = \bigcup_{i \in [q]} \mcT_i$,
  which by Lemma~\ref{lem:tdunion} is precisely the required decomposition as  $B \cup \bigcup_{j=1}^l C_j = V(H)$.
\end{proof}

\begin{algorithm}[h]
\SetKwData{Left}{left}\SetKwData{This}{this}\SetKwData{Up}{up}
\SetKwFunction{Union}{Union}\SetKwFunction{FindCompress}{FindCompress}
\SetKw{Halt}{Halt}\SetKw{Reject}{Reject}\SetKw{Accept}{Accept}
\SetKw{Continue}{Continue}\SetKw{And}{and}

\SetKwData{N}{N}

\SetKwInput{Output}{output}
\SetKwInput{Input}{input}

\Input{Hypergraph $H$ and a set $\mathbf{S} \subseteq 2^{V(H)}$.}
\Output{``Accept'', if $\blockreal{\mathbf{S}} \neq \emptyset$ \linebreak
        ``Reject'', otherwise.}
\BlankLine
\SetKwProg{Fn}{Function}{}{}
\SetKwFunction{HasBase}{HasMarkedBasis}
  
  \Begin(\tcc*[f]{\bfseries Main}){
    $blocks =$  all blocks headed by any $S \in \mathbf{S}$ \;
    Mark all blocks $(B, C) \in blocks$ where $C = \emptyset$ \;
    
    \Repeat{no new blocks marked}{
      \ForEach{$(B,C) \in blocks$ where $(B,C)$ is not marked}{
        \tcc{Check if there exists a basis $X$ of $(B,C)$}
        \ForEach{$X \in \mathbf{S} \setminus \{B\}$}{

          $blocks_X = $ all blocks $(X,Y)$ with  $(X, Y)\leq (B,C)$ \;
          \If{Not all blocks in $blocks_X$ are marked}{
            \Continue\;
          }
          $V_X = X \cup \bigcup_{(X, Y) \in blocks_X} Y$\;
          \If{$C \subseteq V_X$ \And
            for each edge $e$ where $e \cap C \neq \emptyset$ also $e \subseteq V_x$}{
            Mark $(B, C)$;
          }
        }
      }

      \If{For some $S \in \mathbf{S}$, all blocks headed by $S$ are marked}{
        \KwRet \Accept \;
      }

    }
    \KwRet \Reject \;
  }
\caption{\compnf Candidate Tree Decomposition}\label{alg:bre}
\end{algorithm}

\emph{Proof of Theorem~\ref{thm:frameworkalg}}: 
We only present the decision procedure. It is clear from the soundness argument that constructing an appropriate
TD from an accepting state is trivial.
  We claim that
Algorithm~\ref{alg:bre} decides, in polynomial time, whether
$\blockreal{\mathbf{S}}\neq \emptyset$. First, the algorithm
runs in polynomial time:
First, observe that a straightforward representation of $\mathbf{S}$ as a list of lists of vertices has size $O(|\mathbf{S}| \cdot |V(H)| \cdot \log |V(H)|)$. For asserting polynomial runtime it is therefore not necessary to distinguish between the size of the representation of $\mathbf{S}$ and $|\mathbf{S}|$.
The set $blocks$ has at most $|\mathbf{S}|\cdot |V(H)|$
initial elements and computing componenets is polynomial in the size of the representation of $H$. The checks
in the innermost loop are clearly polynomial and thus, the whole algorithm requires only polynomial time.

For the soundness of the algorithm, we first observe that every marked
block $(B, C)$ in the algorithm has a \compnf TD of $H[B \cup C]$
using only bags from $\mathbf{S}$, i.e., a marked block satisfies
Condition~\ref{cond:basis3} of a basis.  This is easily verified
by structural induction in combination with
Lemma~\ref{lem:suptotd}. The construction of such a TD by the lemma
uses only the union of TDs, which does not introduce any new bags,
i.e., all bags are still elements of $\mathbf{S}$.  Soundness then
follows immediately from Lemma~\ref{lem:totalsup}. The lemma also
explicitly shows how to construct a TD from the accepting state.

Completeness will follow from Lemma~\ref{lem:support}: if there
exists a $\mcT \in \blockreal{\mathbf{S}}$, then all blocks headed by
bags of $\mcT$ are clearly contained in the set $blocks$ in the algorithm.
We proceed by induction on the height $h(u)$ of a node $u$ in $\mcT = \defTD$, 
where $h(u)$ denotes the maximum distance of $u$ from a descendant leaf node.
Let $u \in T$ and let $C$ be a component associated to a child of $u$ as in Definition~\ref{def:compnf}, or $C= \emptyset$ if $u$ is a leaf. We claim that
after $h(u)$ iterations of the repeat-until loop, the block $(B_u, C)$ will be marked by the algorithm.

For $h(u)=0$, i.e., leaf nodes, the situation is clear. The block
$(B_u, \emptyset)$ is marked before the loop.  Suppose the claim holds
for all nodes $u'$ where $h(u')<j$. We have to show that then it also holds for $h(u)=j$:
let $s$ be a child of $u$ with associated component $C_s$. By
$h(s) = j-1$ and the induction hypothesis, all blocks of $B_s$ with an
associated component of a child of $s$ have already been marked before the $j$-th iteration of the loop. Therefore, in
combination with Lemma~\ref{lem:support}, it follows that
$(B_r, C_s)$ will be marked in the $j$-th iteration of the loop.
\qed

\medskip
\noindent
\textbf{A LogCFL upper-bound.}
For the sake of simplicity, the algorithm presented here uses dynamic programming to establish a \ptime upper-bound. However, it is not difficult to see that the \compnf CTD problem lies in the class LogCFL and is therefore highly parallelizable: Consider the LogCFL algorithm for computing hypertree decompositions presented in~\cite{2002gottlob}. To guess the next separator we now, roughly speaking, guess some element of $\mathbf{S}$ instead of guessing a set of up to $k$ edges. Since $\mathbf{S}$ is an input, it is sufficient to guess an index into $\mathbf{S}$. The argument for LogCFL membership of computing hypertree decompositions then also applies to computing \compnf CTDs.

The LogCFL upper-bound also extends to our main tractability result for GHDs. Observe that the set $\mathbf{S}$ in Lemma~\ref{lem:ghdreal} can be computed in logarithmic space (for fixed $a$, $c$, and $k$). Hence, the proof of Theorem~\ref{theo:LogBMIP} also establishes LogCFL membership of computing GHDs of fixed width assuming the LogBMIP. For FHDs it remains open whether this applies. There, our approach requires the computation of the fractional cover number of sets of vertices (cf., Lemma~\ref{lem:fhwreal}). No LogCFL algorithm is known for this task.
 
\end{document}